\documentclass[12pt,reqno]{article}

\usepackage[authoryear]{natbib}

\usepackage{caption}
\usepackage{subcaption}
\usepackage{bm}

\usepackage[T1]{fontenc}
\usepackage[latin9]{inputenc}
\usepackage[letterpaper]{geometry}
\usepackage{color}
\usepackage{amsmath}
\usepackage{amsthm}
\usepackage{amssymb}
\usepackage{setspace}
\usepackage[plain]{fullpage}           %Pepe added this for wider margins
\usepackage[authoryear]{natbib}
\usepackage[unicode=true,pdfusetitle,
 bookmarks=true,bookmarksnumbered=false,bookmarksopen=false,
 breaklinks=false,pdfborder={0 0 0},pdfborderstyle={},backref=false,colorlinks=true]
 {hyperref}
\hypersetup{
 linkcolor=magenta,citecolor=blue}
 
\usepackage{graphicx}
\graphicspath{ {./graphs/} }
\usepackage{babel}

\makeatother
\usepackage{array}
\usepackage{booktabs}
\usepackage{multirow}
\usepackage{color}

\makeatletter

\usepackage{babel}

\makeatletter
\numberwithin{equation}{section}
\theoremstyle{remark}
\newtheorem{rem}{\protect\remarkname}

\theoremstyle{definition}
\newtheorem{defn}{\protect\definitionname}
\theoremstyle{plain}
\newtheorem{thm}{\protect\theoremname}
\theoremstyle{plain}
\newtheorem{prop}{\protect\propositionname}
\theoremstyle{definition}

\theoremstyle{definition}

\newtheorem{assumption}{Assumption}

\newtheorem{lem}{\protect\lemmaname}[section]

\@ifundefined{date}{}{\date{}}
\usepackage{amsmath}
\usepackage{graphicx,psfrag,epsf}
\usepackage{enumerate}
\usepackage{natbib}
\usepackage{url} % not crucial - just used below for the URL 

\usepackage{amsthm}
\usepackage{amsfonts}
\usepackage{amstext}

\usepackage{pdflscape}
\usepackage{threeparttable}

\usepackage{natbib}
\makeatother

\providecommand{\conditionname}{Condition}
\providecommand{\definitionname}{Definition}
\providecommand{\examplename}{Example}
\providecommand{\lemmaname}{Lemma}
\providecommand{\propositionname}{Proposition}
\providecommand{\remarkname}{Remark}
\providecommand{\theoremname}{Theorem}

\usepackage{verbatim}

\begin{document}
\sloppy

\title{Decision Theory for Treatment Choice Problems \\ with Partial Identification\thanks{We thank Xavier D'Haultfoeuille (the Editor) and three anonymous referees, Karun Adusumilli, Isaiah Andrews, Tim Armstrong, Larry Blume, Tim Christensen, Tommaso Denti, Kei Hirano, Hiro Kaido, Tetsuya Kaji, Charles Manski, Douglas Miller, Francesca Molinari, Alex Tetenov, Kohei Yata, seminar audiences at Berkeley, Binghamton, and Cornell, as well as participants at the Bravo/JEA/SNSF Workshop on ``Using Data to Make Decisions'' (Brown University), the 2023 Cornell-Penn State Econometrics and IO Conference, 2023 Greater NY Metropolitan Area Econometrics Colloquium and CEME (Georgetown), ASSA 2025 Annual Meeting (San Francisco) for feedback. Krishiv Shah, Maximilian Yap, Finn Ye, Sara Yoo, Naiqi Zhang, and especially Brenda Quesada Prallon provided excellent research assistance. We gratefully acknowledge financial support from the NSF under grant SES-2315600.}}

\author{Jos\'e Luis Montiel Olea\thanks{ Email: jlo67@cornell.edu}\and Chen Qiu\thanks{Email: cq62@cornell.edu}  \and J\"{o}rg Stoye\thanks{Email: stoye@cornell.edu}}
\date{June 2025}
\maketitle

\vspace{-.5cm}

\begin{abstract}
We apply classical statistical decision theory to a large class of \emph{treatment choice problems with partial identification}. We show that, in a general class of problems with Gaussian likelihood, all decision rules are admissible; it is maximin-welfare optimal to ignore all data; and, for severe enough partial identification, there are infinitely many minimax-regret optimal decision rules, all of which sometimes randomize the policy recommendation. We uniquely characterize the minimax-regret optimal rule that least frequently randomizes, and show that, in some cases, it can outperform other minimax-regret optimal rules in terms of what we term \emph{profiled regret}. We analyze the implications of our results in the aggregation of experimental estimates for policy adoption, extrapolation of Local Average Treatment Effects, and policy making in the presence of omitted variable bias.
\medskip


\end{abstract}

\newpage{}

\onehalfspacing

\section{Introduction}
A policy maker must decide between implementing a new policy or preserving the status quo. Her data provide information about the potential benefits of these two options. Unfortunately, these data only \emph{partially identify} payoff-relevant parameters and may therefore not reveal, even in large samples, the correct course of action. Such \emph{treatment choice problems with partial identification} have recently received growing interest; for example, see \cite{d2021policy}, \cite{ishihara2021}, \cite{yata2021}, \cite*{christensen2022optimal}, \cite{kido2022distributionally} or \cite{manski2022identification}. Several interesting problems that arise in empirical research can be recast  using this framework. A non-exhaustive list includes extrapolation of experimental estimates for policy adoption \citep{ishihara2021,menzel2023transfer}, policy-making with quasi-experimental data in the presence of omitted variable bias \citep*{diegert2022assessing}, and extrapolation of Local Average Treatment Effects \citep*{mogstad2018using,mogstad2018identification}.

In this paper, we analyze such problems in terms of \citeauthor{Wald50}'s (\citeyear{Wald50}) Statistical Decision Theory. We do so in a finite-sample framework characterized by a Gaussian likelihood and partial identification.

{\scshape Main Results:} Three optimality criteria are routinely used to endorse or discard decision rules:  \emph{admissibility}, \emph{maximin welfare}, and \emph{minimax regret}. 

\emph{Admissibility.} A decision rule is (welfare-)admissible if one cannot improve its expected welfare uniformly over the parameter space. This is usually considered a weak requirement for a decision rule to be considered ``good''. Our first result shows that, whenever problems in our setting exhibit partial identification, \emph{every} decision rule, no matter how exotic, is admissible (Theorem \ref{thm:admin}). As we discuss later, this result stands in stark contrast with applications of the admissibility criterion to point-identified treatment choice problems, in which admissibility meaningfully refines the class of
decision rules (by, for example, discarding rules that randomize the policy recommendation). We also show that our result is not tied to the choice of Gaussian likelihood, but instead to the \emph{bounded completeness} of the statistical model for the available data (Theorem \ref{thm:general}).

\emph{Maximin Welfare.} A decision rule is maximin(-welfare) optimal if it attains the highest worst-case expected welfare. Echoing critiques from \cite{Savage51} to \cite{manski2004statistical}, our second result shows that maximin decision rules will preserve the status quo regardless of the data (Theorem \ref{thm:maximin}).

\emph{Minimax Regret.} A decision rule is minimax-regret (MMR) optimal if it attains the lowest worst-case expected regret, where an action's regret is its welfare loss relative to the action that would be optimal if payoff-relevant parameters were known.  In some point-identified treatment choice problems, the MMR rule is both essentially unique and nonrandomized \citep{canner,stoye2009minimax,tetenov2012statistical}. We show this may not be true under partial identification. To make this point, we specialize our framework to the class of treatment choice problems recently studied by \cite{yata2021}.  Our third result shows that for cases where the identified set for payoff relevant parameters is large enough, there are  \emph{infinitely many} MMR optimal decision rules, \emph{all of which} randomize the policy recommendation for some or all data realizations (Theorem \ref{thm:main.regret.1}). This presents an important challenge for the application of MMR, at least if one hopes for the resulting recommendation to be unique. Moreover, as we explain later by means of an example, different MMR optimal rules can lead to meaningfully different policy choices for the same data. 

\emph{Least Randomizing MMR rule:} Finally, we refine the set of MMR optimal rules by identifying the \emph{least randomizing} (in a sense we make precise) MMR optimal rule. Our main motivation is the following trade-off. Recall that, for a wide range of parameter values, any MMR optimal rule must randomize for some data realizations. At the same time, despite wide adoption of randomized treatment allocations in economics and the social sciences for the purpose of experimentation, it might be difficult in many policy applications to randomize one's policy. Thus, we look for a rule that recommends such randomization as infrequently as possible. We explicitly characterize (in Theorem \ref{thm:main.regret.2}) an essentially unique least randomizing rule for the problems considered by \cite{yata2021}. 

We analyze the regret of the least randomizing MMR rule after profiling out some parameters of the risk function. We specifically analyze \emph{profiled regret}, by considering a parameter of interest and reporting worst-case expected regret at each of these parameter values (in a sense we make precise). We show, in the context of our running example, that our least-randomizing rule can profiled-regret dominate the MMR rule suggested by \cite{stoye2012minimax} and recently extended by \cite{yata2021} for a general class of treatment choice problems with partial identification (Proposition \ref{prop:linear.vs.RT}). More generally, we show that the use of profiled regret renders the rule that uniformly randomizes the policy recommendation inadmissible (Proposition \ref{prop:profiled.regret.dominance}), where the profiled regret function we consider reports the worst-case regret for a fixed value of the problem's point-identified parameters.

We also discuss the extent to which our least-randomizing rule can be obtained by explicitly penalizing randomized policy recommendations in the policy maker's welfare function. We show that, under some conditions, the least-randomizing rule is minimax regret optimal (among a suitably defined class of decision rules) for a penalty function that penalizes all randomized assignments equally (Proposition \ref{prop:aversion.fraction}).

{\scshape Applications:} We illustrate the practical implications of our results for three problems that recently arose in applied work. 

First, we analyze in detail a running example based on \citeauthor{ishihara2021}'s (\citeyear{ishihara2021}) ``evidence aggregation'' framework. Here, a policy maker is interested in implementing a new policy in country $i=0$. She has access to estimates of the effect of the same policy for other countries $i=1,...,n$ and attempts to extrapolate results using baseline covariates. We give explicit MMR treatment rules for this example. An interesting finding is that, when the identified set is large enough, the least randomizing rule can be related to the estimated bounds on the treatment effect of interest and randomizes only (though not always) if these bounds contain both positive and negative values. This illustrates how an estimator of the identified set can be used for optimal decision making.  

Second, we study extrapolation of Local Average Treatment Effects \citep{mogstad2018using,mogstad2018identification} with binary instrument and no covariates. Here, the payoff-relevant parameter is a ``policy-relevant treatment effect'' \citep{heckman2005structural} that corresponds to expanding the complier subpopulation. We show that Theorem \ref{thm:admin} applies in this example, so that all decision rules are admissible. In particular, a decision rule that implements the policy for large values of the usual instrumental-variables estimator is not dominated. 

Third, we consider a policy maker who uses quasi-experimental data to decide on a new policy. She is willing to assume a constant treatment effect model and unconfoundedness given a set of covariates $(X,W)$; however, only $X$ is observable and $W$ is not. In this setting, \cite{diegert2022assessing} argue that researchers may be interested in how much selection on unobservables is required to overturn findings that are based on a feasible linear regression. The least randomizing MMR rule can inform a complementary, decision-theoretic breakdown point analysis: For a given estimated effect of the policy, what is the largest effect of unobserved confounding under which it is still optimal to adopt the seemingly better policy without any hedging? We show that this breakdown point tolerates more confounding than the one of \cite{diegert2022assessing}.

\textbf{Related literature}. The econometric literature on treatment choice has grown rapidly since \cite{manski2004statistical} and \cite{Dehejia2005}. When welfare is partially identified, \cite{Manski2000,manski2005social,manski2007identification} and \citet{Stoye07} provide optimal treatment rules assuming the true distribution of the data is known. 
\cite{stoye2012minimax,stoye2012new}, \cite{yata2021}, and \cite{ishihara2021} focus on finite sample MMR optimal rules, and \citet*{twopointprior} on multiple prior MMR rules, in such settings. For different settings with point-identified welfare, finite- and large-sample results on optimal treatment choice rules were derived by \citet{canner}, \citet{ChenGug}, \citet{HiranoPorter2009,HiranoPorter2020}, \citet*{kitagawa2022treatment}, \citet{schlag2006eleven}, \citet{stoye2009minimax}, and \citet{tetenov2012statistical}. \cite{christensen2022optimal} extend \citeauthor{HiranoPorter2009}'s (\citeyear{HiranoPorter2009}) limit experiment framework to a class of partially identified settings; see on this also \cite{kido2023locally}. Treatment choice is furthermore related to a large literature on optimal policy learning that contains many results for point identified \citep{BhattacharyaDupas2012,kitagawa2018should,KT19,MT17,KW20,AW20,KST21,ida2022choosing} as well as partially identified \citep{kallus2018confounding,ben2021safe,ben2022policy,d2021policy, christensen2022optimal,adjaho2022externally,kido2022distributionally,lei2023policy} treatment choice that may condition on covariates. Bayesian aspects of treatment choice are discussed in \cite{chamberlain2012}. 

The remainder of this paper is organized as follows: Section \ref{sec:framework} introduces the formal framework and the running example. Section \ref{sec:challenges} is devoted to the aforementioned main results on admissibility, maximin wellfare, minimax regret and the least randomizing MMR rule. The applications are presented in Section \ref{sec:examples}. Section \ref{sec:conclude} concludes. Appendix \ref{sec:app.main} contains proofs of our main results and Appendix \ref{sec:profile.regret} discusses  the notion of profiled regret. Additional proofs and results can be found  in Online Appendix \ref{sec:technical}.

\section{Notation and Framework}\label{sec:framework}

Statistical decision theory calls for three ingredients: the menu of actions available, their consequences as a function of an unknown state of the world, and a statistical model of how the data distribution depends on that state. We now present these elements and lay out an example that will be used to illustrate objects, terms, and results throughout. 

The policy maker can choose an \emph{action} $a\in[0,1]$, which we interpret as the proportion of a population that will be randomly assigned to the new policy. Thus, $a=1$ means that everyone is exposed to the new policy and $a=0$ means that the status quo is preserved. Under this interpretation, $a=.5$ means that 50\% of the population will be exposed at random to the new policy; however, the formal development equally applies to either individual or population-level randomization. Our interpretation abstracts from integer issues arising with small populations.

The payoff for the policy maker when taking action $a\in[0,1]$ is captured by a welfare function
\begin{equation}\label{eq:welfare}
W(a,\theta):=aW(1,\theta)+(1-a)W(0,\theta),    
\end{equation}
where $\theta\in\Theta$ is an unknown state of the world or \emph{parameter} (possibly of infinite dimension) and  $W(1,\cdotp):\Theta\rightarrow\mathbb{R}$ and $W(0,\cdotp):\Theta\rightarrow\mathbb{R}$ are known functions. Thus, welfare is linear in the action, a common assumption in the literature.\footnote{In particular, this applies if $W(\cdot,\theta)$ is an expectation and, for the case where randomization is interpreted as fractional assignment, there are no spillover effects or externalities. These assumptions are the default in the literature. An exception is \cite{manski2007admissible}, who consider the welfare of an action to be a concave transformation of $W(\cdot,\theta)$.} The form of (\ref{eq:welfare}) also implies that if $\theta$ were known to the policy maker, her optimal choice of action would simply be
\begin{equation}\label{eq:oracle}
\mathbf{1}\left\{ U(\theta)\geq0\right\}, \quad \textrm{ where } \quad U(\theta):=W(1,\theta)-W(0,\theta). 
\end{equation}
Following \citet{HiranoPorter2009}, we refer to $U(\theta)$ as the \emph{welfare contrast} at $\theta$. Thus, the policy maker's optimal action in \eqref{eq:oracle} is to expose the whole population to the new policy if the welfare contrast is nonnegative and to preserve the status quo otherwise. 

The policy maker observes a realization of $Y \in \mathbb{R}^n$ distributed as
\begin{equation} \label{eq:normal_model}
Y\sim N(m(\theta),\Sigma).
\end{equation}
Here, the function $m(\cdotp):\Theta\rightarrow\mathbb{R}^n$ and the positive definite covariance matrix $\Sigma$ are known. However, $m(\cdot)$ need not be injective: $m(\theta) = m(\theta')$ does \emph{not} imply $\theta=\theta'$. As a result, even perfectly identifying $m(\theta)$ (from infinite data) may not identify the optimal action.

In economics applications, the normality assumption in \eqref{eq:normal_model} is unlikely to hold exactly; however, the data can often be summarized by statistics that are asymptotically normal and whose asymptotic variances can be estimated. Treating the limiting model as a finite-sample statistical model then eases exposition and allows us to focus on the core features of the policy problem. 
Working directly with such a limiting model is common in applications of statistical decision theory to econometrics; see  \cite{muller2011efficient} and the references therein for theoretical support and applications in the context of testing problems and \citet{ishihara2021}, \citet{stoye2012minimax}, or \citet{tetenov2012statistical} for precedents in closely related work. 

We finally define a \emph{decision rule}, $d:\mathbb{R}^{n}\rightarrow[0,1]$, as (measurable) mapping from the data $Y$ to the unit interval.
We let $\mathcal{D}_n$ denote the set of all decision rules. We call $d\in \mathcal{D}_n$ \emph{non-randomized} if $d(y)\in\{0,1\}$ for (Lebesgue) almost every $y\in\mathbb{R}^n$ and \emph{randomized} otherwise.

{\scshape Running Example:} Our running example is a special case of \citeauthor{ishihara2021}'s (\citeyear{ishihara2021}; see also \citet{manski2020towards}) ``evidence aggregation'' framework. A policy maker is interested in implementing a new policy in country $i=0$ and observes estimates of the policy's effect for countries $i=1,...,n$. Let $Y=(Y_1,...,Y_n)^\top \in \mathbb{R}^n$ denote such experimental estimates and let $(x_0,\ldots,x_n)$ be nonrandom, $d$-dimensional baseline covariates. The policy maker is willing to extrapolate from her data by assuming that $W(1,\theta)=\theta(x_0)$, $W(0,\theta)=0$,  $U(\theta)=\theta(x_0)$ and that \begin{equation*}
Y = 
\begin{pmatrix}
Y_1 \\
\vdots \\
Y_n
\end{pmatrix}
\sim N(m(\theta), \Sigma), \quad m(\theta) = 
\begin{pmatrix}
\theta(x_1) \\
\vdots \\
\theta(x_n)
\end{pmatrix}, 
\quad
\Sigma =  
\operatorname{diag}(\sigma_1^2,\ldots,\sigma_n^2),
\end{equation*} 
where $\theta: \mathbb{R}^d \rightarrow \mathbb{R}$ is an unknown Lipschitz function with known constant $C$. For simplicity, in the example we write $\mu_i$ for $\theta(x_i)$ henceforth, thus $Y_i \sim N(\mu_i,\sigma_i^2)$. We also assume w.l.o.g. that countries are arranged in nondecreasing order of $\left\Vert x_i-x_0 \right\Vert$. While our analysis extends to $x_1=x_0$, we focus on the case of $x_1 \neq x_0$, so that the sign of $\mu_0$ is not necessarily identified. We assume that $(x_1,\ldots,x_n)$ are distinct. Even if this were not the case in raw data, one would presumably want to induce it (by adding fixed effects, whose size can be bounded) because the Lipschitz constraint would otherwise imply that countries with same $x$ exhibit no heterogeneity whatsoever. Finally, we assume $\Vert x_1-x_0\Vert<\Vert x_2-x_0\Vert$, i.e., the nearest neighbor  of country $0$ is unique.\footnote{If two  or more countries  are  nearest neighbors, their signals can be merged into one more precise signal.}\hfill $\square$

\section{Main Results} \label{sec:challenges}

In statistical decision theory, three criteria are commonly used to recommend decision rules: \emph{admissibility}, \emph{maximin welfare}, and \emph{minimax regret}.\footnote{See \citet{stoye2012new} and references therein for theoretical trade-offs between these criteria.} 
In this section, we show that the application of these criteria to our setting presents nontrivial challenges.

\subsection{Everything is Admissible} 

Let $\mathbb{E}_{m(\theta)}[\cdot]$ denote expectation with respect to $Y\sim N(m(\theta),\Sigma)$. Recall the following definition:
\begin{defn}
A rule $d\in\mathcal{D}_{n}$ is \emph{(welfare-)admissible} if there does not exist $d^{\prime}\in\mathcal{D}_{n}$ such that 
\[ \mathbb{E}_{m(\theta)}[W(d^{\prime}(Y),\theta)] \geq \mathbb{E}_{m(\theta)}[W(d(Y),\theta)], \quad \forall \theta \in \Theta, \]
with strict inequality for some $\theta\in\Theta$. 
\end{defn}
Thus, a rule is admissible if it is not dominated (in the usual sense of weak dominance everywhere and strict dominance somewhere). This is generally considered a minimal but compelling requirement for a decision rule to be ``good'' and goes back at least to \cite{Wald50}. Admissibility can be used to recommend classes of rules whose payoff cannot be uniformly improved and/or whose members improve uniformly on non-members, as in complete class theorems \citep{karlin1956theory,manski2007admissible}; conversely, one may use it to caution against particular (classes of) decision rules, as was recently done by \cite{andrews2022gmm}. 

Our first result shows that, under mild assumptions, admissibility cannot serve either purpose in our problem. This is because any decision rule is admissible. To formalize this, let
\begin{equation}\label{eq:M}
M:=\left\{ \mu\in\mathbb{R}^{n}:m(\theta)=\mu,\theta\in\Theta\right\} 
\end{equation}
collect all means that can be generated as $\theta$ ranges over $\Theta$. We refer to elements $\mu \in M$ as $\emph{reduced-form}$ parameters because they are identified in the statistical model \eqref{eq:normal_model}.\footnote{The notation is consistent with our running example, in which the observable moments are $(\mu_1,\ldots,\mu_n)$.} Define the \emph{identified set} for the welfare contrast as function of $\mu$ as 
\begin{equation}\label{eq:I_mu}
I(\mu):=\left\{ u \in\mathbb{R}: U(\theta)=u,m(\theta)=\mu, \theta \in \Theta \right\}
\end{equation}
and the corresponding upper and lower bounds as
\begin{equation}\label{eq:I_bound}
\overline{I}(\mu)  :=\sup I(\mu) ,\quad\underline{I}(\mu):=\inf I(\mu).    
\end{equation}
When we refer to models as ``partially identified,'' we henceforth mean that partial identification obtains on an open set in parameter space (i.e., not almost nowhere).\footnote{For simplicity of exposition, Definition \ref{asm:1} is stated in a way that forces $M$ to have a non-empty interior. This would, for example, exclude equality constraints. For the purpose of Theorem \ref{thm:main.regret.1} below, we can weaken the assumption to allow such cases as long as $\mathcal{S}$ is rich within $M$.} 

\begin{defn} [Nontrivial partial identification]\label{asm:1}
The treatment choice problem with payoff function \eqref{eq:welfare} and statistical model \eqref{eq:normal_model} exhibits \emph{nontrivial partial identification} if there exists an open set $\mathcal{S}\subseteq M \subseteq \mathbb{R}^{n}$ such that
\begin{align*}
\underline{I}(\mu) & <0<\overline{I}(\mu)\text{, for all }\mu\in\mathcal{S}.
\end{align*}
\end{defn}

\noindent {\scshape Running Example---Continued:} 
The identified set for the welfare contrast $\mu_0$ is 
\begin{equation*}
I(\mu) = \{ u \in \mathbb{R} :  \left\vert \mu_i-u \right\vert \leq C\left\Vert x_i-x_0\right\Vert, ~~ i=1,\ldots,n \}.
\end{equation*}
Its extrema can be written as intersection bounds:
\[ \underline{I}(\mu) = \max_{i=1,\ldots,n} \{ \mu_i - C\left\Vert x_i-x_0 \right\Vert  \},  \quad  \overline{I}(\mu) = \min_{i=1,\ldots,n} \left \{ \mu_i + C \left\Vert x_i-x_0 \right\Vert \right \}.  \]
For $\mu$ sufficiently close to the zero vector, we therefore have nontrivial partial identification.
\hfill $\square$

We are now ready to state the first main result.

\begin{thm}\label{thm:admin}
If a treatment choice problem with payoff function \eqref{eq:welfare} and statistical model \eqref{eq:normal_model} exhibits nontrivial partial identification in the sense of Definition \ref{asm:1},  then every decision rule $d\in\mathcal{D}_{n}$
is (welfare-)admissible.
\end{thm}

\begin{proof}
See Appendix \ref{sec:app.a.1}. 
\end{proof}

For a proof sketch, suppose by contradiction that some rule $d$ is inadmissible. Then some other rule $d'$ dominates it. This $d'$  must perform weakly better at every $\theta \in m^{-1}(\mathcal{S})$, where $\mathcal{S}$ is the set that appears in Definition \ref{asm:1}. Because of nontrivial partial identification, all  $\theta \in m^{-1}(\mathcal{S})$ are compatible with positive and negative welfare contrast $U(\theta)$. This implies that 
\[ \mathbb{E}_{m(\theta)}[d(Y)] = \mathbb{E}_{m(\theta)}[d'(Y)] \quad \textrm{for each } \theta \in m^{-1}(\mathcal{S}).\]
By i) completeness of the Gaussian statistical model in (\ref{eq:normal_model}) and ii) mutual absolute continuity of the Gaussian and Lebesgue measures in $\mathbb{R}^{n}$, we then have $d(\cdot) = d'(\cdot)$ (Lebesgue) almost everywhere in $\mathbb{R}^{n}$, a contradiction.\footnote{A family  $\mathcal{P}$ of distributions $P$ is complete if $\mathbb{E}_{P}[f(X)]=0$ for all $P\in\mathcal{P}$ implies $f(x)=0$ $P$-almost everywhere, for every $P \in \mathcal{P}$. See, for example, 
 \citet[][p.115]{lehmann05testing}.  } 
  
While the statement of Theorem \ref{thm:admin} makes reference to the Gaussian model in \eqref{eq:normal_model}, the above proof sketch only uses this model's \emph{bounded completeness} (and the mutual absolute continuity of the Gaussian and Lebesgue measures in $\mathbb{R}^n$).  
Indeed, Theorem \ref{thm:general} in Appendix \ref{sec:general}, establishes a stronger version of Theorem \ref{thm:admin} that applies to boundedly complete statistical models (in a sense we make precise).\footnote{We would like to thank Tim Christensen and the Editor for pointing this out.} It is known that many exponential family distributions \citep[Theorem 4.3.1, p.116]{lehmann2006theory} as well as some location families \citep[][Theorem 2.1]{mattner1993some} are boundedly complete. However, our proof does not apply to distributions that are not boundedly complete even if they are close to Normal; for example, if $Y=Z+\epsilon$, where $Z\sim N(m(\theta),\Sigma)$ and $\epsilon$ is independent of $Z$ and contains i.i.d components with characteristic functions containing zeros such as the uniform distributions in $[-\eta,\eta]$ for any $\eta>0$ (c.f. \citealt[][Theorem 2.1]{mattner1993some} and references therein). While we do not know whether bounded completeness is strictly necessary for this result, it cannot just be dropped. To see this, consider the preceding counterexample with a uniform error term, in which we further impose $n=1$, $\Sigma =0$ and  $\eta=1$. Furthermore, suppose  $\overline{I}(\mu)=\mu+1$ and  $\underline{I}(\mu)=\mu-1$. Then the coin-flip rule $d_{\text{coin-flip}}(Y)=1/2$ would be dominated by the rule that chooses $0$ if $Y<-2$, chooses $1$ if $Y>2$, and $1/2$ otherwise.

\begin{rem} \label{rem:Remark_2}
Theorem \ref{thm:admin} shows that the notion of admissibility does not have any refinement power on the class of decision rules considered by the policy maker. We think this result is quite surprising, given that in treatment choice problems with point identification, admissibility does meaningfully refine the class of decision rules. For example, let $n=1$, $\Theta = \mathbb{R}$, $m(\theta) = W(1,\theta) = \theta$, and $W(0,\theta) = 0$. In this case, the policy maker observes a noisy signal, $Y \sim N(\theta, \sigma^2)$, of the payoff-relevant, point-identified parameter $\theta \in \mathbb{R}$. 
By \citeauthor{karlin1956theory}'s (\citeyear[][Theorem 2]{karlin1956theory}) classic result, any decision rule that is not a threshold rule (i.e., is not of form $\mathbf{1}\{ Y > c \}$ for some fixed $c \in \mathbb{R} \cup \{-\infty,\infty\}$) is dominated. In fact, the class of all threshold rules is \emph{complete} (in the sense that any non-threshold rule is dominated by a threshold rule). Therefore, any decision rule that is not a threshold rule can be dismissed by appealing to the notion of admissibility alone. 
We think it is quite surprising that introducing partial identification renders \emph{all} decision rules admissible. Indeed, Theorem 1 implies that,  no rule, regardless how eccentric it may appear, is (welfare-)dominated. This makes it much more challenging to recommend a rule or a class of rules, at least without commitment to a more specific decision-theoretic optimality criterion (such as maximin welfare or minimax regret).

\end{rem}

Theorem \ref{thm:admin} also admits a more optimistic interpretation.  The positive interpretation is that the procedures suggested in the related literature will all perform well relative to one another in some parts of parameter space. This is an important observation because some of these suggestions contain novel or nonstandard components. For example, \cite{ishihara2021} place ex-ante restrictions on the class of decision rules, while \cite{christensen2022optimal} transform the original loss function by profiling out partially identified parameters. By Theorem \ref{thm:admin}, all these approaches are at least admissible. This was arguably obvious in the former case (since for any linear threshold rule, it is easy to find a prior that it uniquely best responds to) but certainly not the latter one.

\subsection{Maximin Welfare is Ultra-Pessimistic}
We next analyze the \emph{maximin welfare} criterion. Our main result echoes earlier findings by \cite{Savage51} and \cite{manski2004statistical}: Maximin typically leads to ``no-data rules'' that preserve the status quo.   
\begin{defn}
A rule $d_{\text{maximin}}\in\mathcal{D}_{n}$ is maximin optimal if 
\[\inf_{\theta\in\Theta}\mathbb{E}_{m(\theta)}[W(d_{\text{maximin}}(Y),\theta)]=\sup_{d\in\mathcal{D}_{n}}\inf_{\theta\in\Theta}\mathbb{E}_{m(\theta)}[W(d(Y),\theta)].
\]    
\end{defn}

\begin{thm}\label{thm:maximin}
Suppose that there exists $\theta \in \Theta$ such that $U(\theta) \leq 0$. If
\begin{equation}\label{eq:maximin}
\inf_{\theta\in\Theta}W(0,\theta) = \inf_{\theta\in\Theta:U(\theta)\leq0}W(0,\theta), 
\end{equation}
then the no-data rule $d_{\text{no-data}}(y):=0$ is maximin optimal. The maximin value is \[
\inf_{\theta\in\Theta}W(0,\theta).\]
\end{thm}

\begin{proof}
See Appendix \ref{sec:app.a.2}.
\end{proof}

\noindent This result can be seen as follows. When $U(\theta)\leq 0$, it is optimal to preserve the status quo; substituting in for this response, we find that $\inf_{\theta\in\Theta}\mathbb{E}_{m(\theta)}[W(d(Y),\theta)]\leq \inf_{\theta\in\Theta,U(\theta)\leq0}W(0,\theta)$ for any rule $d\in\mathcal{D}_n$. Under condition (\ref{eq:maximin}), this upper bound is attained by $d_{\text{no-data}}$.

{\scshape Running Example---Continued:} Theorem \ref{thm:maximin} applies to the running example. In particular, the example's maximin welfare equals $0$ and is achieved by never assigning the new policy. \hfill $\square$

A similar result was shown by \cite{manski2004statistical} for testing an innovation with point-identified welfare contrast (a result that we generalize\footnote{In \citeauthor{manski2004statistical}'s (\citeyear{manski2004statistical}) example, $W(0,\theta)$ does not depend on $\theta$, so that condition (\ref{eq:maximin}) is trivially satisfied.}), and the concern can be traced back at least to \cite{Savage51}. There is a discussion of whether such ``ultrapessimism'' occurs, in a technical sense, more generically with maximin utility versus minimax regret \citep{ParmigianiTD,SadlerTD}. However, a string of more optimistic results regarding MMR \citep{canner,stoye2009minimax,stoye2012new,tetenov2012statistical,yata2021} suggests that, with state spaces that describe real-world decision problems, the concern is more salient for maximin.\footnote{Given that we next elaborate on multiplicity of MMR rules, it is of interest to also discuss the potential multiplicity of maximin rules: if $W(0,\theta) $ is treated as known (e.g., equal to 0), the maximin rule may be unique. However, in settings where both $W(1,\theta)$ and $W(0,\theta)$ are  unknown and can be arbitrarily ``bad'', the typical result is that all decision rules are maximin. }

\subsection{Minimax Regret Admits Many Solutions} \label{sec:minimax.regret.result}

In view of Theorem 1 and Theorem 2, it seems natural to consider the \emph{minimax regret} (MMR) optimality criterion. It is known that some point-identified treatment choice problems admit nonrandomized, essentially unique MMR optimal rules. In contrast, our results below will show that uniqueness of MMR optimal rules should not be expected to hold generally with partial identification. To make this point, we consider a class of problems where finite-sample minimax optimal results are available. Within this class, we show that, if the identified set for the welfare contrast is sufficiently large relative to sampling error, then i) we can find infinitely many MMR optimal decision rules, all of which depend on data through an optimal linear index and ii) under weak additional conditions, any MMR rule that depends on the data through the optimal linear index must be randomized. 

Our numerical and analytic findings below (see Figure \ref{figure:NewRule} and Proposition \ref{prop:linear.vs.RT}) demonstrate that, beyond having the same worst-case regret, two MMR optimal rules might be qualitatively and quantitatively very different. This presents a practical challenge for a policy maker as, given the same data realizations, two MMR optimal rules may recommend different policy actions.

The expected regret of a decision rule $d\in\mathcal{D}_n$ in state $\theta$ is its expected welfare loss compared to the oracle rule:
\begin{eqnarray}\label{eq:expected.regret}
R(d,\theta) &:=& U(\theta)\left\{ \mathbf{1}\{U(\theta)\geq0\}-\mathbb{E}_{m(\theta)}[d(Y)]\right\}.
\end{eqnarray}
\begin{defn}
A rule $d^*\in \mathcal{D}_{n}$ is minimax regret (MMR) optimal if
\begin{equation}\label{eq:MMR}
\sup_{\theta \in \Theta}R(d^*,\theta)=\inf_{d\in\mathcal{D}_{n}}\sup_{\theta \in \Theta}R(d,\theta).    
\end{equation}
\end{defn}

\noindent Solving minimax regret problems is often hard, both analytically and algorithmically.  Algorithms exist for certain cases \citep{yu1995min, chamberlain2000econometric, fernandez2024epsilon, guggenberger2025numerical}, but it is known that obtaining the minimax solution of a decision problem---and sometimes even deciding whether a minimax solution exists---is NP hard in general; see \cite{daskalakis2021complexity}. In important, recent work, \cite{yata2021} characterizes MMR optimal rules for a large class of binary action problems. He imposes the following restrictions on the parameter space and welfare function. 

\begin{assumption}  \label{asm:yata.1} 
\begin{itemize}
\item[(i)]  $\Theta$ is convex, centrosymmetric (i.e., $\theta\in\Theta$ implies $-\theta\in\Theta$) and nonempty.
\item[(ii)] $m(\cdot)$ and $U(\cdot)$ are linear.\end{itemize}
\end{assumption}

\noindent {\scshape Running Example---Continued:} Our example satisfies Assumption \ref{asm:yata.1}. In particular, the space of $C$-Lipschitz functions, $\theta:\mathbb{R}^{d} \rightarrow \mathbb{R}$ is convex as well as centrosymmetric, and the functions $U(\cdot)$ and $m(\cdot)$ are linear in $\theta$ as they simply report the values of $\theta$ at points $(x_0, \ldots, x_n)$. \hfill $\square$

\noindent Under Assumption \ref{asm:yata.1}, \cite{yata2021} shows existence of an MMR rule that depends on the data only through $(w^*)^{\top} Y$, where the unit vector $w^*$ can be approximated by solving a sequence of tractable optimization problems. When the identified set for the welfare contrast at $\mu= \mathbf{0} := 0_{n\times1}$ is large enough, \citeauthor{yata2021}'s (\citeyear{yata2021}) MMR rule can be expressed as 
\begin{equation}\label{eq:yata.rule}
d^{*}_{\text{RT}}(Y):=d^{*}_{\text{RT}}((w^*)^{\top}Y):=\Phi\bigl((w^*)^{\top}Y/\tilde{\sigma}) \end{equation}
for some uniquely characterized $\tilde{\sigma} > 0$.\footnote{For readability, we slightly abuse notation and,  when  a decision rule $d$ depends on data $Y$ only through a simple feature like $w^{\top}Y$ or $Y_1$, we write $d(Y)$ as $d(w^{\top}Y)$ or $d(Y_1)$. } Moreover, it then has two important algebraic properties:
\begin{equation}\label{eq:yata.property.1}
\mathbb{E}_{m(\theta)}[d_{\text{RT}}^{*}(Y)]=1/2\quad\textrm{ when }m(\theta)=\mathbf{0}
\end{equation}
and
\begin{equation}\label{eq:yata.property.2}
\sup_{\theta\in\Theta}R(d^*_{\text{RT}},\theta)=\sup_{\theta\in\Theta,m(\theta)=\mathbf{0}}R(d^*_{\text{RT}},\theta). 
\end{equation}

In words, these features are as follows: First, if the mean function $m(\cdot)$ equals zero, expected exposure to the new policy is $1/2$. Second, worst-case regret occurs precisely at this point.\footnote{\cite{yata2021} also provide sufficient conditions under which these conditions hold true. Our proof is constructive, and our results below apply to a large class of models with these properties and shall not be considered as a specific example. } The latter is due to careful calibration of the MMR decision rule, and one might conjecture that it renders this rule unique. However, the following result establishes the contrary.   
\begin{thm}\label{thm:main.regret.1}
Consider a treatment choice problem with payoff function \eqref{eq:welfare} and statistical model \eqref{eq:normal_model} that exhibits nontrivial partial identification in the sense of Definition \ref{asm:1}. Suppose that Assumption \ref{asm:yata.1} holds and that there is a MMR optimal rule $d^*$ that depends on the data only through $(w^*)^\top Y$ and that satisfies \eqref{eq:yata.property.1} and \eqref{eq:yata.property.2}. If there exists $\mu\in M$ such that $\overline{I}(\mu)>\overline{I}(\mathbf{0})$ and $\overline{I}(\mathbf{0})$ is large enough, then
\begin{itemize}
    \item[(i)] There are infinitely many MMR optimal rules.

\item[(ii)] Any MMR rule that depends on the data only through $(w^*)^{\top}Y$ (and is weakly increasing in this argument) must randomize for some data realizations. 

\item[(iii)] If $\overline{I}(\mu)$ is differentiable at $\mu=\mathbf{0}$, then no linear threshold rule, i.e., no rule of form $\mathbf{1}\{w^{\top}Y\geq c\}$ for some $w\in\mathbb{R}^{n}$ and $c\in\mathbb{R}\cup\{-\infty,\infty\}$, is MMR optimal.
\end{itemize}
\end{thm}

\begin{proof}
See Appendix \ref{sec:app.a.3}. 
\end{proof}

To be clear, this finding applies if the problem is sufficiently far from point identification, with an exact condition given in the proof.\footnote{This condition will always be met if $\Sigma$ is sufficiently small, holding other parameter values fixed. Heuristically, this means that our multiplicity result is always relevant as sample size becomes large. However, to what extent our result is a useful characterization in a limit experiment is beyond the scope of the current paper and an interesting question that we leave for future research.  } Close to point identification and under mild additional assumptions, \cite{yata2021} shows MMR optimality of a linear threshold rule.

Part (i) of Theorem \ref{thm:main.regret.1} is established constructively: We show that, whenever $d^*_{\text{RT}}$ is MMR optimal, then so is the \emph{piecewise linear} rule
\begin{equation}\label{eq:linear.rule}
d^*_{\text{linear}}((w^*)^{\top}Y):=\begin{cases}
0, & (w^*)^{\top}Y<-\rho^{*},\\
\frac{(w^*)^{\top}Y+\rho^{*}}{2\rho^{*}}, & -\rho^{*}\leq(w^*)^{\top}Y\leq\rho^{*},\\
1, & (w^*)^{\top}Y>\rho^{*},
\end{cases}
\end{equation}
where $\rho^*>0$ is characterized in Appendix \ref{sec:app.a.3}, Equation \eqref{eq:rho.star.main.proof}. This implies existence of infinitely many MMR rules because the set of such rules is closed under convex combination.

Next, if the identified set is large enough for given $\Sigma$ (or as $\Sigma$ vanishes for given identified set), all of the above rules randomize for some data realizations. A natural question to ask is whether this feature is shared by all MMR rules. Parts (ii) and (iii) give qualified affirmative answers: If we focus on decision rules that increase in $(w^*)^{\top}Y$ and if $\overline{I}(\mathbf{0})$ is large enough, then randomization is necessary for MMR optimality; if bounds are furthermore differentiable in reduced-form parameters at $\mathbf{0}$, randomization is necessary for any MMR rule that depends on a linear index of the data.\footnote{The class of nonrandomized but otherwise unrestricted rules is not interestingly different from the class of all rules due to the possibility of purifying randomized rules. See Remark \ref{rem:lexicographic} for discussion and references. Similarly, the differentiability condition is needed to preclude that some component of $Y$ can effectively be used as randomization device.} In particular, the threshold rule
\begin{equation}\label{eq:d.erm}
d_0^*((w^*)^\top Y )=\mathbf{1}\{(w^*)^\top Y \geq 0\}
\end{equation}
is not MMR optimal.  

\noindent{\scshape Running Example---Continued:} Theorem \ref{thm:main.regret.1} applies to our running example. We next improve on this observation by providing explicit MMR optimal rules for the example. Broadly speaking, these are characterized by a weighting of studentized signals that resembles a triangular kernel for small enough $C$ and turns into a nearest neighbor kernel as $C$ becomes large.\hfill $\square$

\begin{prop}\label{thm:running.example}
In the running example, the following statements hold true.
\begin{itemize}
\item[(i)] If
\begin{equation}\label{eq:small.id.condition}
C \left\Vert x_1 - x_0 \right\Vert < \sqrt{\pi / 2} \cdot \sigma_1, 
\end{equation}
then the following decision rule is uniquely (up to almost sure agreement) MMR optimal:
\begin{eqnarray*}
    d_{m_0^*}& :=& \mathbf{1}\{w_{m_0^*}^\top Y \geq 0\}, \\
    w_{m_0^*}^\top &:= & \left(1, \frac{\max \{m_0^* - C\left\Vert x_2-x_0 \right\Vert,0\}/\sigma_2^2}{(m_0^* - C\left\Vert x_1-x_0 \right\Vert)/\sigma_1^2},\ldots,\frac{\max \{m_0^* - C\left\Vert x_n-x_0 \right\Vert,0\}/\sigma_n^2}{(m_0^* - C\left\Vert x_1-x_0 \right\Vert)/\sigma_1^2}\right),
\end{eqnarray*}
where $m_0^* >  C\left\Vert x_1-x_0 \right\Vert$ solves a simple fixed point problem (\eqref{eq:small.discrepancy.m.star} in Online Appendix \ref{sec:proof.thm.running}).
\item[(ii)]
If
\begin{equation*}
C \left\Vert x_1 - x_0 \right\Vert = \sqrt{\pi / 2} \cdot \sigma_1, 
\end{equation*}
then $\mathbf{1}\{ Y_1 \geq 0\}$
is MMR optimal.
\item[(iii)] If 
\begin{equation}\label{eq:large.id.condition}
C \left\Vert x_1 - x_0 \right\Vert  > \sqrt{\pi / 2} \cdot \sigma_1, 
\end{equation}
then the rule 
\begin{equation}\label{eq:linear.rule.running.example}
d^{*}_{\text{linear}}(Y_1):=\begin{cases}
0, & Y_1<-\rho^{*},\\
\frac{Y_1+\rho^{*}}{2\rho^{*}}, & -\rho^{*}\leq Y_1\leq\rho^{*},\\
1, & Y_1>\rho^{*},
\end{cases}
\end{equation}
where $\rho^* \in (0,C \left\Vert x_1-x_0 \right\Vert)$ is uniquely defined by
$\rho^*=C \left\Vert x_1-x_0 \right\Vert (1-2\Phi\left(\rho^{*}/\sigma_1\right))$
is MMR optimal. So is the rule $d_{\text{RT}}^*(Y_1):=\Phi\bigl(Y_1/\tilde{\sigma})$  with
\begin{equation} \label{eqn:yata.example.1}
\tilde{\sigma} = \sqrt{ 2C^2 \left\Vert x_1-x_0\right\Vert^2/\pi - \sigma_1^2 }
\end{equation}
as well as all convex combinations of these rules.

\item[(iv)] If Equation \eqref{eq:large.id.condition} holds, no linear threshold rule is MMR optimal.
\end{itemize}
\end{prop}
\begin{proof}
See Online Appendix \ref{sec:proof.thm.running}.
\end{proof}
\begin{rem}
 All decision rules above converge to the one from case (ii) as $C \Vert x_1 - x_0 \Vert \to \sqrt{\pi / 2} \cdot \sigma_1$.
\end{rem}
\begin{rem}
This solution relates to the literature as follows. The problem is within the framework considered by \cite{yata2021} (who uses results from \citealt{stoye2012minimax}), and his analysis applies; in particular, our linear index differs from his $w^*$ only by a more explicit characterization. Although some of our proof steps use algebra from \cite{stoye2012minimax}, 
the alternative solutions in (iii), the uniqueness statement, and part (iv) are entirely new. \cite{ishihara2021} numerically find a solution within the class of symmetric threshold rules (i.e.,  rules of form $\mathbf{1}\{w^{\top}Y\geq 0\}$). This in principle recovers the global solution if $C \left\Vert x_1 - x_0 \right\Vert \leq \sqrt{\pi / 2} \cdot \sigma_1$ but will exclude all globally MMR optimal decision rules otherwise. That said, \citeauthor{ishihara2021}'s (\citeyear{ishihara2021})  solution approach applies considerably more generally. This is because we view their approach primarily as a numerical strategy for finding a minimax solution over a constrained class. Indeed, given a statistical model and a risk function, one can always try to find a minimax decision rule numerically within a class of threshold rules, which may or may not pin down the true minimax rule. \end{rem}
\begin{rem}
If $\mu$ is exogenous and known, then the decision rule $d^*_{\text{known }\mu}(\mu):= \max\{\min\{\overline{I}(\mu)/(\overline{I}(\mu)-\underline{I}(\mu)),1\},0\}$ uniquely attains MMR \citep{Manski2007}. Only our new rule $d^*_\text{linear}$ converges to $d^*_{\text{known }\mu}$ in certain special cases. Similarly, \citet[][Section 3.3]{ishihara2021} discuss that an analogous convergence fails for any MMR rule that they propose.
This may appear puzzling; however, any presumption that MMR rules ``should''  converge to such limits delicately depends on how one conceives the limit of the decision problem. 
Hence, it is not clear that we observe failure of any convergence that ``should'' have occurred.
\end{rem}

\begin{figure}
\begin{subfigure}[c]{0.45\textwidth}
\includegraphics[scale=0.41]{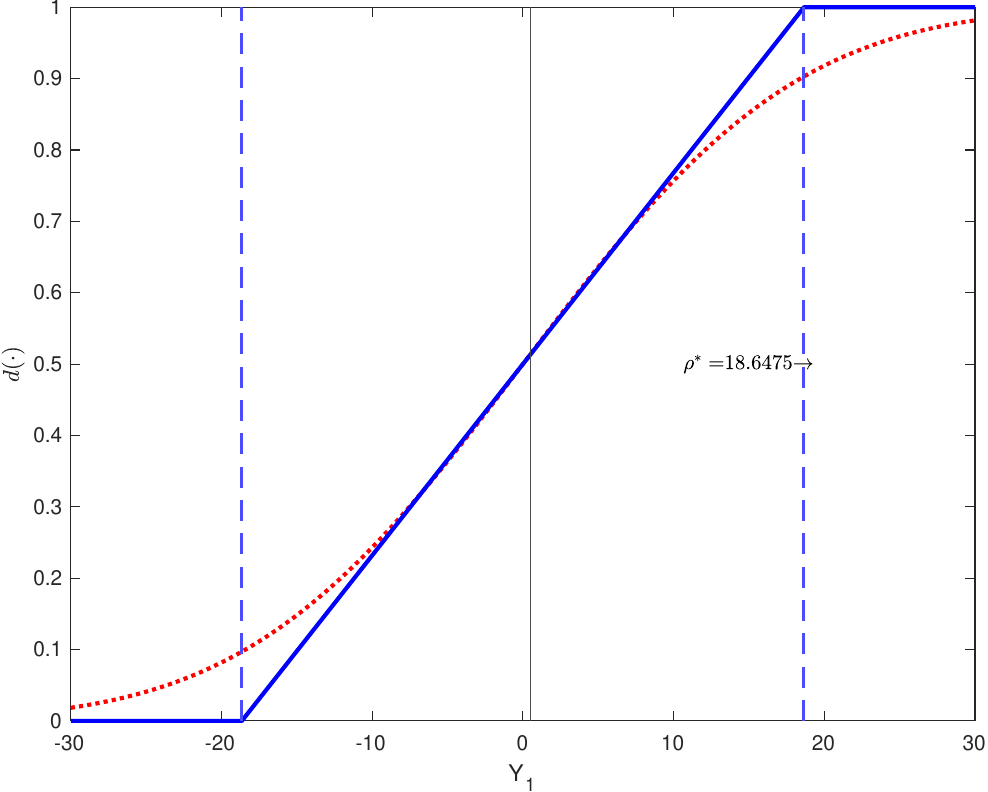}
\caption{$d^*_{\textrm{RT}}$ (red) and $d^*_{\textrm{linear}}$ (blue)}   
\end{subfigure}
\hfill
\begin{subfigure}[c]{0.45\textwidth}
\includegraphics[scale=0.41]{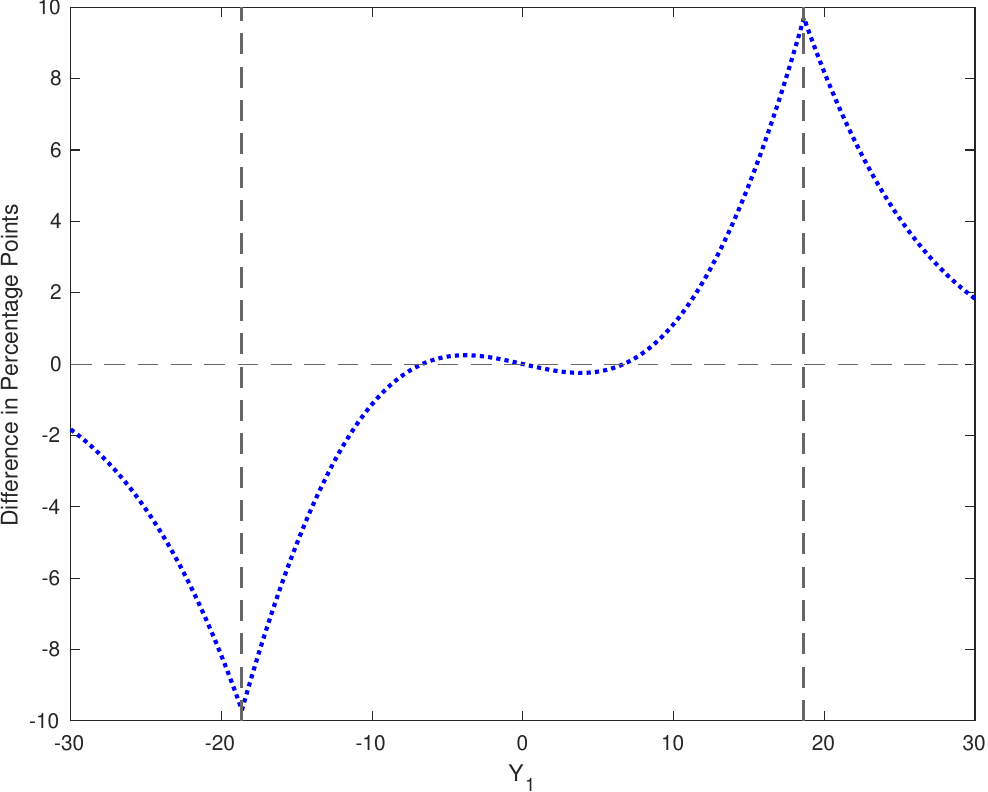}
\caption{$d^*_{\textrm{linear}}-d^*_{\textrm{RT}}$}    
\end{subfigure}
\caption{Panel a): Visualization of rules \eqref{eq:yata.rule} and \eqref{eq:linear.rule} in the running example with estimates from two countries (setting $x_0=0$, $x_1=-7.5$, $\sigma_1=3.9$, $x_2=7.9$, $\sigma_2=2.4$,  $C=2.5$). Panel b): Difference between rules as a function of the nearest neighbor's outcome.}  \label{figure:NewRule.1}
\end{figure}

We conclude that different MMR rules can lead to rather different policy actions for the same data. This difference is illustrated in Figure \ref{figure:NewRule.1} for parameter values calibrated to \citeauthor{ishihara2021}'s (\citeyear{ishihara2021}) empirical example. How serious a challenge it is depends on one's view. If one truly thinks of MMR as encoding a decision maker's complete preferences, and hence of competing optimal rules as mutually indifferent, it is not much of a concern. However, it may make it harder to communicate MMR-based decisions to policy makers. In addition, our next results  below will show that one may plausibly have preferences among the different MMR rules.

\noindent {\scshape Running Example---Continued:} Recall that, for some parameter values, there  are infinitely many MMR rules depending on the data only through $Y_1$. Thus, to better compare and visualize different rules, consider the $w^*$-profiled regret for $w^*=(1,0,\ldots,0)^{\top}$ defined as follows (also see the general definition and more detailed discussions of profiled regret in Appendix \ref{sec:profile.regret}): 
\begin{eqnarray*}
\overline{R}_{w^*}(d,\gamma) = \sup_{(\mu_0,\mu_1,\ldots,\mu_n)^\top \in \mathbb{R}^{n+1}: \mu_1 = \gamma,\mu_0 \in I(\mu),\mu\in M } \mu_0\left(\mathbf{1}\{\mu_0\geq0\}-\mathbb{E}_\mu[d(Y)]\right),
\end{eqnarray*}
and we can easily verify that $ \gamma \in \mathbb{R}$.\hfill  $\square$

For the same parameters used in Figure \ref{figure:NewRule.1}, Figure \ref{figure:NewRule} depicts $w^*$-profiled regret  over the range $\gamma \in [-30,30]$ for four decision rules. The blue (solid) line  is $d^*_{\text{linear}}$ (see Equation \eqref{eq:linear.rule}); the red (dotted) line is $d^*_{\text{RT}}$ (see Equation \eqref{eq:yata.rule}); the black (bimodal, solid) line is the symmetric threshold rule $d^*_0(Y)=\mathbf{1}\{Y_1 \geq 0\}$; and the green (dashed) line is $d_{\text{coin-flip}}=1/2$.\footnote{Appendix \ref{sec:profile.regret.computation} presents algebraic and computational details that underlie this figure.}
\begin{figure}[h!]
\centering
\includegraphics[scale=0.4]{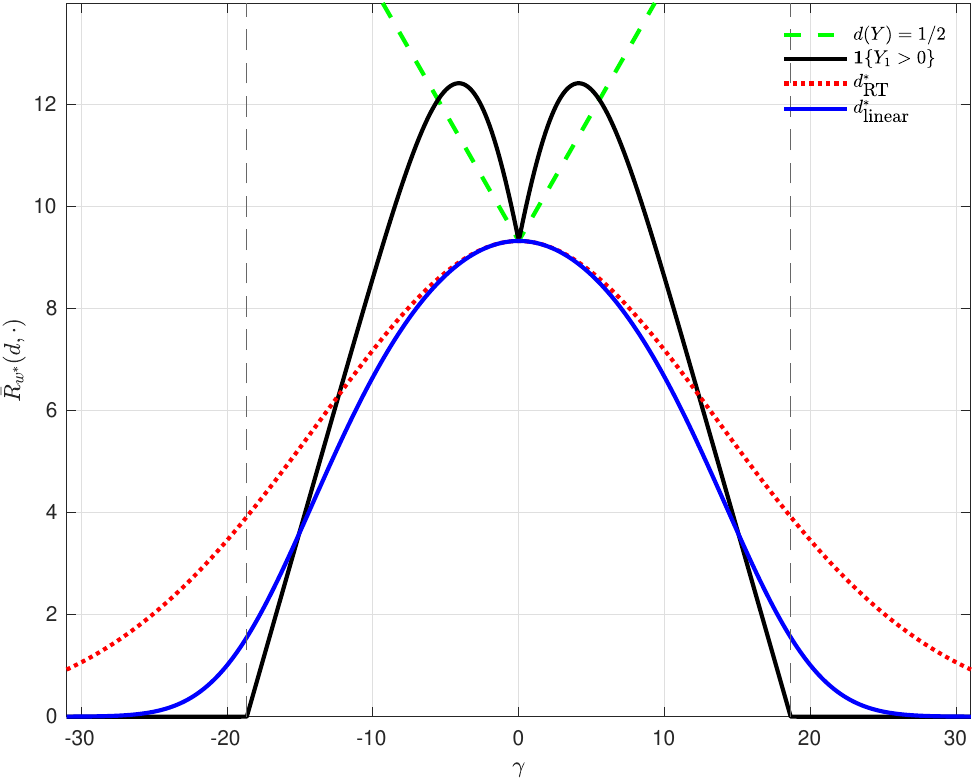}
\caption{$w^*$-profiled regret of four decision rules in the running example; parameter values are as in Figure \ref{figure:NewRule.1}}   
\label{figure:NewRule}
\end{figure}

An immediate use of Figure \ref{figure:NewRule} is to compare decision rules in terms of their worst-case regret. For example, consistent with Theorem \ref{thm:main.regret.1}-(ii), $d^*_0$ is not MMR optimal (the maximum of the black curve is clearly above those of $d^*_{\text{RT}}$ and $d^*_{\text{linear}}$).  For the specific parameter values used here, we can furthermore show that $d_0^*$ is minimax regret optimal among linear threshold rules; thus, the black line also illustrates the minimax regret efficiency loss from restricting attention to linear threshold rules. Figure \ref{figure:NewRule} furthermore reveals interesting differences among MMR optimal rules. In particular, $d^*_{\text{linear}}$ appears to have smaller $w^*$-profiled regret than $d^*_{\text{RT}}$ virtually everywhere. Indeed, Proposition \ref{prop:linear.vs.RT} in Appendix \ref{sec:profile.regret} shows that whenever condition (iii) of Proposition \ref{thm:running.example} holds, $d^*_{\text{linear}}$ dominates $d^*_{\text{RT}}$ in terms of $w^*$-profiled regret for a large range of values of $C \left\Vert x_1-x_0 \right\Vert$ and $\sigma_1$, including those used in Figure \ref{figure:NewRule} (though there exists parameter values under which the 
dominance does not hold for small nonzero $\gamma$).
  Moreover, the considerable difference between profiled regret functions in Figure \ref{figure:NewRule} continues beyond the figure: The ratio $\overline{R}_{w^*}(d^*_{\text{linear}},\gamma)/\overline{R}_{w^*}(d^*_{\text{RT}},\gamma)$ decays to zero at exponential rate as $\gamma \to \pm \infty$.  Finally, Figure \ref{figure:NewRule} illustrates that $d_{\text{coin-flip}}$ is $w^*$-profiled regret inadmissible in this example. In fact, profiled-regret inadmissibility of $d_{\text{coin-flip}}$ holds in a more general framework; see Proposition \ref{prop:profiled.regret.dominance} in Appendix \ref{sec:profile.regret} for additional results.

\subsection{Least Randomizing MMR Optimal Rules}\label{sec:least.randomizing} 

We next argue that further refining the MMR criterion presents an interesting research opportunity and may even lead to unique recommendations. To this purpose, we propose consideration of, and characterize, the \emph{least randomizing} MMR rule. Our main motivation is that, despite the wide adoption of randomized treatment allocations in economics and the social sciences, policy makers might shy away from exposing only a fraction of a population to the new policy. Thus, we attempt to recommend actions $a \in (0,1)$ as infrequently as possible. We justify our approach by a decision maker with a lexicographic preference (See Remark \ref{rem:lexicographic} below).  Alternatively, one may also explicitly incorporate aversion to randomization in the loss function and try to solve the MMR criterion under a modified  risk function. See Section \ref{sec:aversion.randomization} for detailed discussions.

To formalize this, observe that both $d^*_\text{RT}$ and $d^*_{\text{linear}}$ can be considered smoothed versions of $d^*_0$ in a sense that we now make precise. Let $F:\mathbb{R}\rightarrow[0,1]$ be a c.d.f. and consider a decision rule of form $F \circ w^*:=F((w^*)^{\top}Y)\in \mathcal{D}_n$; that is, the step function $d^*_0$ was smoothed into a c.d.f. We will restrict attention to c.d.f.'s whose associated distributions are \emph{symmetric} (i.e., $F(-x)=1-F(x)$) and \emph{unimodal} (i.e., $F(\cdot)$ is convex for $x\leq 0$ and concave otherwise). Let $\mathcal{F}$ be the set of all such c.d.f's and let
\[
\tilde{\mathcal{D}}_{n}:=\{F\circ w^*\in\mathcal{D}_{n}: F\in \mathcal{F}\}.
\]
Note that each rule $F\circ w^* \in \tilde{\mathcal{D}}_n$ depends on the data only via $(w^*)^\top Y$ and is nondecreasing in $(w^*)^\top Y$. Moreover, for each $F\circ w^*\in \tilde{\mathcal{D}}_n$, the interval on which treatment assignment is randomized equals (up to closure)
\begin{equation}\label{eq:randomization.region}
V(F\circ w^*):=\ensuremath{\bigl(\sup\bigl\{ x\in\mathbb{R}:F(x)=0\bigr\},\inf\bigl\{ x\in\mathbb{R}:F(x)=1\bigr\}\bigr)}.   
\end{equation}
All MMR decision rules considered in this paper are in $\tilde{\mathcal{D}}_{n}$. We next show that $d^*_{\text{linear}}$ is least randomizing among them and among all other MMR decision rules that might exist in this class. 
\begin{thm}\label{thm:main.regret.2}
Suppose all conditions of Theorem \ref{thm:main.regret.1} hold. If $F\circ w^*\in\tilde{\mathcal{D}}_n$ is MMR optimal, then $V(d^*_{\text{linear}}\circ w^*)\subseteq V(F\circ w^*)$, with equality if and only if $F=d^*_{\text{linear}}$.
\end{thm}
\begin{proof}
See Appendix \ref{sec:app.a.4}.
\end{proof}

In words, any symmetric, weakly increasing and unimodal MMR optimal rule that depends on data only via $(w^*)^{\top} Y$ must have a randomization area that is wider than that of $d^*_{\text{linear}}$, strictly so if it is a meaningfully distinct rule. Thus, the least randomizing criterion provides a pragmatic and unique refinement among the set of known MMR optimal rules.

To establish Theorem \ref{thm:main.regret.2}, we first show that for any rule $F\circ w^* \in \tilde{\mathcal{D}}_n$, its expected regret at any $\theta$ for which $(w^*)^\top m(\theta)=0$ equals the MMR value of the problem. If $F\circ w^* \in \tilde{\mathcal{D}}_n$ is MMR optimal, its expected regret must therefore be maximized at $(w^*)^\top m(\theta)=0$. For any symmetric and unimodal c.d.f. $F$ with $V(d^*_{\text{linear}}\circ w^*)\nsubseteq V(F\circ w^*)$, we can show that a necessary condition for this maximization fails. 

\noindent{\scshape Running Example---Continued:} Recall that, applied to the running example and for $C$ large enough, $d^*_{\text{linear}}$ can be expressed as \eqref{eq:linear.rule.running.example}.
Of note, an identified set for $\mu_0$ given the true mean of $Y_1$ (i.e., $\mu_1$) is 
\[
\left [ \: \mu_1 - C \left\Vert x_1-x_0 \right\Vert\:,\: \mu_1 + C \left\Vert x_1-x_0 \right\Vert\: \right], 
\]
which can be estimated naturally by 
\begin{equation}\label{eq:id.set.estimator.example}
\left [ \: Y_1 - C \left\Vert x_1-x_0 \right\Vert\:,\: Y_1 + C \left\Vert x_1-x_0 \right\Vert\: \right].  
\end{equation}

As $d^{*}_{\text{linear}}$ only randomizes when $\left\vert Y_1 \right\vert < \rho^*$, we see that interval \eqref{eq:id.set.estimator.example}
contains $0$ whenever $d^*_{\text{linear}}$ randomizes. Equivalently, $d^*_{\text{linear}}$ always (never) implements the new policy when \eqref{eq:id.set.estimator.example} 
is to the right (left) of zero. In this sense, the estimated identified set is explicitly used for decision making.\footnote{One may consider a scenario in which, to aid interpretability of decision rules, the decision maker does not wish to randomize whenever the estimated identified set for $\mu_0$ \eqref{eq:id.set.estimator.example} does not contain zero, i.e. the estimate suggests that the sign of $\mu_0$ is identified. This restricts decision rules to the following set:
\begin{equation}\label{eq:set.interpretability}
\{d\in \tilde{\mathcal{D}}_n:F\in\mathcal{F}, F(x)=0\text{ for all }x\leq -C \left\Vert x_1-x_0 \right\Vert\,\text {and } F(x)=1\text{ for all }x\geq C \left\Vert x_1-x_0 \right\Vert\}   
\end{equation}
As \eqref{eq:set.interpretability} is a subset of $\tilde{\mathcal{D}}_n$ containing $d^*_{\text{linear}}$, Theorem \ref{thm:main.regret.2} implies that $d^*_{\text{linear}}$ is also least-randomizing MMR in \eqref{eq:set.interpretability}. Moreover,  $d^*_{\text{linear}}$ is so far the only known rule in the literature  in set \eqref{eq:set.interpretability}  that is also globally MMR optimal.
 } In contrast, $d^*_\text{RT}$ always randomizes the policy recommendation, although for large $Y_1$ the fraction of population assigned to treatment will be large.\hfill $\square$
\begin{rem}\label{rem:lexicographic}
Consider a decision maker who wishes to pick an optimal rule according to MMR but is hesitant to implement randomized rules due to  additional inconvenience cost from implementation or concerns over \emph{ex-post} fairness (i.e., units in the same population may receive different treatments).  In this case, it is indeed possible to define the MMR problem among nonrandomized rules. However, if we do not place any restriction to the class of nonrandomized rules, game-theoretic purification arguments \citep{dww,purify} suggest existence of nonrandomized solutions that emulate arbitrarily well the risk profile of the randomized MMR optimal rule. Moreover, these solutions would be unattractive and unnatural; for example, they cannot be monotone due to Theorem \ref{thm:main.regret.1}(iii). In light of these observations, we think there is little value in pursuing solutions over nonrandomized rules in our setup, unless one focuses on a specific class of benign or reasonable nonrandomized rules, e.g. as in \cite{ishihara2021}. Furthermore, within the class of rules $\tilde{\mathcal{D}}_n$, our least randomizing MMR rule can be justified by a decision maker who possesses a lexicographic preference with a priority given to minimizing MMR criterion over minimizing the inconvenience cost of randomization (in our case, the cost is proxied by the Lebesgue measure of the realization of $(w^*)^\top Y$ taking a value in $(0,1)$). Such lexicographic ordering of statistical decisions has ample precedents in the literature, e.g., for statistical decision rules, we first remove dominated rules and then choose among undominated ones according to further optimality criteria \citep{Wald50,manski2021econometrics}; In the Neyman-Pearson paradigm of selecting hypothesis tests, one first controls size and then maximizes power; In choosing estimators, it is customary to focus on unbiased estimators, among which the one that minimizes variances is regarded as optimal. 
 \end{rem}

\section{Further Applications}\label{sec:examples}  

\subsection{Extrapolating Local Average Treatment Effects}\label{sec:iv.example}

We next apply our analysis to extrapolation of Local Average Treatment Effects \citep{mogstad2018using,mogstad2018identification}. Let $Z\in\{0,1\}$ be a binary instrument, $D\in\{0,1\}$ a binary treatment assignment, and $(Y(1),Y(0))$ potential outcomes under treatment and control.  As usual, the observed outcome is $Y=DY(1)+(1-D)Y(0)$. To simplify exposition, we assume that there are no covariates and that $Y(1),Y(0)\in\{0,1\}$. Following \citet{heckman1999local,heckman2005structural},\footnote{See, for example, Assumption I and Equation (2) in \cite{mogstad2018identification}. Also see \cite{imbens1994identification} for additional references.} let $p(z):=P\{D=1 \mid Z=z\}$ be the propensity score and write $D = \mathbf{1}\{ V \leq  p(Z)\}$, where $(V\mid Z=z) \sim \textrm{Unif}[0,1]$. The parameter space $\Theta$  contains all tuples $\theta:=(p(1),p(0),\text{MTE}(\cdot))$, where  $p(1)\in[0,1]$, $p(0)\in[0,1]$, $p(1)\geq p(0)$, and $\text{MTE}(\cdot)$ is the marginal treatment effect function
\[\text{MTE}(v):=\mathbb{E}[Y(1)-Y(0)\mid V=v].\]
The policy maker observes 
\begin{equation}\label{eq:IV.data}
\left(\begin{array}{c}
\hat{m}_1\\
\hat{m}_2
\end{array}\right)\sim N\left(\left(\begin{array}{c}
m_1(\theta)\\
m_2(\theta)
\end{array}\right),\Sigma\right),    
\end{equation}
where
\begin{eqnarray*}
m_1(\theta) &:=& \mathbb{E}[Y\mid Z=1]-\mathbb{E}[Y\mid Z=0]=\int_{p(0)}^{p(1)}\text{MTE}(v)dv \\
m_2(\theta) &:=& \mathbb{E}[D\mid Z=1]-\mathbb{E}[D\mid Z=0]=p(1)-p(0)
\end{eqnarray*}
are the population reduced-form and first-stage coefficients and $\Sigma$ is positive definite. We assume that the policy of interest would expand the complier subpopulation through an additive shift of size $\alpha>0$ in the propensity score. \cite{mogstad2018using} show that the payoff relevant parameter then is the ``policy-relevant treatment effect'' \citep{heckman2005structural} 
\[\text{PRTE}(\alpha)=\mathbb{E}[Y(1)-Y(0)\mid V\in(p(0),p(1)+\alpha]].  \]
Suppose the welfare contrast $U(\theta)$ equals $\text{PRTE}(\alpha)-\text{PRTE}(0)$, which can be written as 
\begin{equation}\label{eq:IV.welfare}
U(\theta)=\frac{m_1(\theta)}{\alpha+m_2(\theta)}-\frac{m_1(\theta)}{m_2(\theta)}+\frac{1}{\alpha+m_2(\theta)}\int_{p(1)}^{p(1)+\alpha}\textrm{MTE}(v)dv.
\end{equation}
Hence, the decision maker wants to find an optimal treatment policy given partial identification of parameter \eqref{eq:IV.welfare} in model \eqref{eq:IV.data}. In Online Appendix \ref{sec:iv.example.verification}, we verify that Theorem \ref{thm:admin} applies. Therefore, any decision rule is admissible in this example. For example, implementing a policy for large values of the IV estimator would be admissible, as would be the approach of \citet{christensen2022optimal}, who discuss the same application.

\subsection{Decision-theoretic Breakdown Analysis} 
Consider a policy maker who uses quasi-experimental data but is worried about confounding. More specifically, she assumes a constant treatment effect model and unconfoundedness given covariates $(X,W)$, motivating the linear regression model
\[
Y=\gamma_{0}+\beta_{\text{long}}D+\gamma_{1}^{\top}X+\gamma_{2}^{\top}W+e,
\]
where $Y$ is observed outcome, $D$ is the binary treatment, and $e$ is a projection residual. The infeasible optimal treatment policy is $\mathbf{1}\{\beta_{\text{long}}\geq0\}$. However, $W$ is unobserved, so that the policy maker can only estimate the ``medium'' regression
\[
Y=\pi_{0} +\beta_{\text{med}} D+\pi_{1}^{\top} X+u,
\]
where $u$ is a projection residual.\footnote{We express all regressions as projections for alignment with the literature and because only projection algebra is used. However, motivating $\mathbf{1}\{\beta_{\text{long}}\geq 0\}$ as optimal usually requires causal interpretation and therefore slightly stronger assumptions on $e$; in other words, readers may want to think of the long regression as causal and the medium one as best linear prediction. See \citet[][ Chapter 2]{Hansen}, whose notation we also borrow, for a lucid discussion.} In general, if there exists selection on unobservables (that is, $D$ is correlated with $W$), then $\beta_{\text{long}}$ is only partially identified. Specifically, \citet[][Theorem 4]{diegert2022assessing} show that the identified set of $\beta_{\text{long}}$ given $\beta_{\text{med}}$ is
\[
\beta_{\text{long}}\in[\beta_{\text{med}}-k,\beta_{\text{med}}+k],
\]
where 
\[
k:=\begin{cases}
\sqrt{\frac{\operatorname{var}\left(Y^{\bot D,X}\right)}{\operatorname{var}\left(D^{\bot X}\right)}\frac{\overline{r}_{D}^{2}R_{D\sim X}^{2}}{1-R_{D\sim X}^{2}-\overline{r}_{D}^{2}}},& \quad\text{ if } 0\leq\overline{r}_{D}< \sqrt{1-R^2_{D\sim X}},\\
\infty,&\quad \text{ if } \overline{r}_{D}\geq \sqrt{1-R^2_{D\sim X}},
\end{cases}
\]
and where $\operatorname{var}\left(Y^{\bot D,X}\right)$ is the variance of
the residual from projecting $Y$ onto $(1,D,X)$, $\operatorname{var}\left(D^{\bot X}\right)$
is the variance of the residual from projecting $D$ onto $(1,X)$,
$R_{D\sim X}^{2}$ is the $R^{2}$ from projecting $D$ onto $(1,X)$,
and $\overline{r}_{D}\geq0$ is a user-specified sensitivity parameter that measures the relative importance of selection on unobservables versus selection on observables.

\citet{diegert2022assessing} use this result to ask: How strong does omitted variables bias have to be to potentially overturn findings based on $\beta_{\text{med}}$? At population level, the answer is that this can happen if $\left\vert k \right\vert > \left\vert \beta_{\text{med}}\right\vert$, a condition that can be related to primitive parameters through the above display and for which \citet{diegert2022assessing} provide estimation and inference theory.

Suppose now that there is an estimator $\hat{\beta}_{\text{med}} \sim N(\beta_{\text{med}},\sigma^{2})$. The results of \citet{diegert2022assessing} imply an estimated breakdown point
$\tilde{k}(\hat{\beta}_\text{med}):=\hat{\beta}_{\text{med}}$ for positive $\hat{\beta}_{\text{med}}$. Our results apply upon letting $\theta =(\beta_{\text{long}},\beta_{\text{med}})^\top\in \mathbb{R}^2$, $U(\theta)=\beta_{\text{long}}$ and $m(\theta)=\beta_{\text{med}}$. In particular, when $k>\sqrt{\frac{\pi}{2}}\sigma$, there are infinitely many MMR optimal rules, with the least randomizing one among known ones being
\begin{equation}\label{eq:ovb.linear.rule}
d^{*}_{\text{linear}}(\hat{\beta}_{\text{med}}):=\begin{cases}
0, & \hat{\beta}_{\text{med}}<-\rho^{*},\\
\frac{\hat{\beta}_{\text{med}}+\rho^{*}}{2\rho^{*}}, & -\rho^{*}\leq \hat{\beta}_{\text{med}}\leq\rho^{*},\\
1, & \hat{\beta}_{\text{med}}>\rho^{*},
\end{cases}
\end{equation}
where $\rho^*>0$ uniquely solves $\rho^*=k(1-2\Phi(-\rho^*/\sigma))$.  When $k\leq\sqrt{\frac{\pi}{2}}\sigma$, we know from \cite{stoye2012minimax} that $d^*_0(\hat{\beta}_{\text{med}}) = \mathbf{1}\{\hat{\beta}_{\text{med}}\geq0\}$ is essentially uniquely MMR optimal.

\begin{figure}
    \centering
    \includegraphics[scale=0.6
]{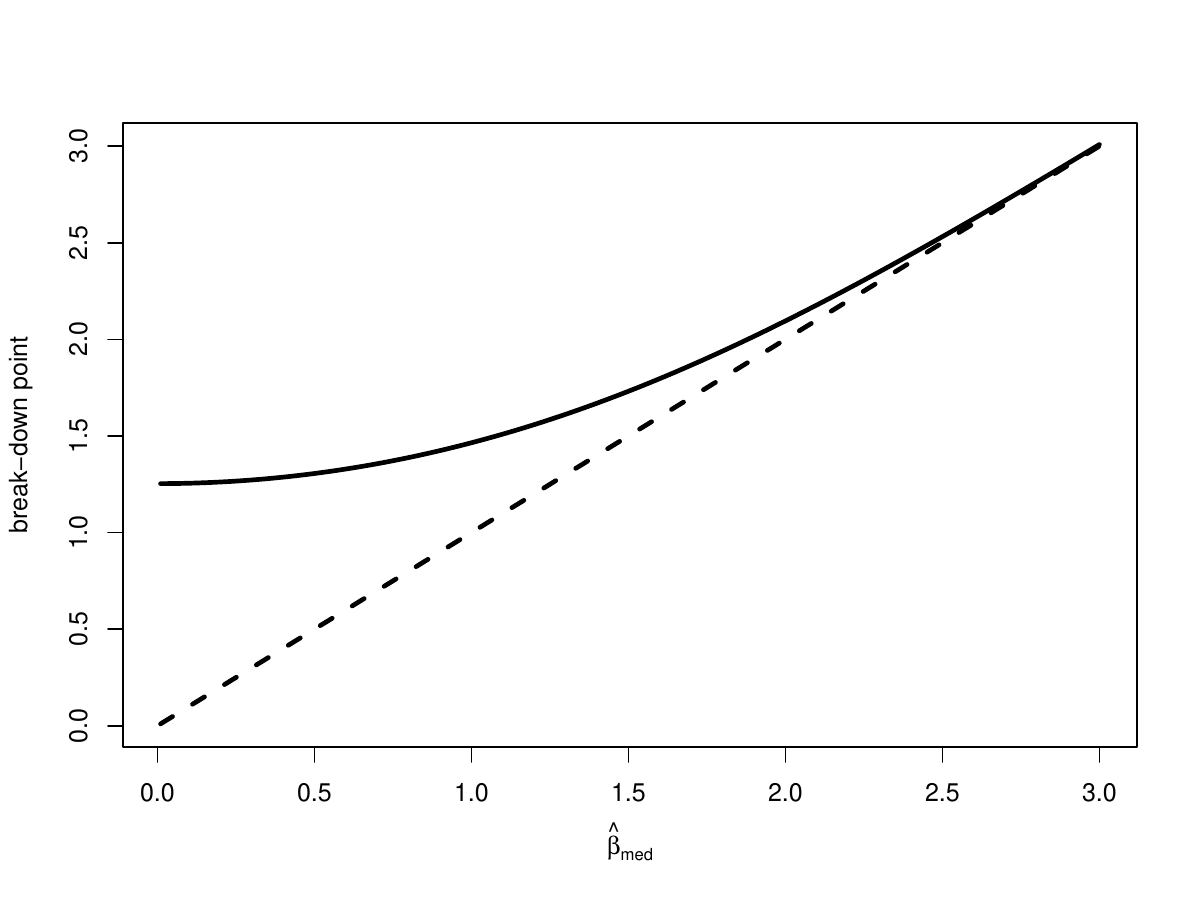}
    \caption{$\overline{k}$ (solid) and $\tilde{k}$ (dashed) as a function of $\hat{\beta}_{\text{med}}$}\label{figure:breakdown}
\end{figure}

These results motivate a complementary breakdown analysis guided by statistical decision theory. For given $\hat{\beta}_{\text{med}}>0$, we can ask: How large could $k$ have to be so that the MMR optimality criterion still supports assigning the new policy without any hedging?\footnote{Informal exploration of this question goes back at least to \citet[][see Table 3]{stoye2009partial}.} Due to its least randomizing property, $d_{\text{linear}}^*$ implies the tightest possible answer to this question. Specifically, MMR supports non-randomized policy assignment up to the ``decision theoretic breakdown point'' 
\begin{eqnarray*}
    \overline{k}(\hat{\beta}_\text{med}) &:= &\sup\{k>0: d^*(\hat{\beta}_\text{med})=1\} \\
    d^*(\hat{\beta}_\text{med})&:= &\begin{cases}
  \mathbf{1}\{\hat{\beta}_{\text{med}}\geq0\}&\text{ if } k\leq\sqrt{\frac{\pi}{2}}\sigma\\
  d^*_{\text{linear}}(\hat{\beta}_{\text{med}})&\text{ if } k>\sqrt{\frac{\pi}{2}}\sigma \end{cases}.
\end{eqnarray*}

In contrast, the implied breakdown point of $d^*_{\text{RT}}$ is a constant across all values of $\hat{\beta}_{\text{med}}$, as $d^*_{\text{RT}}$ always randomizes regardless of data realizations when 
$k>\sqrt{\frac{\pi}{2}}\sigma$. Figure \ref{figure:breakdown} displays both $\tilde{k}(\hat{\beta}_\text{med})$ and  $\bar{k}(\hat{\beta}_\text{med})$  when $\sigma=1$.  It turns out that the decision theoretic breakdown point tolerates more ambiguity; this difference is salient for smaller values of $\hat{\beta}_{\text{med}}$ and vanishes as $\hat{\beta}_{\text{med}}$ diverges.

\section{Conclusion}\label{sec:conclude}

In this paper, we used statistical decision theory to study treatment choice problems with partial identification. For a large and empirically relevant class of such problems, we show that every decision rule is admissible, that maximin welfare optimality criterion often select no-data decision rules, and that there are infinitely many minimax regret optimal rules, all of which randomize the policy action at least for some data realizations. These results stand in stark contrast with treatment choice problems with point-identified welfare. 

We also provide a decision rule that is \emph{least randomizing} in a large class of MMR optimal rules including all known ones. We show, in the context of our running example, that our least-randomizing rule can profiled-regret dominate other MMR rules.

We illustrate our results in three applications that arise in applied work: extrapolation of experimental estimates for policy adoption, policy-making with quasi-experimental data when omitted variable bias is a concern, and extrapolation of Local Average Treatment Effects.

\appendix

\section{Proofs of Main Results}\label{sec:app.main}

\subsection{Proof of Theorem \ref{thm:admin}}\label{sec:app.a.1}

It suffices to verify conditions A1-A2 in Theorem \ref{thm:general} for the class
of distributions $N(m(\theta),\Sigma),\theta\in\Theta$.

\textsc{A1}: Consider the class of distributions $N(m(\theta),\Sigma),\theta\in m^{-1}(\mathcal{S})$.
As $\theta$ ranges over $m^{-1}(\mathcal{S})$, $m(\theta)$ ranges
over $\mathcal{S}$. Thus, it suffices to show that 
\begin{equation}
Y\sim N(\mu,\Sigma),\quad\mu\in\mathcal{S},\label{eq:normal.location}
\end{equation}
is \emph{complete} \citep[][ Definition 6.2.21]{casella2002statistical} and therefore also bounded complete. Define the vector $\tilde{\mu} : = \Sigma^{-1} \mu$ ranging over the set \[ \tilde{\mathcal{S}} := \{ \tilde{\mu} \in \mathbb{R}^{n} \: | \: \tilde{\mu} = \Sigma^{-1} \mu, \quad \mu \in \mathcal{S}   \}.\] 
Note $\tilde{\mathcal{S}}$ is open under Definition \ref{asm:1}. Furthermore, the pdf $p_{\tilde{\mu}}$ of $Y$ given $\tilde{\mu}$ is of the form
\[ p_{\tilde{\mu}}(Y) =  h(Y)C(\tilde{\mu}) \exp \left[ \tilde{\mu}^{\top} Y  \right]. \] Thus, the family of
distributions for $Y$ is complete by \citet[][ Theorem 6.2.25]{casella2002statistical}.

\textsc{A2}: Consider any set $E$ such that $P(E)=0$ for all $P\sim N(m(\theta),\Sigma)$
as $\theta$ ranges over $m^{-1}(\mathcal{S})$. By \citet{skorohod2012integration}
(Theorem 2, p. 83), the Gaussian measures in $\mathbb{R}^{n}$ with
the same positive definite covariance matrix $\Sigma$ are equivalent (mutually absolute continuous with each other),
implying $P(E)=0$ for all $P\sim N(m(\theta),\Sigma)$ such that
$\theta\in\Theta\setminus m^{-1}(\mathcal{S})$ as well.

\subsection{Proof of Theorem \ref{thm:maximin}}\label{sec:app.a.2}

First, we can bound
\begin{multline*}
\sup_{d\in\mathcal{D}_n}\inf_{\theta\in\Theta}\mathbb{E}_{m(\theta)}\left[W(d(Y),\theta)\right]=\sup_{d\in\mathcal{D}_n}\inf_{\theta\in\Theta}\left[W(0,\theta)+U(\theta)\mathbb{E}_{m(\theta)}\left[d(Y)\right]\right]\\
\leq\sup_{d\in\mathcal{D}_n}\inf_{\theta\in\Theta:U(\theta)\leq0}\left[W(0,\theta)+U(\theta)\mathbb{E}_{m(\theta)}\left[d(Y)\right]\right] 
  \leq\inf_{\theta\in\Theta:U(\theta)\leq0}W(0,\theta),
\end{multline*}
using that $\mathbb{E}_{m(\theta)}\left[d(Y)\right]\geq0$. To see that this bound is tight, write
\begin{align*}
\inf_{\theta\in\Theta}\mathbb{E}_{m(\theta)}\left[W(d_{\text{no-data}},\theta)\right]=\inf_{\theta\in\Theta}\left[W(0,\theta)+U(\theta)\mathbb{E}_{m(\theta)}\left[d_{\text{no-data}}(Y)\right]\right] & =\inf_{\theta\in\Theta}W(0,\theta).
\end{align*}
and recall that $\inf_{\theta\in\Theta}W(0,\theta)=\inf_{\theta\in\Theta:U(\theta)\leq0}W(0,\theta)$ by assumption.

\subsection{Proof of Theorem \ref{thm:main.regret.1}}\label{sec:app.a.3}

\subsubsection{Proof of Part (i) of Theorem \ref{thm:main.regret.1}} 

Let $\textbf{R}$ denote the minimax value of the policy maker's decision problem:
\begin{equation}\label{eq:MMR.value}
\textbf{R}:= \inf_{d \in \mathcal{D}_n} \sup_{\theta \in \Theta} \left\{  U(\theta) \left(  \mathbf{1}\{U(\theta) \geq 0\} - \mathbb{E}_{m(\theta)}[d(Y)] \right) \right\}.
\end{equation}

\noindent {\scshape Step 1} (Minimax Regret Value): We first show that, if Assumption \ref{asm:yata.1} holds and there exists an MMR optimal rule that depends on the data only through $(w^*)^{\top} Y$ and satisfies Equations \eqref{eq:yata.property.1} and \eqref{eq:yata.property.2}, then
\begin{equation*} 
\textbf{R} = (1/2) \cdot \overline{k}_{w^*}(0),
\end{equation*}
where
\begin{equation*} 
\overline{k}_{w^*}(0) := \sup_{\theta \in \Theta} U(\theta) \quad \textrm{s.t }  \quad (w^*)^{\top} m(\theta) = 0. 
\end{equation*}

Lemma \ref{lem:yata.lower.bound} in Appendix \ref{sec:technical} shows that, under this step's premise, Equation \eqref{eq:yata.property.1} implies
\[(1/2) \cdot \overline{k}_{w^*}(0) \leq \textbf{R}.\]
Since $d^*$ is MMR optimal, Equations \eqref{eq:yata.property.1}-\eqref{eq:yata.property.2} and centrosymmetry of $\Theta$ imply
\[
\textbf{R} = \sup_{\theta\in\Theta,m(\theta)=\mathbf{0}}R(d^*,\theta)= \sup_{\theta \in \Theta, m(\theta)= \mathbf{0}} U(\theta) \left( \mathbf{1}\{ U(\theta) \geq 0 \} - \frac{1}{2} \right)=\frac{\overline{I}(\mathbf{0})}{2}.
\]
By definition, $\overline{I}(\mathbf{0})\leq\overline{k}_{w^*}(0)$. Thus, we have
\[ (1/2) \cdot \overline{k}_{w^*}(0) \leq \textbf{R} = (1/2) \cdot \overline{I}(\mathbf{0}) \leq  (1/2) \cdot \overline{k}_{w^*}(0).\] 

\noindent {\scshape Step 2} (Upper bound for the worst-case regret of decision rules that depend on the data only through $(w^*)^{\top} Y$): We obtain an upper bound for the worst-case regret of such rules by linearizing the parameter space. We introduce some notation to formalize this step.    

Let  $\Gamma_{w^*}:=\left\{ \gamma\in\mathbb{R}|(w^*)^{\top} m(\theta) =\gamma,\theta\in\Theta\right\}$ be the image of the transformation $\theta \mapsto (w^*)^{\top} m(\theta)$. We define the identified set for the welfare contrast $U(\theta)$ given $\gamma\in \Gamma_{w^*}$ as 
\begin{equation*}
ISU_{w^*}(\gamma) := \{ u \in \mathbb{R} \: | \: U(\theta)=u, (w^*)^{\top} m(\theta) = \gamma, \theta \in \Theta  \}.
\end{equation*}

Any decision rule that depends on the data only through $(w^*)^{\top} Y$ can be identified with a (measurable) function $d$ from $\mathbb{R}$ to $[0,1]$. For future reference, let $\mathcal{D}$ collect all such functions. The worst-case expected regret of any $d \in \mathcal{D}$ can be expressed as
\begin{align} 
& \sup_{\theta \in \Theta} \left( U(\theta)\left(\mathbf{1}\{U (\theta)\geq0\}-\mathbb{E}_{\gamma}[d((w^{*})^{\top}Y)]\right)\right)  \label{eq:worst.case.of.linear.d}  \\
=& \sup_{\gamma\in\Gamma_{w^{*}}}\left(\sup_{ \theta \in \Theta, \: (w^*)^{\top} m(\theta) = \gamma } U(\theta)\left(\mathbf{1}\{U (\theta)\geq0\}-\mathbb{E}_{\gamma}[d((w^{*})^{\top}Y)]\right)\right) \notag\\
= & \sup_{\gamma\in\Gamma_{w^{*}}}\left(\sup_{U^{*}\in ISU_{w^*}(\gamma) }U^{*}\left(\mathbf{1}\{U^{*}\geq0\}-\mathbb{E}_{\gamma}[d((w^{*})^{\top}Y)]\right)\right) \label{eq:R.isu.bound.1},
\end{align} 
where the expectation $\mathbb{E}_\gamma[\cdot]$ is taken over $(w^*)^\top Y \sim N(\gamma,(w^*)^\top \Sigma w^*)$. For  $\gamma\in \Gamma_{w^*}$, define
\begin{align}
\underline{k}_{w^*}(\gamma)  := & \inf ISU_{w^*}(\gamma) =  \inf \{U(\theta): (w^*)^{\top} m(\theta) = \gamma, \theta \in \Theta  \}  \notag \\
\overline{k}_{w^*}(\gamma) := & \sup ISU_{w^*}(\gamma) =  \sup \{U(\theta): (w^*)^{\top} m(\theta) = \gamma, \theta \in \Theta  \}  \label{eq:isu}.
\end{align}
By centrosymmetry of $\Theta$ and linearity of $U(\theta)$ and $m(\theta)$, we have that
\begin{align*}
    & \inf \{U(\theta): (w^*)^{\top} m(\theta) = \gamma, \theta \in \Theta  \} =- \sup \{U(\theta): (w^*)^{\top} m(\theta) = -\gamma, \theta \in \Theta  \} \\
    \implies & \underline{k}_{w^*}(\gamma)=-\overline{k}_{w^*}(-\gamma)
\end{align*}
and therefore that $ISU_{w^*}(\gamma) \subseteq \left [ - \overline{k}_{w^*}(-\gamma), \:   \overline{k}_{w^*}(\gamma)  \right],  \forall \gamma \in \Gamma_{w^*}$.
Equation \eqref{eq:R.isu.bound.1} then implies that \eqref{eq:worst.case.of.linear.d} is bounded above by
\begin{equation}\label{eq:R.isu.bound.2}
\sup_{\gamma\in\Gamma_{w^{*}}}\left(\sup_{U^{*}\in \left [ - \overline{k}_{w^*}(-\gamma), \:   \overline{k}_{w^*}(\gamma)  \right]}U^{*}\left(\mathbf{1}\{U^{*}\geq0\}-\mathbb{E}_{\gamma}[d((w^{*})^{\top}Y)]\right)\right).
\end{equation}

Lemma \ref{lem:ex.1.3} shows that, under Assumption \ref{asm:yata.1}, $\overline{k}_{w^*}(\gamma) $ is concave (and therefore $- \overline{k}_{w^*}(-\gamma)$ is convex). Furthermore, Lemma \ref{lem:ex.1.4} shows that, under Assumption \ref{asm:yata.1}, the superdifferential of the function $\overline{k}_{w^*}(\cdot)$ at $\gamma=0$, 
\begin{equation*}
 \partial \overline{k}_{w^*}(0):= \{s \in \mathbb{R} \: | \: \overline{k}_{w^*} (\gamma) \leq \overline{k}_{w^*}(0) + s  \gamma, \quad \forall \gamma \in \Gamma_{w^*} \},      
\end{equation*}
(see \citet[][p. 214-215]{rockafellar1997convex}) is nonempty, bounded, and closed. 

Let $s_{w^*}(0)$ be the largest element of $\partial \overline{k}_{w^*}(0)$ and suppose without loss of generality that $s_{w^*}(0) \geq 0$. Step 3 in the proof of Lemma \ref{lem:ex.1.4} established that $\Gamma_{w^*}$ is symmetric around $0$. The definition of superdifferential then gives   
\begin{equation*}\label{pf:linear.lower.bound}
s_{w^*}(0) \gamma -\overline{k}_{w^*}(0)  \leq  -\overline{k}_{w^*}(-\gamma)  \quad \forall \gamma \in \Gamma_{w^*}.
\end{equation*}
It follows that
\[  [-\overline{k}_{w^*} (-\gamma) \: , \:  \overline{k}_{w^*} (\gamma)] \subseteq [s_{w^*}(0) \gamma -\overline{k}_{w^*}(0) \: , \: s_{w^*}(0) \gamma + \overline{k}_{w^*}(0)   ]. \]
Substituting into \eqref{eq:R.isu.bound.2} then implies that \eqref{eq:worst.case.of.linear.d} is further bounded above by
\begin{equation*} 
  \sup_{\gamma\in\Gamma_{w^{*}}}\left(\sup_{U^{*}\in \left [s_{w^*}(0) \gamma -\overline{k}_{w^*}(0) \: , \: s_{w^*}(0) \gamma + \overline{k}_{w^*}(0)   \right ]}U^{*}\left(\mathbf{1}\{U^{*}\geq0\}-\mathbb{E}_{\gamma}[d((w^{*})^{\top}Y)]\right)\right).
\end{equation*}
We note that the choice set for $U^*$ in this expression is linear in $\gamma$.

{\scshape Step 3} (``Linear Embedding'' Minimax Regret Problem): The previous step and the fact that $\Gamma_{w^*} \subseteq \mathbb{R}$ imply that \eqref{eq:worst.case.of.linear.d} is bounded above by \begin{equation} \label{eq:linear.embedding.1}
  \inf_{d \in \mathcal{D}} \sup_{\gamma\in \mathbb{R}}\left(\sup_{U^{*}\in \left [s_{w^*}(0) \gamma -\overline{k}_{w^*}(0) \: , \: s_{w^*}(0) \gamma + \overline{k}_{w^*}(0)   \right ]}U^{*}\left(\mathbf{1}\{U^{*}\geq0\}-\mathbb{E}_{\gamma}[d\left( \widehat{\gamma} \right)]\right)\right),
\end{equation}
where 
\[ \widehat{\gamma} \sim N \left( \gamma \: , \: (w^*)^\top\Sigma w^* \right), \quad \gamma \in \mathbb{R}. \]

Lemma \ref{lem:linear.embedding} in Appendix \ref{sec:technical} shows that if $s_{w^*}(0)>0$ and $\overline{k}_{w^*}(0) > \sqrt{\frac{\pi}{2}} \sqrt{(w^*)^\top\Sigma w^*}\cdot s_{w^*}(0)$, then (\ref{eq:linear.embedding.1}) equals $\overline{k}_{w^*}(0)/2$ and there are infinitely many rules that give such value. In particular, a solution is given by any convex combination of the following rules: 
\begin{eqnarray} 
d^*_{\text{RT}}(\widehat{\gamma}) & := &\Phi\left(\widehat{\gamma}/\sqrt{\frac{2 \cdot \overline{k}_{w^*}(0)^2}{\pi \cdot  s_{w^*}(0)^2}-(w^*)^\top\Sigma w^*}\right) \label{eq:yata.stoye.appendix} \\
d^*_{\text{linear}}(\widehat{\gamma}) & := &\begin{cases}
0, & \widehat{\gamma}<-\rho^{*},\\
\frac{\widehat{\gamma}+\rho^{*}}{2\rho^{*}}, & -\rho^{*}\leq\widehat{\gamma}\leq\rho^{*},\\
1, & \widehat{\gamma}>\rho^{*},
\end{cases}   \label{eq:linear.appendix} 
\end{eqnarray}
where $\rho^* \in \left(0,\frac{ \overline{k}_{w^*}(0)}{s_{w^*}(0)} \right)$ is the unique solution to
\begin{eqnarray}
\left( \frac{s_{w^*}(0)}{2 \cdot \overline{k}_{w^*}(0)} \right) \rho^{*} -\frac{1}{2} + \Phi\left(-\frac{\rho^{*}}{\sqrt{(w^*)^{\top}\Sigma w^*}}\right) =0.
\label{eq:rho.star.main.proof} 
\end{eqnarray}

{\scshape Step 4} (Rules that solve the ``Linear Embedding'' minimax regret problem also solve the original problem). If $s_{w^*}(0)>0$, $\overline{k}_{w^*}(0) > \sqrt{\frac{\pi}{2}} \sqrt{(w^*)^\top\Sigma w^*}\cdot s_{w^*}(0)$, and $d^{\bigstar} \in D$ solves the linear embedding problem, then 
\begin{equation*}
d^{\bigstar} \circ w^* (Y) := d^{\bigstar}((w^*)^{\top} Y) \in \mathcal{D}_n 
\end{equation*}
is MMR optimal in the original decision problem (\ref{eq:MMR.value}). 

This is because Step 1 implies that
\[ (1/2) \overline{k}_{w^*}(0) = \textbf{R} \leq \sup_{\theta \in \Theta} \left\{  U(\theta) \left(  \mathbf{1}\{U(\theta) \geq 0\} - \mathbb{E}_{m(\theta)}[d^{\bigstar}((w^*)^{\top} Y)] \right) \right\} \]

and Step 2 implies that
\begin{eqnarray*}
&& \sup_{\theta \in \Theta} \left\{  U(\theta) \left(  \mathbf{1}\{U(\theta) \geq 0\} - \mathbb{E}_{m(\theta)}[d^{\bigstar}((w^*)^{\top} Y)] \right) \right\} \\
&\leq &  \sup_{\gamma\in \mathbb{R}}\left(\sup_{U^{*}\in \left [s_{w^*}(0) \gamma -\overline{k}_{w^*}(0) \: , \: s_{w^*}(0) \gamma + \overline{k}_{w^*}(0)   \right ]}U^{*}\left(\mathbf{1}\{U^{*}\geq0\}-\mathbb{E}_{\gamma}[d^{\bigstar}\left( \widehat{\gamma} \right)]\right)\right) \\
&= & \frac{\overline{k}_{w^*}(0)}{2},
\end{eqnarray*}
where the last equality follows from Step 3. Consequently, 
\[ \sup_{\theta \in \Theta} \left\{  U(\theta) \left(  \mathbf{1}\{U(\theta) \geq 0\} - \mathbb{E}_{m(\theta)}[d^{\bigstar}((w^*)^{\top} Y)] \right) \right\} =  \frac{\overline{k}_{w^*}(0)}{2}.\]

{\scshape Step 5}: Finally, we show that the assumptions of Theorem \ref{thm:main.regret.1} imply 
\[ s_{w^*}(0)>0 \] 
and
\[ \overline{k}_{w^*}(0) > \sqrt{\pi/2} \cdot \sqrt{(w^*)^\top\Sigma w^*}\cdot s_{w^*}(0).\]

First, we show that $s_{w^*}(0)>0$. The definitions of $\overline{I}(\cdot)$ and $\overline{k}_{w^*}(\cdot)$ imply 
\[  \overline{I}(\mu) \leq \overline{k}_{w^*} ((w^*)^{\top} \mu), \text{ for all }\mu\in M. \]
As $s_{w^*}(0)$ is a supergradient of $\overline{k}_{w^*}(0)$,
\[
\overline{k}_{w^*}((w^*)^{\top} \mu) \leq \overline{k}_{w^*}(0) + s_{w^*}(0) ((w^*)^{\top} \mu), \text{ for all }\mu\in M. 
\]
Step 1 showed that $\overline{I}(\mathbf{0})=\overline{k}_{w^*}(0)$. Hence, combining the above equations yields
\begin{equation} \label{pf:lem.main.regret.1}
\overline{I}(\mu) \leq \overline{I}(\mathbf{0}) + s_{w^*}(0) (w^*)^{\top} \mu, \text{ for all }\mu\in M.
\end{equation}
If $s_{w^*}(0)=0$, Equation (\ref{pf:lem.main.regret.1}) then implies $\overline{I}(\mu) \leq \overline{I}(\mathbf{0})$ for all $\mu\in M$, contradicting the assumption that there exists $\mu \in M$ such that $\overline{I}(\mu) > \overline{I}(\mathbf{0})$.

Second, since Step 1 showed that $\overline{I}(\mathbf{0})=\overline{k}_{w^*}(0)$, 
then 
\[ \overline{k}_{w^*}(0) > \sqrt{\pi/2} \cdot \sqrt{(w^*)^\top\Sigma w^*}\cdot s_{w^*}(0) \] 
holds when $\overline{I}(\mathbf{0})$ is large enough, in particular whenever
\[\overline{I}(\mathbf{0}) > \sqrt{\pi/2} \cdot \sqrt{(w^*)^\top\Sigma w^*}\cdot s_{w^*}(0).\]

\noindent {\scshape Conclusion:} Steps 1-5 imply there are infinitely many rules that solve the problem \eqref{eq:MMR.value}.

\subsubsection{Proof of Part (ii) of Theorem \ref{thm:main.regret.1}}
Consider any decision rule $d_{\text{m}}(\cdot)$ that depends on the data only as nondecreasing function of $w^\top Y$ (for some $w \neq \mathbf{0}$ that is not necessarily $w^*$) and such that $d_{\text{m}}(\cdot)\in\{0,1\}$ for all data realizations. Then we must have $d_{\text{m}}(w^\top Y)=\mathbf{1}\left\{ w^\top Y \geq c\right\}$ for some $-\infty\leq c\leq\infty$. The worst-case expected regret of such a rule satisfies
\begin{eqnarray*}
 &  & \sup_{\theta\in\Theta}\:U(\theta)\left(\mathbf{1}\{U(\theta)\geq0\}-\mathbb{E}_{m(\theta)}[d_{m}(Y)]\right) \notag\\
 & \geq & \sup_{\theta\in\Theta:m(\theta)=\mathbf{0}}\:U(\theta)\left(\mathbf{1}\{U(\theta)\geq0\}-\mathbb{E}_{m(\theta)}[d_{m}(Y)]\right) \nonumber \\
 & = & \max\left\{-\underline{I}(\mathbf{0}) \mathbb{E}_{0}[d_{m}(Y)],\overline{I}(\mathbf{0}) (1-\mathbb{E}_{0}[d_{m}(Y)])\right\} \notag \\
&\geq & \overline{I}(\mathbf{0})/2,\end{eqnarray*}
where we used that $\underline{I}(\mathbf{0})=-\overline{I}(\mathbf{0})$ by centrosymmetry, and the last inequality is strict unless $c=0$. As $\overline{I}(\mathbf{0})/2$ is the MMR value of the problem, $d_{\text{m}}^*(\cdot)$ cannot be MMR optimal if $c \neq 0$. For $w=w^*$, substantial additional algebra that we relegate to Lemma \ref{lem:erm} extends the result to $c=0$ (by showing that a first-order condition cannot hold at $(w^*)^\top Y=0$).

\subsubsection{Proof of Part (iii) of Theorem \ref{thm:main.regret.1}}

The preceding argument established the claim for rules of form $\mathbf{1}\{w^{\top}Y\geq c\}$, where $c \neq 0$. It remains to consider symmetric threshold rules $\mathbf{1}\{w^{\top}Y\geq0\}$. To this end, bound the worst-case
expected regret of such rules as follows: 
\begin{align*}
 & \sup_{\theta\in\Theta}U(\theta)\left(\mathbf{1}\{U(\theta)\geq0\}-\mathbb{E}_{m(\theta)}[\mathbf{1}\{w^{\top}Y\geq0\}]\right)\\
\geq & \sup_{\theta\in\Theta,m(\theta)=\mu,U(\theta)\geq0,\mu\in M}U(\theta)\left(\mathbf{1}\{U(\theta)\geq0\}-\mathbb{E}_{\mu}[\mathbf{1}\{w^{\top}Y\geq0\}]\right)\\
= & \sup_{\mu\in M,\overline{I}(\mu)>0}\overline{I}(\mu)\left(1-\mathbb{E}_{\mu}[\mathbf{1}\{w^{\top}Y\geq0\}]\right)\\
= & \sup_{\mu\in M,\overline{I}(\mu)>0}\overline{I}(\mu)\Phi\left(-\frac{w^{\top}\mu}{\sqrt{w^{\top}\Sigma w}}\right) \\
:= &  \sup_{\mu\in M,\overline{I}(\mu)>0}\mathrm{g}_{w}(\mu).
\end{align*}
Note that $\mathrm{g}_{w}(\mathbf{0})=\overline{I}(\mathbf{0})/2$ is the MMR value of the problem, implying that $\mu=\mathbf{0}$ attains this value under any symmetric threshold rule. For such a rule to be MMR optimal, $\mu=\mathbf{0}$ must then be a local constrained maximum point of $\mathrm{g}_{w}(\mu)$. Because $M$ contains an open set including $\mathbf{0}$ by Definition \ref{asm:1} and Assumption \ref{asm:yata.1} and since we assumed differentiability of $\overline{I}(\mu)$ at $\mathbf{0}$, this requires a first-order condition
\begin{align}
\frac{\partial\mathrm{g}_{w}(\mu)}{\partial\mu_{j}} & \mid_{\mu=\mathbf{0}}~=~\frac{1}{2}\frac{\partial\overline{I}(\mathbf{0})}{\partial\mu_{j}}-\frac{w_{j}}{\sqrt{w^{\top}\Sigma w}}\overline{I}(\mathbf{0})\phi\left(0\right) ~=~ 0, ~j=1\ldots n. \label{eq:FOC_T3}
\end{align}
To simplify expressions, change co-ordinates (if necessary) so that $w^{*}=(1,0,\ldots,0)^\top$. Because $w^*$ must fulfil \eqref{eq:FOC_T3}, we have that $\frac{\partial\overline{I}(\mathbf{0})}{\partial\mu_{j}}=0$ for $j=2,\ldots, n.$ But this, in turn, means that \eqref{eq:FOC_T3} requires $w_2=\ldots=w_n=0$. Next, noting that $w$ in a symmetric threshold rule is determined only up to scale, we restrict attention to $w_1 \in \{-1,0,1\}$. But if $w^*=(1,0,\ldots,0)$ solves \eqref{eq:FOC_T3} for $j=1$, then $-w^*$ cannot solve it because the sign change does not affect the denominator dividing $w_1$. Finally, $w=\mathbf{0}$, i.e. never adopting treatment, is excluded by part (ii) (it is the same as setting $c=\infty$ there) and is also easily seen directly to not be MMR optimal.

\subsection{Proof of Theorem \ref{thm:main.regret.2}}\label{sec:app.a.4}

{\scshape Step 1}: If  $F\circ w^* \in\tilde{\mathcal{D}}_n$  is MMR optimal, then $V(d^*_{\text{linear}}\circ w^*)\subseteq V(F\circ w^*)$. 

To see this, pick any $F\circ w^* \in\tilde{\mathcal{D}}_n$ that is MMR optimal. Then we can write
\[F\circ w^*(Y)=F((w^*)^{\top }Y)=F(\widehat{\gamma}),\] 
where $F\in \mathcal{F}$ is a symmetric and unimodal c.d.f (thus weakly increasing as well),   $\widehat{\gamma}:=(w^*)^{\top}Y\sim N(\gamma, \sigma^2)$, with $\gamma \in \Gamma_{w^*}$ defined in Step 2 of the proof for Theorem \ref{thm:main.regret.1}(i),  and $\sigma^2=(w^*)^\top\Sigma w^*$. The worst-case expected regret of rule $F\circ w^*$ is
\begin{align*}
\sup_{\theta\in\Theta}U(\theta)\left(\mathbf{1}\{U(\theta)\geq0\}-\mathbb{E}_{m(\theta)}[F((w^{*})^{\top}Y)]\right) 
=\sup_{\gamma\in\Gamma_{w^{*}}:\overline{k}_{w^{*}}(\gamma)>0}\overline{k}_{w^{*}}(\gamma)\left(1-\mathbb{E}_{\gamma}[F(\widehat{\gamma})]\right),
\end{align*}
where $\overline{k}_w^*$ is defined in \eqref{eq:isu}. Letting $g_{F}(\gamma):=\overline{k}_{w^{*}}(\gamma)\left(1-\mathbb{E}_{\gamma}[F(\widehat{\gamma})]\right)$.
Since $\widehat{\gamma}\sim N(\gamma,\sigma^{2})$, we may further calculate (using integration by parts)
\begin{align*}
\mathbb{E}_{\gamma}[F(\widehat{\gamma})] & =\int F(s)d\Phi(\frac{s-\gamma}{\sigma})\\
& =\Phi\left(\frac{s-\gamma}{\sigma}\right)F(s)|_{-\infty}^{\infty}-\int\Phi\left(\frac{s-\gamma}{\sigma}\right)dF(s)\\
& =1-\int\Phi\left(\frac{s-\gamma}{\sigma}\right)dF(s).
\end{align*}
Therefore, $g_{F}(\gamma)=\overline{k}_{w^{*}}(\gamma)\int\Phi(\frac{s-\gamma}{\sigma})dF(s)$.
Note that $F(-x)=1-F(x)$ for all $x\in \mathbb{R}$; hence, $\int\Phi(\frac{s}{\sigma})dF(s)=\frac{1}{2}$ and therefore $g_{F}(0)=\frac{\overline{k}_{w^{*}}(0)}{2}$. By Step 1 for the proof of Theorem \ref{thm:main.regret.1}(i),  $\frac{\overline{k}_{w^{*}}(0)}{2}$ is the MMR value of the problem. MMR optimality of $F\circ w^*$ implies 
\begin{equation*}
0\in\arg\sup_{\gamma\in\Gamma_{w^{*}},\overline{k}_{w^{*}}(\gamma)>0}g_{F}(\gamma).  
\end{equation*}

By Lemma \ref{lem:ex.1.4}, $0$ is an interior point of $\{\gamma \in \mathbb{R}:\overline{k}_{w^{*}}(\gamma)>0,\gamma \in\Gamma_{w^{*}}\}$. Thus, let $\partial g_{F}(0)$
denote the generalized gradient of $g_{F}(\cdot)$ at $0$.\footnote{\label{foot:general.gradient}The generalized gradient of $g:\mathbb{R}\rightarrow \mathbb{R}$ at $x$ equals $\partial g(x) := \{ \xi \in \mathbb{R}: \limsup_{y \rightarrow x,t\downarrow0}\frac{f(y+tv)-f(y)}{t} \geq \xi v,\forall v\in\mathbb{R}\}$. See \citet[][p. 27]{clarke1990optimization}.} Then $0\in \partial g_{F}(0)$ is necessary for optimality. To show that it fails, compute the generalized gradient as\footnote{We can verify, following the same steps in the proof for Lemma \ref{lem:erm}, that conditions of Proposition 2.3.13 in \cite{clarke1990optimization} are satisfied, so that the chain rule can be applied.}
\begin{align*}
& \tilde{s}_{w^*}(0)\int\Phi(\frac{s}{\sigma})dF(s)-\frac{\overline{k}_{w^*}(0)}{\sigma}\int\phi\left(\frac{s}{\sigma}\right)dF(s)\\
=& \frac{\tilde{s}_{w^*}(0)}{2}-\frac{\overline{k}_{w^*}(0)}{\sigma}\int\phi\left(\frac{s}{\sigma}\right)dF(s),
\end{align*}
where $\tilde{s}_{w^*}(0)$ is a supergradient of $ \overline{k}_{w^*}(\gamma)$ at $\gamma=0$.
Therefore, we conclude 
\[
\frac{\tilde{s}_{w^*}(0)}{2}-\frac{\overline{k}_{w^*}(0)}{\sigma}\int\phi\left(\frac{s}{\sigma}\right)dF(s)=0
\Longleftrightarrow
\int\phi\left(\frac{s}{\sigma}\right)dF(s)=\frac{\tilde{s}_{w^*}(0)\sigma}{2\overline{k}_{w^*}(0)}
\]
for some $\tilde{s}_{w^*}(0)>0$.

Next, $d_{\text{linear}}^{*}\in\mathcal{F}$ can be verified to solve the linear embedding problem (\ref{eq:linear.embedding.1}) (Lemma \ref{lem:linear.embedding}). In particular, evaluating $g_{\text{linear}}^{(1)}(\gamma)$ at $\gamma=0$, where $g_{\text{linear}}(\gamma)$ is defined in Lemma \ref{lem:linear.embedding.upper.bound} with $k=\frac{\overline{k}_{w^*}(0)}{s_{w^*}(0)}$ and $\sigma^2=(w^*)^{\top}\Sigma w^*$, one finds
\begin{equation}\label{eq:foc.linear.rule}
\int\phi\left(\frac{s}{\sigma}\right)d(d_{\text{linear}}^{*}(s))=\int_{-\rho^{*}}^{\rho^{*}}\phi\left(\frac{s}{\sigma}\right)\frac{1}{2\rho^{*}}ds=\frac{s_{w^*}(0)\sigma}{2\overline{k}_{w^*}(0)},
\end{equation}
where $s_{w^*}(0)>0$ is the largest  supergradient of $ \overline{k}_{w^*}(\gamma)$ at $\gamma=0$. 

Since $f$ is symmetric around $0$, $V(F\circ w^*)$ is as well. Write $V(F\circ w^*):=(-a_{F},a_{F})$ for some $a_{F}$ and $V(d_{\text{linear}}^{*}\circ w^*):=(-\rho^{*},\rho^{*})$. Suppose by contradiction that $V(d_{\text{linear}}^{*}\circ w^*)\nsubseteq V(F\circ w^*)$.
Then $a_{F}<\rho^{*}$, but that would imply
\begin{equation*}
\int\phi(\frac{s}{\sigma})dF(s)=\int_{-a_{F}}^{a_{F}}\phi(\frac{s}{\sigma})dF(s)>\frac{s_{w^*}(0)\sigma}{2\overline{k}_{w^*}(0)}\geq\frac{\tilde{s}_{w^*}(0)\sigma}{2\overline{k}_{w^*}(0)},\label{eq:pf.2}
\end{equation*}
where the first inequality follows by the assumption that $dF(x)=0$ for all $x\notin(-a_{F},a_{F})$
and $F(x)$ is symmetric and unimodal, and the second inequality follows as $s_{w^*}(0)$ is the largest supergradient of $\overline{k}_{w^*}(0)$. Thus, $0\notin \partial g_{F}(0)$, a contradiction.

{\scshape Step 2}: Next, $V(d_{\text{linear}}^{*}\circ w^*)=V(F\circ w^*)$ if and only if $F=d_{\text{linear}}^{*}$. The ``if'' direction is obvious. To see ``only if,'' suppose by contradiction that there exists some $\tilde{F}\in\mathcal{F}$ not equal to $d^*_{\text{linear}}$ but such that  $V(d_{\text{linear}}^{*}\circ w^*)=V(F\circ w^*)$ and $\tilde{F}\circ w^*$ is MMR optimal. Then
\begin{equation*}\int\phi\left(\frac{s}{\sigma}\right)d\tilde{F}(s)=\int_{-\rho^{*}}^{\rho^{*}}\phi\left(\frac{s}{\sigma}\right)d\tilde{F}(s)>\int_{-\rho^{*}}^{\rho^{*}}\phi\left(\frac{s}{\sigma}\right)\frac{1}{2\rho^{*}}ds=\frac{s_{w^*}(0)\sigma}{2\overline{k}_{w^*}(0)}\geq \frac{\tilde{s}_{w^*}(0)\sigma}{2\overline{k}_{w^*}(0)},
\end{equation*}
where the first step uses $V(d_{\text{linear}}^{*}\circ w^*)=V(F\circ w^*)$, the second one that $\tilde{F}$ is symmetric and unimodal, the third one uses \eqref{eq:foc.linear.rule}, and the last one that $s_{w^*}(0)$ is the largest supergradient of $\overline{k}_{w^*}(0)$. Thus, $\tilde{F}\circ w^*$ cannot be MMR optimal, a contradiction.

\subsection{Proof of Theorem \ref{thm:general}: General Theorem on Admissibility.}\label{sec:general}
In this section, we prove a general theorem regarding the admissibility
of treatment choice rules in partially identified models. Let
$\theta\in\Theta$, $U(\cdotp)$ and $m(\cdotp)$ be defined as in
Section \ref{sec:framework}. We observe a random vector $Y\in\mathcal{Y}\subseteq\mathbb{R}^{n}$
that follows a distribution $P_{m(\theta)}$, a member of the class
of distributions 
\begin{equation}
\mathcal{P}:=\left\{ P_{m(\theta)}:\theta\in\Theta\right\} \label{eq:general.distribution}
\end{equation}
defined over the common sample space $\mathcal{Y}$ (endowed with a $\sigma$-algebra $\mathcal{A}$).
\begin{thm}
\label{thm:general} Consider a treatment choice problem with payoff
function \eqref{eq:welfare} and a class of statistical models \eqref{eq:general.distribution}
that exhibit nontrivial partial identification in the sense of Definition
\ref{asm:1}. In addition, suppose the following two conditions hold
true:
\begin{itemize}
    \item[A1] The family of distributions $\mathcal{P}_{\mathcal{S}}:=\left\{ P_{m(\theta)},\theta\in m^{-1}(\mathcal{S})\right\} $
is bounded complete, i.e., for all bounded functions $f:\mathbb{R}^{n}\rightarrow\mathbb{R}$,
$\mathbb{E}_{m(\theta)}[f(Y)]=0$ for all $\theta\in m^{-1}(\mathcal{S})$
implies $f(y)=0$, a.e. $\mathcal{P}_{\mathcal{S}}$.\footnote{A statement is said to hold a.e. $\mathcal{P}$ if it holds except
on a set $N$ with $P(N)=0$ for all $P\in\mathcal{P}$.} 

\item[A2] For any set $E\in\mathcal{A}$, $P(E)=0$ for all $P\in\mathcal{\mathcal{P}_{\mathcal{S}}}$ implies $P(E)=0$ for all $P\in\mathcal{P}\setminus\mathcal{P}_{\mathcal{S}}$.
\end{itemize}
Then, every decision rule $d\in\mathcal{D}_{n}$ is (welfare-)admissible. 
\end{thm}

\begin{proof}
Suppose by contradiction that some rule $d$ is dominated. Then there
exists $d^{\prime}$ such that 
\begin{equation}
U(\theta)\mathbb{E}_{m(\theta)}\left[d^{\prime}(Y)\right]\geq U(\theta)\mathbb{E}_{m(\theta)}\left[d(Y)\right]\label{eq:thm1.1}
\end{equation}
for all $\theta\in\Theta$, with a strict inequality for some $\theta$. 

\textsc{Step 1:} We first show that \eqref{eq:thm1.1} must hold with
equality for any $\theta\in m^{-1}(\mathcal{S})$. Suppose not, then
there exists $\theta^{*}\in m^{-1}(\mathcal{S})$ such that 
\begin{equation*}
U(\theta^{*})\mathbb{E}_{m(\theta^{*})}\left[d^{\prime}(Y)\right]>U(\theta^{*})\mathbb{E}_{m(\theta^{*})}\left[d(Y)\right],\label{eq:thm1.2}
\end{equation*}
which implies that $U(\theta^{*})\neq0$. Without loss of generality,
assume $U(\theta^{*})>0$. Then 
\begin{equation}
\mathbb{E}_{m(\theta^{*})}\left[d^{\prime}(Y)\right]>\mathbb{E}_{m(\theta^{*})}\left[d(Y)\right].\label{eq:thm1.4}
\end{equation}
Define $\mu^{*}:=m(\theta^{*})$. Since $\mu^{*}\in\mathcal{S}$,
there must exist some $\theta_{\mu^{*}}\in\Theta$, such that $m(\theta_{\mu^{*}})=\mu^{*}=m(\theta^{*})$
and $U(\theta_{\mu^{*}})<0$. Therefore, (\ref{eq:thm1.4}) implies
\begin{equation*}
U(\theta_{\mu^{*}})\mathbb{E}_{m(\theta_{\mu^{*}})}\left[d^{\prime}(Y)\right]<U(\theta_{\mu^{*}})\mathbb{E}_{m(\theta_{\mu^{*}})}\left[d(Y)\right],\label{eq:thm1.3}
\end{equation*}
contradicting \eqref{eq:thm1.1}. We conclude that 
\[
U(\theta)\mathbb{E}_{m(\theta)}\left[d^{\prime}(Y)\right]=U(\theta)\mathbb{E}_{m(\theta)}\left[d(Y)\right]
\]
for any $\theta\in m^{-1}(\mathcal{S})$.

\textsc{Step 2:} We next show that for any $\mu\in\mathcal{S}$, $\mathbb{E}_{\mu}\left[d^{\prime}(Y)-d(Y)\right]=0.$
Because of nontrivial partial identification, for any $\mu\in\mathcal{S}$
there exists $\theta_{\mu}$ such that $m(\theta_{\mu})=\mu$ and $U(\theta_{\mu})\neq0$. Step
1 then implies that for any $\mu\in\mathcal{S}$, 
\[
U(\theta_{\mu})\mathbb{E}_{\mu}\left[d^{\prime}(Y)\right]=U(\theta_{\mu})\mathbb{E}_{\mu}\left[d(Y)\right].
\]
Since $U(\theta_{\mu})\neq0$, the desired result follows.

\textsc{Step 3:} The conclusion from step 2 implies that $\mathbb{E}_{m(\theta)}\left[d^{\prime}(Y)-d(Y)\right]=0$
for any $\theta\in m^{-1}(\mathcal{S})$. Consider now the family
of distributions 
\[
Y\sim\mathcal{P}_{\mathcal{S}}=\left\{ P_{m(\theta)},\theta\in m^{-1}(\mathcal{S})\right\} ,
\]
which is bounded complete under A1.  Let $E:=\left\{ y\in\mathcal{Y}:d^{\prime}(y)-d(y)\neq0\right\} $.
The definition of bounded completeness implies that $P(E)=0$ for
all $P\in\mathcal{P}_{\mathcal{S}}$. It follows by A2 that $P(E)=0$
for all $P\in\mathcal{P}\setminus\mathcal{P}_{\mathcal{S}}$. Therefore,
$P(E)=0$ for all $P\in\mathcal{P}$.

{\scshape Conclusion}: We find that $\mathbb{E}_{m(\theta)}[d^{\prime}(Y)]=\mathbb{E}_{m(\theta)}[d(Y)]$
for all $\theta\in\Theta$. This means that
$d^{\prime}$ cannot dominate $d$. A contradiction. 
\end{proof} 

\section{Profiled Risk}\label{sec:profile.regret}

Given the  findings in Section \ref{sec:challenges}, in this section, we show that profiling out some parameters of expected welfare or regret may yield interesting insights. We illustrate this idea by exploring a \emph{profiled expected regret} criterion. Beyond allowing to better visualize risk profiles of decision rules, there are several additional values the profiling approach can provide in our problem. First, as we have shown in Figure \ref{figure:NewRule}, two MMR optimal rules may display drastically different regret profiles. In fact, as we show in Proposition \ref{prop:linear.vs.RT}, one dominates the other in terms of profiled regret for many parameter values.  Therefore, the notion of profiled risk may help decision makers better distinguish two rules that cannot otherwise be  ranked according to the MMR criterion alone. Second, as we will show below, the profiled approach may render some rules, e.g., the no-data rule $d_{\text{coin-flip}}$, inadmissible (in a redefined sense, and see Proposition \ref{prop:profiled.regret.dominance} for a general result).\footnote{Although $d_{\text{coin-flip}}$ may also be ruled out by the MMR criterion, being inadmissible is a stronger statement than not being MMR optimal. For example, it would be unclear whether $d_{\text{coin-flip}}$ can be optimal for a decision maker with a robust Bayes criterion for a class of priors advocated by \cite{GiacominiKitagawa}. As a Bayes rule with respect to our profiled risk may be viewed as a robust Bayes criterion  under some conditions, $d_{\text{coin-flip}}$ being inadmissible under the profiled risk criterion implies that it also cannot be Robust Bayes optimal for a large class of priors in \cite{GiacominiKitagawa}'s framework (see Remark \ref{rem:robust.bayes} below for related discussions).} While we could profile out any known function $h(m(\theta))$, we simplify exposition by restricting attention to linear indices.\footnote{Alternatively, one may take $h(\cdot)$ as an identity function to profile out with respect to the entire reduced form parameter $m(\theta)$.   } Thus, for a vector $w \in \mathbb{R}^{n} \backslash \{ \mathbf{0} \}$, let  \[ \Gamma_{w}:=\left\{ \gamma\in\mathbb{R}:w^{\top} m(\theta) =\gamma, \theta\in\Theta\right\}\] 
be the image of the transformation $\theta \mapsto w^{\top} m(\theta)$.  Then the worst-case expected regret of a rule $d$ can be expressed as
\begin{equation*} \label{eq:regret.outer.inner}
\sup_{\theta \in \Theta} R(d,\theta)=\sup_{\gamma\in\Gamma_{w{}}}\left(\sup_{ \theta \in \Theta : w^{\top} m(\theta) = \gamma } R(d,\theta)\right).
\end{equation*}
That is, we split \eqref{eq:expected.regret} into an inner optimization problem with fixed level set $\{w^\top m(\theta) = \gamma\}$ and an outer optimization over $\gamma$. The value function of the inner problem may be of interest, thus define:

\begin{defn}[$w$-Profiled Regret]\label{def:profile.regret}
The $w$-profiled regret function $\overline{R}_{w}:\mathcal{D}_n \times \Gamma_{w} \rightarrow \mathbb{R}$ is given by
\begin{equation}\label{eq:profile.regret}
\overline{R}_{w}(d,\gamma):=\sup_{ \theta \in \Theta: w^{\top} m(\theta) = \gamma } U(\theta)\left(\mathbf{1}\{U (\theta)\geq0\}-\mathbb{E}_{m(\theta)}[d(Y)]\right), 
\end{equation}
where $Y \sim N(m(\theta),\Sigma)$.
\end{defn}

In problems where MMR optimal rules depend on the data only through some linear combination $(w^*)^{\top} Y$, it seems reasonable to set $w$ equal to $w^*$. However, there could be other vectors $w$ of interest. For example, when extrapolating local average treatment effects in Section \ref{sec:iv.example}, Theorem \ref{thm:main.regret.1} does not apply and there are multiple reasonable alternatives for parameters to be profiled out. One may use 
\[ w =  (1, -\beta_0)^{\top} / \sqrt{(1, -\beta_0)\Sigma (1,-\beta_0)^{\top}} \]
for some $\beta_0 \in \mathbb{R}$. For motivation, note that the square of
\[ (m_1(\theta) -\beta_0 m_2(\theta) ) / \sqrt{(1, -\beta_0)\Sigma (1,-\beta_0)^{\top}}    \]
can be viewed as the population \cite{anderson49} statistic for the null hypothesis of the local average treat effect equal to $\beta_0$. Thus, the profiled regret function reports the worst-case regret as one keeps constant the population analog of that statistic.  One may also profile out the intention-to-treat effect $m_1(\theta)$, the
local average treatment effect $m_1(\theta)/m_2(\theta)$, or even the entire reduced-form parameter vector $(m_1(\theta),m_2(\theta))^\top$. In this paper, we do not argue for one specific criterion for profiling out, but believe researchers should have the freedom to choose the device judiciously depending on the specific contexts of their problems. 

\begin{rem}[Not all decision rules are admissible with respect to $w$-profiled regret] One could say that a decision rule $d$ is \emph{$w$-profiled regret admissible} if there is no other rule $d'$ for which 
\[ \overline{R}_{w}(d',\gamma) \leq \overline{R}_{w}(d,\gamma)  \]
for every $\gamma \in \Gamma_{w}$, with strict inequality for some $\gamma \in \Gamma_{w}$. Even if Theorem \ref{thm:admin} applies, a decision rule can fail to be $w$-profiled regret admissible for some $w$. Indeed, the no-data rule $d_{\text{coin-flip}}(Y)=1/2$ is $w^*$-profiled regret dominated by any MMR optimal rule in our running example.\footnote{Simple algebra shows that that the profiled regret of $d_{\text{coin-flip}}$ is \emph{minimized} at $\gamma=0$, where it coincides with the \emph{maximized} regret of any MMR optimal rule. This dominance can also be verified if $C \left\Vert x_1 - x_0 \right\Vert \leq \sqrt{\pi / 2} \cdot \sigma_1$. See Figure \ref{figure:NewRule}.} Moreover, the following proposition shows that, in the running example, $d^*_{\text{linear}}$ indeed $w^*$-profiled regret dominates $d^*_{\text{RT}}$ for a large range of parameter values, including those used in Figure \ref{figure:NewRule}. Let $\phi(\cdot)$ denote the p.d.f. of a standard normal random variable, and let $0 < \rho^* \leq C \Vert x_1-x_0 \Vert$ be the threshold after which $d^*_{\text{linear}}$ implements the policy with probability one. Let $\gamma$ be the effect of the policy of interest in the nearest neighbor to Country 0 (which, by assumption, we have set to be Country 1). Let $\overline{R}_{w
^*}(d,\gamma)$ be the profiled-regret associated to the nearest-neighbor weights; that is, the largest value of expected regret fo decision rule $d$ for a fixed value of $\gamma$.  

\begin{prop}\label{prop:linear.vs.RT}
Suppose $C\left\Vert x_{1}-x_{0}\right\Vert >\sqrt{\frac{\pi}{2}}\sigma_{1}$.

\begin{itemize}
\item[(i)] If $\phi\left(\frac{\rho^{*}}{\sigma_{1}}\right)\left(\frac{C\left\Vert x_{1}-x_{0}\right\Vert }{\sqrt{\frac{\pi}{2}} \sigma_1}\right)^{3}\leq\phi(0)$,
then $\overline{R}_{w^{*}}(d_{\text{linear}}^{*},\gamma)\leq\overline{R}_{w^{*}}(d_{\text{RT}}^{*},\gamma)$
for all $\gamma$ with the inequality strict for $\gamma\neq0$. 
\item[(ii)] If $\phi\left(\frac{\rho^{*}}{\sigma_{1}}\right)\left(\frac{C\left\Vert x_{1}-x_{0}\right\Vert }{\sqrt{\frac{\pi}{2}} \sigma_1}\right)^{3}> \phi(0)$,
then $\overline{R}_{w^{*}}(d_{\text{linear}}^{*},\gamma)<\overline{R}_{w^{*}}(d_{\text{RT}}^{*},\gamma)$
for all $\left|\gamma\right|>\underline{\gamma}$, where $\underline{\gamma}>0$
is unique in $(0,\infty)$ such that $\Phi\left(-\frac{\underline{\gamma}}{\frac{C\left\Vert x_{1}-x_{0}\right\Vert}{\sqrt{\frac{\pi}{2}}}}\right)=\int_{0}^{1}\Phi\left(\frac{2\rho^{*}x-\rho^{*}-\underline{\gamma}}{\sigma}\right)dx$.
\end{itemize}
\end{prop}

\begin{proof}
See Online Appendix \ref{sec:proof.prop.linear.vs.RT}.
\end{proof}

Part i) of Proposition \ref{prop:linear.vs.RT} gives conditions on the Lipschitz constant ($C$), the variance of the estimated effect in Country 1 ($\sigma^2_1$), and the difference between the covariates between Countries 1 and 0 $(\|x_1-x_0\|)$ under which $d^*_{\textrm{linear}}$ $w^*$-profiled regret dominates $d^*_{\textrm{RT}}$. Part ii) of Proposition \ref{prop:linear.vs.RT} shows that for those parameter configurations for which we cannot establish $w^*$-profiled regret dominance (which we later show can only happen when $C\left\Vert x_{1}-x_{0}\right\Vert\downarrow\sigma_1$), the decision rule $d^*_{\text{linear}}$ still has smaller $w^*$-profiled regrets than $d^*_{\text{RT}}$ for all $\left|\gamma\right|$ large enough. Therefore, our new result shows there is a strong sense that $d^*_{\text{linear}}$ performs better than $d^*_{\text{RT}}$ under the $w^*$-profiled regret criterion, despite the fact that both of them are equivalent under the standard minimax regret criterion. Finally, we also argue that Proposition  \ref{prop:linear.vs.RT} will continue to hold (with suitably redefined parameter values) in a model as general as those studied in Theorem \ref{thm:main.regret.1}, as long as the upper and lower bounds of the identified set of $U(\theta)$ given $\gamma \in \Gamma_{w^*}=\left\{ \gamma\in\mathbb{R}|(w^*)^{\top} m(\theta) =\gamma,\theta\in\Theta\right\}$ are affine  with the same positive gradients. Under some conditions, a sufficient condition for this to happen is for $\gamma$ to be unbounded.

Below, we also present a general result on the profiled-regret inadmissibility of  $d_{\text{coin-flip}}$ in a large class of general models when we profile out with respect to the reduced-form parameter $m(\theta)$.\footnote{We thank a referee for suggesting this result.} Just as before, let $\overline{I}(\mu)$ and $\underline{I}(\mu)$ be the upper and lower bounds of the identified set for the welfare contrast given the reduced-form parameter $\mu$. 

\begin{prop}\label{prop:profiled.regret.dominance}
In the setup of Section \ref{sec:minimax.regret.result}, let the profiled regret (with respect to
the reduced-form parameter $\mu\in M$) of a rule $d$ be
\[
\overline{R}_{r}(d,\mu):=\sup_{\theta\in\Theta,m(\theta)=\mu}R(d,\theta).
\]
Suppose all conditions of Theorem \ref{thm:main.regret.1} hold true, and $\overline{I}(\mathbf{0})\leq\max\left\{ \overline{I}(\mu),-\underline{I}(\mu)\right\} $
for all $\mu\in M$ with the inequality strict for some $\mu\in M$.
Then, $d_{\text{coin-flip}}$ is dominated in terms of
$\overline{R}_{r}(d,\mu)$. 
\end{prop}

\begin{proof}
See Online Appendix \ref{sec:proof.profiled.regret.dominance}.
\end{proof}

One important remark about Proposition \ref{prop:profiled.regret.dominance} is that, while it is possible to rule out $d_{\text{coin-flip}}$ using the minimax-regret criterion, such a rule could still be optimal under some other criteria such as Robust Bayes (for some class of priors). An implication of Proposition \ref{prop:profiled.regret.dominance} is that $d_{\text{coin-flip}}$ could never be Robust Bayes optimal for a class of priors where i) the prior over the reduced-form parameter is fixed, and ii) the priors of the partially identified parameters are arbitrary.   \end{rem}

\begin{rem}[Bayes rules with respect to $w$-profiled regret are $\Pi^*$-minimax under some conditions]\label{rem:robust.bayes} Consider any decision rule that minimizes 
\begin{equation}\label{eqn:Bayes-regret}
\inf_{d \in \mathcal{D}_n} \int_{\Gamma_{w}} \overline{R}_{w}(d,\gamma) d \pi^*(\gamma),     
\end{equation}
the weighted average of $w$-profiled regret for some prior $\pi^*$ over $\Gamma_{w}$. If $\pi^*$ has full support and profiled regret is continuous in $\gamma$ (for any $d$), then any solution to \eqref{eqn:Bayes-regret} is $w$-profiled-regret admissible; see \citet[][Theorem 3, Section 2, p. 62]{Ferguson67} and \citet[][Theorem 9, Section 4, p. 254]{berger1985statistical}.   

We next provide sufficient conditions under which decision rules that solve \eqref{eqn:Bayes-regret} can be interpreted as $\Pi^*$-minimax decision rules in the sense of \citet[][Definition 13, p. 216]{berger1985statistical}, where $\Pi^*$ is a class of priors over $\Theta$.\footnote{To avoid notational conflict, we write $\Pi^*$-minimax instead of the more common $\Gamma$-minimax.} Let $\mathcal{D}_{w}$ denote the class of decision rules that depend on the data only through $w^{\top} Y \sim N(w^{\top}m(\theta), w^{\top}\Sigma w)$. Consider the $\Pi^*$-minimax problem
\begin{equation}\label{eqn:Gamma-minimax}
\inf_{d \in \mathcal{D}_w} \left( \sup_{\pi \in \Pi^*}\int_{\Theta} R(d,\theta) d \pi(\theta) \right).      
\end{equation}
Let $\Pi^*$ collect all priors over $\Theta$ for which $w^{\top} m(\theta) \sim \pi^*$, where $\pi^*$ is the prior over $\Gamma_{w}$ used in \eqref{eqn:Bayes-regret}. This class of priors has recently been advocated by \citet{GiacominiKitagawa}. Their Theorem 2 establishes that, for any $d \in \mathcal{D}_{w}$, 
\begin{equation*} \label{eq:worst-expected-regret}
\sup_{\pi \in \Pi^*}\int_{\Theta} R(d,\theta) d \pi(\theta) = \int_{\Gamma_{w}} \overline{R}_{w}(d,\gamma) d \pi^*(\gamma). 
\end{equation*}
Thus, the problem in \eqref{eqn:Gamma-minimax} is equivalent to minimizing average $w$-profiled regret over decision rules that depend on the data only through $w^{\top} Y$. 
\end{rem}

Second, computation of profiled regret can frequently be simplified. Algebra shows that $\overline{R}_{w}(d,\gamma)$ equals the maximum between 

\begin{equation*}\label{eq:auxiliary.profiled.regret.1}
\overline{k}^{+}_{w} (\gamma) : = \sup_{\theta \in \Theta} \: U(\theta) (1-\mathbb{E}_{m(\theta)} [d(Y)] ) \quad \textrm{ s.t. } w^{\top} m(\theta) = \gamma, \quad U(\theta) \geq 0,  
\end{equation*}
and
\begin{equation*}\label{auxiliary.profiled.regret.2}
\overline{k}^{-}_{w} (\gamma) := \sup_{\theta \in \Theta} \: -U(\theta)  \mathbb{E}_{m(\theta)} [d(Y)] \quad \quad \textrm{ s.t. } w^{\top} m(\theta) = \gamma, \quad U(\theta) \leq 0.   
\end{equation*}
In principle, these are the value functions of two infinite-dimensional, nonlinear optimization problems. However, they can be recast as finite dimensional. For example, $\overline{k}^{+}_{w} (\gamma)$ equals
\begin{equation}\label{eq:auxiliary.profiled.regret.3}
\overline{I}^{+}_{w}(\gamma) : = \sup_{\mu \in M} \: \overline{I}(\mu)(1-\mathbb{E}_\mu[d(Y)])) \quad \textrm{ s.t. } w^{\top} \mu = \gamma, \quad \overline{I}(\mu) \geq 0. 
\end{equation}

This problem has a scalar choice variable, one linear equality constraint, and one potentially nonlinear inequality constraint. The bottleneck is evaluation of $\overline{I}(\mu)$. The running example admits a closed-form solution for $\overline{I}(\mu)$, so that evaluating \eqref{eq:auxiliary.profiled.regret.3} is easy. 
More generally, the computational cost of evaluating \eqref{eq:auxiliary.profiled.regret.3} can be reduced by imposing more structure on the parameter space $\Theta$. For instance, when $\Theta$ is convex, the set $M$ is convex as well; the optimization problem is then over a convex subset of $\mathbb{R}^{n}$. Moreover, if $m(\cdot)$ is linear and $U(\cdot)$ is concave, the function $\overline{I}(\mu)$ can be shown to be concave. This means that under Assumption \ref{asm:yata.1}, the optimization problem in \eqref{eq:auxiliary.profiled.regret.3} is convex.

\newpage

\section{Online Appendix }\label{sec:technical}

\subsection{Aversion to Randomization in the Loss Function}\label{sec:aversion.randomization}
In this section, we present some thoughts on explicitly modeling aversion to random treatment assignment.\footnote{We thank the referees for raising this question.} To do so, we distinguish between two interpretations of randomization: (1) as \emph{fractional assignment} of a new policy or treatment, wherein a fraction $a \in [0,1]$ of the population of interest gets treated (see \cite{manski2007admissible}); (2) as randomization over non-fractional assignments, i.e. the decision maker treats either everybody or nobody but randomizes this choice using probability $a \in [0,1]$. As a result, $d^*_{\text{linear}}$ and $d^*_{\text{RT}}$ can be interpreted as i) \emph{non-randomized fractional rules}, where the action space is the unit interval (i.e., $a\in [0,1]$), or as ii) randomized rules when there are only two actions $(a \in \{0,1\})$. Since the loss function is usually defined over actions, in case i) one can use it to capture what we will term \emph{aversion to fractional assignment}. We will first present a simple (and practical) framework to do so and then argue that there are a number of important caveats to such an analysis.

\subsubsection*{Incorporating Aversion to Fractional Assignment}

Consider then a treatment choice problem where an action $a\in[0,1]$ is interpreted as a \emph{non-randomized fractional assignment}. We would like to posit a loss function that captures \emph{aversion to fractional assignment} which could arise, for example, due to concerns for ex-post equity, or simply inconvenience cost. To do so, we first introduce a cost function $c(a):[0,1]\rightarrow\mathbb{R}^{+}$ that penalizes fractional actions $a\in(0,1)$. More precisely, we assume that $c(a)=0$ for
$a\in\left\{ 0,1\right\}$, but that $c(a) \geq 0$ for all $0<a<1$, with at least some strict inequality (we provide examples below). 

Assuming the policy maker incorporates the cost linearly in his/her welfare calculations, we define the \emph{net-of-cost welfare} of action $a\in[0,1]$ as
\[
W^{c}(a,\theta):=aW(1,\theta)+(1-a)W(0,\theta)-c(a).
\]
Since cost is only positive when $a\in (0,1)$, the oracle action for the policy maker based on the modified welfare is still to expose everyone to the policy ($a=1$) if the sign of the welfare contrast is nonnegative, and to preserve the status quo ($a=0$) otherwise; that is, $\mathbf{1}\left\{ U(\theta)\geq0\right\}$,
where the welfare contrast is defined as $U(\theta)=W(1,\theta)-W(0,\theta)$.  

Just as before, we also define a regret loss based on the net-of-cost welfare as follows: 
\begin{align*}
L^{c}(a,\theta) & := \max_{a \in [0,1]} W^{c}(a,\theta) -W^{c}(a,\theta)\\
& = \max\{W(1,\theta),W(0,\theta)\} - W^{c}(a,\theta)\\
& =L(a,\theta)+c(a),
\end{align*}
where $L(a,\theta)=U(\theta)(\mathbf{1}\left\{U(\theta)\geq0\right\}-a)$ is the usual regret loss for our original problem. 

Finally, interpreting $\mathcal{D}_n$ explicitly as the class of fractional assignment rules, the \emph{net-of-cost} risk function of a fractional rule $d\in\mathcal{D}_n$ is modified as
\begin{align*}
R^{c}(d,\theta)	&:=\mathbb{E}_{m(\theta)}\left[L^{c}(d(Y),\theta)\right]\\	&=R(d,\theta)+\mathbb{E}_{m(\theta)}\left[c(d(Y))\right],
\end{align*}
where $R(d,\theta)$ is the standard expected regret without
aversion to fractional treatment defined in \eqref{eq:expected.regret}. This shows that, under our assumptions, \emph{aversion to fractional assignment} can be captured by a penalized regret function, where the penalty for fractional assignment is captured by the second term in the net-of-cost risk function, $R^{c}(d,\theta)$.

In order to illustrate the effects of the penalty term in the net-of-cost risk function, $R^{c}(d,\theta)$, it is helpful to consider our running example and compare $d_{\text{RT}}^{*},d_{\text{linear}}^{*},d^*_{0}$ for different penalty functions.  

To visualize the risk comparisons, we once again rely on our $w^{*}$-profiled regret device discussed
in Section \ref{sec:profile.regret}. More specifically, with the modified risk function $R^{c}(d,\theta)$,
the $w^{*}$-profiled regret of a decision rule $d$ 
is written as 
\begin{align*}
\overline{R}_{w^{*}}^{c}(d,\gamma)&:=\sup_{\theta\in\Theta\text{ s.t. }m_{1}(\theta)=\gamma} R^c(d,\theta) \\&=\overline{R}_{w^{*}}(d,\gamma)+\mathbb{E}_{\gamma}\left[c(d(Y_{1}))\right],
\end{align*}
where $\overline{R}_{w^{*}}(d,\gamma)$ is the original profiled regret
in Definition \ref{def:profile.regret}, and $\mathbb{E}_{\gamma}\left[c(d(Y_{1}))\right]$
can be calculated either numerically or analytically once the functional
form of $c(\cdot)$ is given. Note as $d^*_{0}$ is non-fractional, $\mathbb{E}_{\gamma}\left[c(d^*_{0}(Y_{1}))\right]=0$. 

In Figure \ref{fig:profiled.regret.cost}, we present the $w^*$-profiled regret of $d_{\text{RT}}^{*}$, $d_{\text{linear}}^{*}$,
$d^*_{0}$ with three different types of cost functions: 

\begin{enumerate}
\item [(a)] \emph{Linear cost:} $c(a)=c\left(a\mathbf{1}\left\{ a\leq0.5\right\} +\left(1-a\right)\mathbf{1}\left\{ a>0.5\right\} \right),$
\item [(b)] \emph{Quadratic cost:} $c(a)=ca(1-a),$
\item [(c)] \emph{Constant cost:} $c(a)=c\mathbf{1}\left\{0<a<1\right\}$.
\end{enumerate}

Figure \ref{fig:profiled.regret.cost} shows that our least randomizing MMR rule, $d^*_{\textrm{linear}}$ has better profiled regret than $d^*_{\text{RT}}$ in all cost schemes considered. We find this interesting because the cost function captures both i) how frequently the rules yield fractional assignment, and also ii) how this fractional assignment is chosen. Since the least randomizing MMR rule is only concerned with i), it is unlikely to be optimal for \emph{every} cost function. In fact, we see that even when the cost function is constant (so that only i) affects the profiled regret), and the cost is large, the decision rule $d^*_0$ has better profiled regret than $d^*_{\textrm{linear}}$. However, we next present conditions under which the least randomizing decision rule is \emph{net-of-cost minimax regret optimal}.

\begin{figure}[h!]
\centering
\includegraphics[width=1\linewidth]{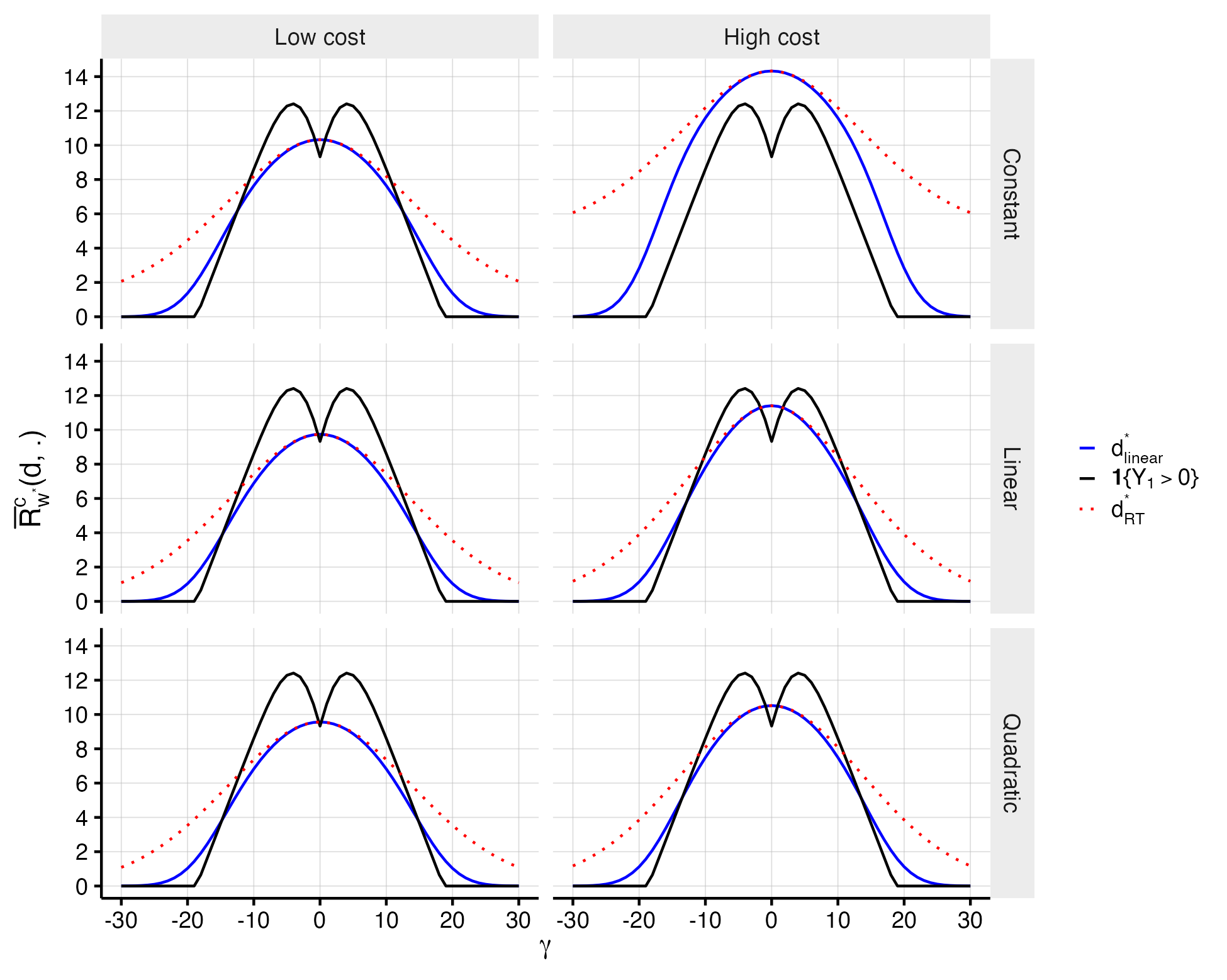}
\caption{$w^*$-profiled regret of $d^*_0,d^*_{\text{RT}}$ and $d^*_{\text{linear}}$. Upper left: constant cost function with a small cost ($c=1$); upper right: constant cost function with a large cost ($c=5$). Middle left: linear cost function with a small cost ($c=1$); middle right: linear cost function with a large cost ($c=5$). Down left: quadratic cost function with a small cost ($c=1$); down right: quadratic cost function with a large cost ($c=5$).}
    \label{fig:profiled.regret.cost}
\end{figure}

Let $\mathcal{\overline{D}}_{n}:=\mathcal{D}_{n,1}\cup\mathcal{D}_{n,2}$,
where 
\begin{align*}
\mathcal{D}_{n,1} & :=\left\{ d\in\tilde{\mathcal{D}}_{n}:\sup_{\theta}R(d,\theta)=\inf_{\tilde{d}\in\mathcal{D}_{n}}\sup_{\theta}R(\tilde{d},\theta)\right\} ,\\
\mathcal{D}_{n,2} & :=\left\{ d\in\mathcal{D}_{n}:d(Y)=\mathbf{1}\left\{ \left(w^{*}\right)^{\top}Y\geq t\right\} ,t\in\mathbb{R}\bigcup\left\{ -\infty,\infty\right\} \right\} .
\end{align*}
In words, $\mathcal{D}_{n,1}$ is the class of MMR optimal rules (for
the original no-cost risk function) that are in $\tilde{\mathcal{D}}_{n}$
defined in Section \ref{sec:least.randomizing}; $\mathcal{D}_{n,2}$ is the class of threshold
rules that depend on data only via $\left(w^{*}\right)^{\top}Y$.
We show that, in a framework as general as that in Theorem \ref{thm:main.regret.1},  $d_{\text{linear}}^{*}$ remains minimax optimal within
$\mathcal{\overline{D}}_{n}$ when we consider a constant cost function
with $c$ sufficiently small. Write $P_{0}\left\{ (-\rho^{*},\rho^{*})\right\} :=Pr\left\{ \left(w^{*}\right)^{\top}Y\in(-\rho^{*},\rho^{*})\right\} $
when $\left(w^{*}\right)^{\top}Y\sim N(0,\left(w^{*}\right)^{\top}\Sigma w^{*})$, then we have:
\begin{prop}\label{prop:aversion.fraction}
Under the conditions of Theorem \ref{thm:main.regret.1}, consider the net-of-cost risk function
$R^c$ with $c(a)=c\mathbf{1}\left\{ 0<a<1\right\} $. Let
\[
\mathrm{g}_{w}(\mu)=\overline{I}(\mu)\Phi\left(-\frac{w^{\top}\mu}{\sqrt{w^{\top}\Sigma w}}\right).
\]
If 
\[
c\leq\frac{\sup_{\mu\in M,\overline{I}(\mu)>0}\mathrm{g}_{w^{*}}(\mu)-\overline{I}(0)/2}{P_{0}\left\{ (-\rho^{*},\rho^{*})\right\} },
\]
then 
\[
\sup_{\theta}R^{c}(d_{\text{linear}}^{*},\theta)=\inf_{d\in\mathcal{\overline{D}}_{n}}\sup_{\theta}R^{c}(d,\theta).
\]
\end{prop}

\begin{proof}
See Online Appendix \ref{sec:proof.aversion.fraction}.
\end{proof}

\subsubsection*{Caveats}

There are several important caveats to the framework we just introduced. 
First, in statistical decision theory and game theory, it is commonly assumed that the decision masker can randomize in the sense of ii) above, and such randomization is usually compounded with other, e.g. sampling, uncertainty. That is, the decision maker is not averse to randomization. For example, this assumption underlies mixed equilibria on game theory and maximin theorems. But if we allowed that here, the decision maker could avoid any cost of fractional assignment by using the decision rules from this paper's main body and just interpreting $a \in (0,1)$ as randomization.

In addition, while we are currently unable to characterize the exact unconstrained MMR solution for net-of-cost risk $R^c$, we conjecture that this characterization is only interesting if we add further restrictions to the class of fractional decision rules (as we did in Proposition \ref{prop:aversion.fraction}). This is because of the game theoretic ``purification'' arguments also alluded to in Remark \ref{rem:lexicographic}: Take any MMR solution $d^*$ to our original problem with risk function $R(d^*,\theta)$, then results in \citet{dww} suggest that, for any $\epsilon>0$, we can find a decision rule $d_{\epsilon}$ that only recommends allocations in $\{0,1\}$ and has a risk function at most $\epsilon$ away from the risk function of $d^*$, uniformly over the parameter space. Consequently, $d_{\epsilon}$ is a fractional allocation rule that only recommends actions in $\{0,1\}$. This means that this rule will approximately solve the MMR problem with net-of-cost risk function $R^c$. We currently investigate exact  conditions under which this holds. However, we know from results in the main paper that rules generated in this manner must violate monotonicity with respect to the linear index that $d^*$ depends on.

In summary, aversion to fractional assignment may be the simplest framework (within standard statistical decision theory) to formalize preferences against randomization. However, we suspect that this framework is not truly interesting without additional restrictions on the class of fractional decision rules. There may still be some gap between our framework and a true working definition of \emph{aversion to randomization}. This could be an interesting avenue for future research.

\subsection{Lemmas for Theorem \ref{thm:main.regret.1}}

\begin{lem} \label{lem:yata.lower.bound}
Consider a treatment choice problem with payoff function \eqref{eq:welfare} and statistical model \eqref{eq:normal_model} that exhibits nontrivial partial identification in the sense of Definition \ref{asm:1}. Suppose that Assumption \ref{asm:yata.1} holds and there exists an MMR optimal rule $d^*$ depending on the data only through $(w^{^*})^\top Y$. If Equation \eqref{eq:yata.property.1} holds, then
\[ \textbf{R}:= \sup_{\theta \in \Theta} U(\theta) \left(  \mathbf{1}\{U(\theta) \geq 0\} - \mathbb{E}_{m(\theta)}[d^*((w^*)^\top Y)] \right) \geq \frac{\overline{k}_{w^*}(0)}{2},\]
where
\[ \overline{k}_{w^*}(0) := \sup_{\theta \in \Theta} U(\theta) \quad \textrm{s.t }  \quad (w^*)^{\top} m(\theta) = 0. \]
\end{lem}

\begin{proof}
Since the distribution of ${w^*}^{\top}Y$ only depends on $m(\theta)$ through $(w^*)^{\top} m(\theta)$, we can write
\begin{align}
\textbf{R} &= \sup_{\theta\in \Theta} U(\theta) \left(  \mathbf{1}\{U(\theta) \geq 0\} - \mathbb{E}_{(w^*)^\top m(\theta)}[d^*((w^*)^{\top}Y)] \right) \notag \\
 &\geq \sup_{\theta\in \Theta:(w^*)^\top m(\theta)=0} U(\theta) \left(  \mathbf{1}\{U(\theta) \geq 0\} - \mathbb{E}_{0}[d^*((w^*)^{\top}Y)] \right) \notag \\  
& \overset{(1)}{=} \sup_{\theta\in \Theta,(w^*)^\top m(\theta)=0} U(\theta) \left(  \mathbf{1}\{U(\theta) \geq 0\} - \frac{1}{2} \right) \notag \\
& = \frac{1}{2} \max \left\{  \sup_{\theta\in \Theta:(w^*)^\top m(\theta)=0, U(\theta) \geq 0} U(\theta) \: , \:  \sup_{\theta\in \Theta:(w^*)^\top m(\theta)=0, U(\theta) \leq 0} -U(\theta)\right\} \notag  \\
&\overset{(2)}{=} \frac{1}{2} \max \left \{  \sup_{\theta\in \Theta:(w^*)^\top m(\theta)=0, U(\theta) \geq 0 } U(\theta) \: , \: \sup_{\theta\in \Theta:(w^*)^\top m(\theta)=0, U(-\theta) \geq 0} U(-\theta)  \right\} \notag \\
&\overset{(3)}{=} \frac{1}{2}  \sup_{\theta\in \Theta:(w^*)^\top m(\theta)=0,  } U(\theta), \notag
\end{align}
where (1) uses that $\mathbb{E}_{0}[d^*((w^*)^{\top}Y)]=1/2$ by \eqref{eq:yata.property.1}, (2) uses linearity of $U(\cdot)$, and (3) uses centrosymmetry of $\Theta$ and linearity of $m(\cdot)$ and $U(\cdot)$.
\end{proof}

\begin{lem}\label{lem:ex.1.3}
Suppose Assumption \ref{asm:yata.1} holds. Then $\overline{k}_{w^*}(\gamma)$ is concave in $\gamma\in  \Gamma_{w^*}$.
\end{lem}

\begin{proof}
Fix any $\gamma_{1},\gamma_{2}\in\Gamma_{w^{*}}$ and $\varepsilon>0$. By the definition of $\overline{k}_{w^{*}}(\gamma_{1})$ and $\overline{k}_{w^{*}}(\gamma_{2})$,
there exist $\theta_{\gamma_{1},\varepsilon},\theta_{\gamma_{2},\varepsilon}\in\Theta$
such that 
\begin{align*}
(w^{*})^{\top}m(\theta_{\gamma_{1},\varepsilon}) & =\gamma_{1},\text{ }U(\theta_{\gamma_{1},\varepsilon})>\overline{k}_{w^{*}}(\gamma_{1})-\varepsilon,\\
(w^{*})^{\top}m(\theta_{\gamma_{2},\varepsilon}) & =\gamma_{2},\text{ }U(\theta_{\gamma_{2},\varepsilon})>\overline{k}_{w^{*}}(\gamma_{2})-\varepsilon,
\end{align*}
By convexity of $\Theta$, $t\theta_{\gamma_{1},\varepsilon}+(1-t)\theta_{\gamma_{2},\varepsilon}\in\Theta$
as well. Moreover, as $m(\cdot)$ is linear, 
\begin{align*}
(w^{*})^{\top}m\left(t\theta_{\gamma_{1},\varepsilon}+(1-t)\theta_{\gamma_{2},\varepsilon}\right) & =t(w^{*})^{\top}m(\theta_{\gamma_{1},\varepsilon})+(1-t)(w^{*})^{\top}m(\theta_{\gamma_{2},\varepsilon})\\
 & =t\gamma_{1}+(1-t)\gamma_{2}.
\end{align*}
Next, for $t\in[0,1]$, let $\tilde{\gamma}_{t}:=t\gamma_{1}+(1-t)\gamma_{2}$. Then 
\begin{align*}
\overline{k}_{w^{*}}(\tilde{\gamma}_{t}) & =\sup_{(w^{*})^{\top}m(\theta)=t\gamma_{1}+(1-t)\gamma_{2},\theta\in\Theta}U(\theta)\\
& \geq U(t\theta_{\gamma_{1},\varepsilon}+(1-t)\theta_{\gamma_{2},\varepsilon})\\
& =t U(\theta_{\gamma_{1},\varepsilon})+(1-t) U(\theta_{\gamma_{2},\varepsilon})\\
& >t\left(\overline{k}_{w^*}(\gamma_{1})-\varepsilon\right)+(1-t)\left(\overline{k}_{w^*}(\gamma_{2})-\varepsilon\right)\\
& =t\overline{k}_{w^{*}}(\gamma_{1})+(1-t)\overline{k}_{w^{*}}(\gamma_{2})-\varepsilon,     
\end{align*}
where the third relation follows from linearity of $U(\cdotp)$. This means that
\[ \overline{k}_{w^{*}}(\tilde{\gamma}_{t}) > t\overline{k}_{w^{*}}(\gamma_{1})+(1-t)\overline{k}_{w^{*}}(\gamma_{2})-\varepsilon,\]
for any $\varepsilon > 0$. We conclude that 
\[\overline{k}_{w^{*}}(\tilde{\gamma}_{t}) \geq t\overline{k}_{w^{*}}(\gamma_{1})+(1-t)\overline{k}_{w^{*}}(\gamma_{2}).\]
\end{proof}

\begin{lem}\label{lem:ex.1.4}
Consider a treatment choice problem with payoff function \eqref{eq:welfare} and statistical model \eqref{eq:normal_model} that exhibits nontrivial partial identification in the sense of Definition \ref{asm:1}. If Assumption \ref{asm:yata.1} holds, then the superdifferential of $\overline{k}_{w^*}(\gamma)$ at $\gamma=0$ is nonempty, bounded, and closed.
\end{lem}

\begin{proof} To see closure, let $s_n \to s^*$ be a converging sequence of elements in the superdifferential. By definition, for every $\gamma \in \Gamma_{w^*}$ we have
\[ \overline{k}_{w^*} (\gamma) \leq \overline{k}_{w^*}(0) + s_n  \gamma. \]
But then, for every $\gamma \in \Gamma_{w^*}$
\[ \overline{k}_{w^*} (\gamma) \leq \overline{k}_{w^*}(0) + s^*  \gamma. \]
Thus, $s^* \in \partial \overline{k}_{w^*}(0)$. 

For nonemptiness and boundedness, by \citet[][Theorem 23.4]{rockafellar1997convex}, it suffices to show that $0\in \text{int}(\Gamma_{w^*})$.

{\scshape Step 1:} We first show that $0\in\Gamma_{w^*}$. As $\Theta$ is centrosymmetric, convex and nonempty, $\overline{\mathbf{0}}\in \Theta$, where $\overline{\mathbf{0}}$ is the zero vector in $\Theta$. As $m(\cdot)$ is linear, it follows $(w^*)^\top m(\overline{\mathbf{0}}) =(w^*)^\top\mathbf{0}=0$, where $\mathbf{0}:=0_{n\times 1}$. That is, $0\in\Gamma_{w^*}$. 

{\scshape Step 2:} We show that there exists some $\gamma\neq0$ such that $\gamma\in \Gamma_{w^*}$. Let $\mathcal{S}$ be an open set in $\mathbb{R}^{n}$ for which
\[ \underline{I}(\mu) < 0 < \overline{I}(\mu), \quad \forall \: \mu \in \mathcal{S}. \]
Such a set exists because the problem exhibits nontrivial partial identification. Pick any $\theta\in\Theta$ such that $m(\theta)\in\mathcal{S}$. If $(w^*)^\top m(\theta)\neq0$, then the initial claim of step 2 follows.  If $(w^*)^\top m(\theta)=0$, then note as $\mathcal{S}$ is open, we can pick some $\epsilon>0$ small enough such that $m(\theta)+\epsilon w^*\in\mathcal{S}$. It follows then $(w^*)^\top (m(\theta)+\epsilon w^*)=\epsilon>0$. Therefore, the claim of step 2 is verified.

{\scshape Step 3:} We show that $\Gamma_{w^*}$ is symmetric around zero. Using the conclusion from Step 2, pick any $\gamma\in \Gamma_{w^*}$ and $\gamma \neq0$. Then there exists $\theta \in \Theta$ such that $\gamma =(w^*)^\top m(\theta)$. Since $\Theta$ is centrosymmetric, $-\theta\in\Theta$ as well. As $(w^*)^\top m(-\theta)=-(w^*)^\top m(\theta)=-\gamma$ by linearity of $m(\cdot)$, one has $-\gamma\in \Gamma_{w^*}$. That is, $M$ is symmetric around zero. 

Steps 1-3 then imply that  $\Gamma_{w^*}$ is a symmetric interval around zero and that $0\in \text{int}(\Gamma_{w^*})$.
\end{proof}

\begin{lem}\label{lem:linear.embedding}
If $s_{w^*}(0)>0$ and $\overline{k}_{w^*}(0) > \sqrt{\frac{\pi}{2}} \sqrt{(w^*)^\top\Sigma w^*}\cdot s_{w^*}(0)$, then the right-hand side of (\ref{eq:linear.embedding.1}) equals $\overline{k}_{w^*}(0)/2$ and infinitely many rules attain this value. In particular, any convex combination of the rules in Equations \eqref{eq:yata.stoye.appendix}-\eqref{eq:linear.appendix} solves the problem in Expression (\ref{eq:linear.embedding.1}).
\end{lem}

\begin{proof}
As $s_{w^*}(0)>0$, the value of the linear embedding problem in (\ref{eq:linear.embedding.1}) equals $s_{w^*}(0)$ times
\begin{equation}\label{eq:linear.embedding.2}
\inf_{d \in \mathcal{D}}\sup_{\gamma \in \mathbb{R}} \left( \sup_{\tilde{U}^*  \in \left [  -k + \gamma, \:   k +  \gamma  \right]  } \tilde{U}^*  \left(  \mathbf{1}\{\tilde{U}^* \geq 0\} - \mathbb{E}_{\gamma}[d(\widehat{\gamma})] \right)  \right),
\end{equation}
where 
\[ \widehat{\gamma} \sim N(\gamma, \sigma^2), \quad \sigma^2 \:= (w^*)^{\top}\Sigma w^*, \quad k := \frac{\overline{k}_{w^*}(0)}{s_{w^*}(0)}.  \]
\cite{stoye2012minimax} shows that, if $k> \sqrt{\pi/2}\: \sigma$ (which holds if and only if $\overline{k}_{w^*}(0) > \sqrt{\frac{\pi}{2}} \sqrt{(w^*)^\top\Sigma w^*}\cdot s_{w^*}(0)$), then 
\[
d_{\text{Gaussian}}^{*}:=\Phi\bigl(\hat{\gamma}/\sqrt{2k^{2}/\pi-\sigma^{2}}\bigr)
\]
solves \eqref{eq:linear.embedding.2} and its worst-case regret is attained at $\gamma=0$; that is, \eqref{eq:linear.embedding.2} equals
\[ \sup_{\gamma \in \mathbb{R}} \left( \sup_{\tilde{U}^*  \in \left [  -k + \gamma, \:   k +  \gamma  \right]  } \tilde{U}^*  \left(  \mathbf{1}\{\tilde{U}^* \geq 0\} - \mathbb{E}_{\gamma}[d^*_{\textrm{Gaussian}}(\widehat{\gamma})] \right)  \right), \]
and this expression equals
\[\sup_{\tilde{U}^*  \in \left [  -k , \:   k  \right]  } \tilde{U}^*  \left(  \mathbf{1}\{\tilde{U}^* \geq 0\} - \mathbb{E}_{\gamma}[d^*_{\textrm{Gaussian}}(\widehat{\gamma})] \right).\]
Moreover, the MMR value \eqref{eq:linear.embedding.2} equals $k/2$. Since $d^*_{\text{RT}} = d_{\text{Gaussian}}^{*}$, this implies that $d^*_{\text{RT}}$ solves the problem in Equation (\ref{eq:linear.embedding.1}) and that this problem has value $\overline{k}_{w^*}(0)/2$.

Lemma \ref{lem:new.minimax} establishes that $d^*_{\text{linear}}$ and $d^*_{\text{mixture}}$ equally attain MMR. Since the set of MMR optimal rules is closed under convex combination, this establishes the claim.
\end{proof}

\begin{lem}\label{lem:new.minimax} If $k > \sqrt{\pi/2} \sigma$, then following rule solves the linear embedding minimax problem defined in \eqref{eq:linear.embedding.2}:
\begin{equation*}\label{eq:d.linear}
d^{*}_{\text{linear}}:=\begin{cases}
0, & \hat{\gamma}<-\rho^{*},\\
\frac{\hat{\gamma}+\rho^{*}}{2\rho^{*}}, & -\rho^{*}\leq\hat{\gamma}\leq\rho^{*},\\
1, & \hat{\gamma}>\rho^{*},
\end{cases}
\end{equation*}
where $\rho^* \in (0,k)$ is the unique solution of
\begin{equation*}
\left(\frac{\rho^{*}}{2\cdot k}\right) -\frac{1}{2} + \Phi\left(-\frac{\rho^{*}}{\sigma}\right) =0. 
\end{equation*} 
\end{lem}

\begin{proof}
Lemma \ref{lem:large.3}(i) shows that $\rho^*\in(0,k)$ exists and is unique. Recall again from \cite{stoye2012minimax} that, if $k> \sqrt{\pi/2}\: \sigma$, the MMR value of the problem is $k/2$, where 
\begin{equation}\label{eq:one.dimensional.id.set.and.risk}
 I(\gamma) := \left [  -k + \gamma, \:   k +  \gamma  \right], \quad  R(d,\gamma,\gamma^{*}) :=  \gamma^*  \left(  \mathbf{1}\{\gamma^* \geq 0\} - \mathbb{E}_{\gamma}[d] \right).    
\end{equation}

Furthermore, by definition of the minimax problem, 
\[
\inf_{d\in\mathcal{D}}\underset{\gamma^{*}\in I(\gamma),\gamma\in\mathbb{R}}{\sup}R(d,\gamma,\gamma^{*})\leq \underset{\gamma^{*}\in I(\gamma),\gamma\in\mathbb{R}}{\sup}R(d^*_{\text{linear}},\gamma,\gamma^{*}).
\]
And Lemma \ref{lem:linear.embedding.upper.bound} shows that, when $k>\sqrt{\pi/2}\sigma$, 
\[ \underset{\gamma^{*}\in I(\gamma),\gamma\in\mathbb{R}}{\sup}R(d^*_{\text{linear}},\gamma,\gamma^{*}) = k/2\] 
and
\[ \sup_{\gamma^* \in I(0)} R(d^*_{\text{linear}},0,\gamma^{*}) = k/2\] 
Hence, we conclude that 
\[
\inf_{d\in\mathcal{D}}\underset{\gamma^{*}\in I(\gamma),\gamma\in\mathbb{R}}{\sup}R(d,\gamma,\gamma^{*})= \underset{\gamma^{*}\in I(\gamma),\gamma\in\mathbb{R}}{\sup}R(d^*_{\text{linear}},\gamma,\gamma^{*})=\frac{k}{2},
\]
and thus $d^*_{\text{linear}}$ is MMR optimal, and its worst-case regret is achieved at $\gamma=0$.
\end{proof}

\begin{lem}\label{lem:large.3}
Consider $\rho^*$ defined in Lemma \ref{lem:new.minimax}.
\begin{itemize}
    \item[(i)] $\rho^{*}\in(0,k)$ exists and is uniquely defined when $\frac{k}{\sigma}>\sqrt{\frac{\pi}{2}}$.
    \item[(ii)] The value of $\rho^{*}$ is strictly decreasing in $\sigma$. Moreover, $\rho^{*}\rightarrow k$ when $\sigma\rightarrow0$.
    \item[(iii)] The value of $\rho^{*}$ is strictly increasing in $k$.  
\end{itemize}
\end{lem}

\begin{proof}
Note 
\begin{equation*}
\left(\frac{\rho^{*}}{2\cdot k}\right) -\frac{1}{2} + \Phi\left(-\frac{\rho^{*}}{\sigma}\right) =0. 
\end{equation*} is equivalent to
\begin{equation*}
1-\frac{k}{\rho^*}\left(1-2\Phi\left(-\frac{\rho^*}{\sigma}\right)\right)=0. 
\end{equation*} 
Write $\mathbf{g}(\rho;k,\sigma)=1-\frac{k}{\rho}\left(1-2\Phi\left(-\frac{\rho}{\sigma}\right)\right)$.

To see (i), further write
\begin{align*}
\frac{\partial}{\partial\rho}\mathbf{g} & =\frac{k}{\rho^{2}}\left(1-2\Phi\left(-\frac{\rho}{\sigma}\right)\right)-\frac{2k}{\rho\sigma}\phi\left(-\frac{\rho}{\sigma}\right)\\
 & =\frac{k}{\rho^{2}}\left(1-2\left(\Phi\left(-\frac{\rho}{\sigma}\right)+\frac{\rho}{\sigma}\phi\left(-\frac{\rho}{\sigma}\right)\right)\right)\\
 & =\frac{k}{\rho^{2}}\left(1-2\left(\Phi\left(-\frac{\rho}{\sigma}\right)-\phi^{\prime}\left(-\frac{\rho}{\sigma}\right)\right)\right).
\end{align*}

Note $\Phi\left(x\right)<\frac{1}{2}$ and $\phi^{\prime}(x)>0$ for
all $x<0$. Thus, $\Phi\left(-\frac{\rho}{\sigma}\right)-\phi^{\prime}\left(-\frac{\rho}{\sigma}\right)<\frac{1}{2}$
for all $\rho>0$. It follows that $\frac{\partial}{\partial\rho}\mathbf{g}>0$, i.e. $\mathbf{g}$ is strictly increasing in $\rho$ for all $\rho>0$. Furthermore, note
\begin{align*}
\lim_{\rho\rightarrow0}\mathbf{g}(\rho,k,\sigma) & =1-\frac{k}{\sigma}\sqrt{\frac{2}{\pi}}=1-\frac{k}{\sigma}\frac{1}{C}\\
\mathbf{g}(k,k,\sigma) & =2\Phi\left(-\frac{k}{\sigma}\right).
\end{align*}
If $\frac{k}{\sigma}>C$, note $\lim\mathbf{g}_{\rho\rightarrow0}(\rho,k,\sigma)<0$, $\mathbf{g}(k,k,\sigma)>0$ and $\mathbf{g}(\cdot;k,\sigma)$ is continuous and strictly increasing. Thus, there exists a unique $\rho^*$ such that $\mathbf{g}(\rho^{*};k,\sigma)=0$. 

To see (ii), note first that $\frac{\partial}{\partial\sigma}\mathbf{g}=\frac{2k}{\sigma^{2}}\phi\left(-\frac{\rho}{\sigma}\right)>0$.
Thus, viewing $\rho^{*}$ as a function of $k$ and $\sigma$, we
can see 
\[
\frac{\partial\rho^{*}}{\partial\sigma}=-\frac{\frac{\partial}{\partial\sigma}\mathbf{g}(\rho^{*};k,\sigma)}{\frac{\partial}{\partial\rho^{*}}\mathbf{g}(\rho^{*};k,\sigma)}<0.
\]
Therefore, $\rho^{*}$ is strictly decreasing in $\sigma$. When $\sigma\rightarrow0$,
$\Phi\left(-\frac{\rho}{\sigma}\right)\rightarrow0$ for each fixed
$\rho>0$. Then, in the limit when $\sigma=0$, $\mathbf{g}(\rho^{*};k,0)=0$
is solved by setting $\rho^{*}=k$. 

To see (iii), note first that $\frac{\partial}{\partial k}\mathbf{g}(\rho^{*};k,\sigma)=-\left(1-2\Phi\left(-\frac{\rho^{*}}{\sigma}\right)\right)/\rho^{*}<0$. The remaining proof mimics that of statement (ii). 
\end{proof}

\begin{lem}\label{lem:linear.embedding.upper.bound}
Suppose $k>\sqrt{\frac{\pi}{2}}\sigma$. Then
\begin{equation}\label{eq:linear.embedding.upper.bound.1}
 \underset{\gamma^{*}\in I(\gamma),\gamma\in\mathbb{R}}{\sup}R(d^*_{\text{linear}},\gamma,\gamma^{*}) = k/2   
\end{equation} 
and
\begin{equation*}\label{eq:linear.embedding.upper.bound.2}
 \sup_{\gamma^* \in I(0)} R(d^*_{\text{linear}},0,\gamma^{*}) = k/2,  
\end{equation*}
where $I(\gamma)$ and $R(d,\gamma,\gamma^*)$ are defined in \eqref{eq:one.dimensional.id.set.and.risk}.
\end{lem}
\begin{proof}
We may write the left-hand side of \eqref{eq:linear.embedding.upper.bound.1} as
\begin{eqnarray*}\label{eq:pf.1}
  \underset{\gamma^{*}\in I(\gamma),\gamma\in\mathbb{R}}{\sup}R(d^*_{\text{linear}},\gamma,\gamma^{*}) =\underset{\gamma+k \geq0}{\sup}(\gamma+k)\left(1-\mathbb{E}_{\gamma}[d_{\text{linear}}^{*}]\right)= \underset{\gamma+k \geq0}{\sup}g_{\text{linear}}(\mu) \\
  g_{\text{linear}}(\gamma):=(\gamma+k)\int_{0}^{1}\Phi\left(\frac{2\rho^{*}x-\rho^{*}-\gamma}{\sigma}\right)dx,
\end{eqnarray*}
where the first equality follows from symmetry of the parameter space and the fact that $d^*_{\text{linear}}(-x)=1-d^*_{\text{linear}}(x)$ for all $x\in\mathbb{R}$, and the second equality follows by applying change-of-variable
twice:
\begin{align*}
&  ~\mathbb{E}_{\gamma}[d_{\text{linear}}^{*}] \\
 =&\int (d_{\text{linear}}^{*}(x))d\Phi\left(\frac{x-\gamma}{\sigma}\right)\\
  =&\int_{-\rho^{*}}^{\rho^{*}}\frac{x+\rho^{*}}{2\rho^{*}}d\Phi\left(\frac{x-\gamma}{\sigma}\right)+\int_{\rho^{*}}^{\infty}d\Phi\left(\frac{x-\gamma}{\sigma}\right)\\
  =&\left[\frac{x+\rho^{*}}{2\rho^{*}}\Phi\left(\frac{x-\gamma}{\sigma}\right)\right]_{-\rho^{*}}^{\rho^{*}}-\int_{-\rho^{*}}^{\rho^{*}}\Phi\left(\frac{2\rho^{*}\left(\frac{x+\rho^{*}}{2\rho^{*}}-\frac{1}{2}\right)-\gamma}{\sigma}\right)d\left(\frac{x+\rho^{*}}{2\rho^{*}}\right)
 +1-\Phi\left(\frac{\rho^{*}-\gamma}{\sigma}\right)\\
  =& ~1-\int_{0}^{1}\Phi\left(\frac{2\rho^{*}x-\rho^{*}-\gamma}{\sigma}\right)dx.
\end{align*}

Now, $\mathbb{E}_{0}[d^*_{\text{linear}}(\hat{\gamma})]=\frac{1}{2}$ and $g_{\text{linear}}(0)=\frac{k}{2}$ by construction. Below, we show that $g_{\text{linear}}(\gamma)$ is first increasing and then decreasing on $[-k,\infty)$ with unique maximum at $\gamma=0$, establishing the claim. To see this, take
first and second derivatives of $g_{\text{linear}}(\gamma)$:
\begin{eqnarray*}
g_{\text{linear}}^{(1)}(\gamma) & = & -\frac{\gamma+k}{\sigma}\int_{0}^{1}\phi\left(\frac{2\rho^{*}x-\rho^{*}-\gamma}{\sigma}\right)dx+\int_{0}^{1}\Phi\left(\frac{2\rho^{*}x-\rho^{*}-\gamma}{\sigma}\right)dx,\\
g_{\text{linear}}^{(2)}(\gamma) & = & \frac{\gamma+k}{\sigma^{2}}\int_{0}^{1}\phi^{(1)}\left(\frac{2\rho^{*}x-\rho^{*}-\gamma}{\sigma}\right)dx-\frac{2}{\sigma}\int_{0}^{1}\phi\left(\frac{2\rho^{*}x-\rho^{*}-\gamma}{\sigma}\right)dx\\
 & = & \frac{\gamma+k}{\sigma^{2}}\int_{0}^{1}\frac{\gamma+\rho^{*}-2\rho x}{\sigma}\phi\left(\frac{2\rho^{*}x-\rho^{*}-\gamma}{\sigma}\right)dx-\frac{2}{\sigma}\int_{0}^{1}\phi\left(\frac{2\rho^{*}x-\rho^{*}-\gamma}{\sigma}\right)dx\\
 & = & \frac{\gamma+k}{2\rho^{*}\sigma}\int_{\frac{-\gamma-\rho^{*}}{\sigma}}^{\frac{-\gamma+\rho^{*}}{\sigma}}-t\phi(t)dt-\frac{1}{\rho^{*}}\int_{\frac{-\gamma-\rho^{*}}{\sigma}}^{\frac{-\gamma+\rho^{*}}{\sigma}}\phi(t)dt\\
 & = & \frac{\gamma+k}{2\rho^{*}\sigma}\int_{\frac{\gamma-\rho^{*}}{\sigma}}^{\frac{\gamma+\rho^{*}}{\sigma}}t\phi(t)dt-\frac{1}{\rho^{*}}\int_{\frac{\gamma-\rho^{*}}{\sigma}}^{\frac{\gamma+\rho^{*}}{\sigma}}\phi(t)dt\\
 & = & \underset{:=A}{\underbrace{\frac{1}{2\rho^{*}}\int_{\frac{\gamma-\rho^{*}}{\sigma}}^{\frac{\gamma+\rho^{*}}{\sigma}}\phi(t)dt}}\left(\frac{\gamma+k}{\sigma}\frac{\int_{\frac{\gamma-\rho^{*}}{\sigma}}^{\frac{\gamma+\rho^{*}}{\sigma}}t\phi(t)dt}{\int_{\frac{\gamma-\rho^{*}}{\sigma}}^{\frac{\gamma+\rho^{*}}{\sigma}}\phi(t)dt}-2\right),
\end{eqnarray*}
where the second equality for $g_{\text{linear}}^{(2)}(\gamma)$ uses
that $\phi'(x)=-x\phi(x)$ for all $x\in\mathbb{R}$, the third one
follows from integration by change-of-variable, and the fourth equality
follows from change-of-variable again and $\phi(x)=\phi(-x)$ for
all $x\in\mathbb{R}$. As $A>0$, the sign of $g_{\text{linear}}^{(2)}(\gamma)$ is determined by 
\[
g_{\text{linear}}^{*}(\gamma):=\frac{\gamma+k}{\sigma}\frac{\int_{\frac{\gamma-\rho^{*}}{\sigma}}^{\frac{\gamma+\rho^{*}}{\sigma}}t\phi(t)dt}{\int_{\frac{\gamma-\rho^{*}}{\sigma}}^{\frac{\gamma+\rho^{*}}{\sigma}}\phi(t)dt}-2.
\]
Furthermore, we can write 
\[
g_{\text{linear}}^{*}(\gamma)=\frac{\gamma+k}{\sigma}\mathbb{E}\left[Z\mid\frac{\gamma-\rho^{*}}{\sigma}\leq Z\leq\frac{\gamma+\rho^{*}}{\sigma}\right]-2,
\]
where $\mathbb{E}\left(Z \mid a\leq Z\leq b\right)$ denotes the conditional expectation
of a standard normal random variable $Z$ conditional on $a\leq Z\leq b$.
We are only interested in $\gamma+k\geq0$. Also, note $\mathbb{E}\left[Z\mid\frac{\gamma-\rho^{*}}{\sigma}\leq Z\leq\frac{\gamma+\rho^{*}}{\sigma}\right]$
strictly increases in $\gamma$ and has the same sign as $\gamma$. Moreover,
note $\mathbb{E}\left[Z\mid-\frac{\rho^{*}}{\sigma}\leq Z\leq\frac{\rho^{*}}{\sigma}\right]=0$,
implying $g_{\text{linear}}^{*}(0)=-2$. Thus, we conclude $g_{\text{linear}}^{(2)}(\gamma)<0$
for all $\gamma$ below some strictly positive threshold and $g_{\text{linear}}^{(2)}(\gamma)>0$
for all larger $\gamma$. That is, $g_{\text{linear}}(\gamma)$ is first
concave and then convex, with the inflexion occurring at a strictly
positive point. 

Since $g_\text{linear}(\gamma)\geq0$ when $\gamma\geq-k$ and it can also be verified that $g_{\text{linear}}(-k)=0$ and
$\lim_{\gamma\to\infty}g_{\text{linear}}(\gamma)=0$, we conclude that $g_{\text{linear}}(\cdot)$
is first strictly increasing and then strictly decreasing, with a unique
maximum. Furthermore, note 
\begin{align*}
g_{\text{linear}}^{(1)}(0) & =-\frac{k}{\sigma}\int_{0}^{1}\phi\left(\frac{2\rho^{*}x-\rho^{*}}{\sigma}\right)dx+\int_{0}^{1}\Phi\left(\frac{2\rho^{*}x-\rho^{*}}{\sigma}\right)dx\\
 & =-\frac{k}{2\rho^{*}}\int_{-\frac{\rho^{*}}{\sigma}}^{\frac{\rho^{*}}{\sigma}}\phi\left(t\right)dt+\frac{\sigma}{2\rho^{*}}\int_{-\frac{\rho^{*}}{\sigma}}^{\frac{\rho^{*}}{\sigma}}\Phi\left(t\right)dt\\
 & =-\frac{k}{2\rho^{*}}\left(\Phi\left(\frac{\rho^{*}}{\sigma}\right)-\Phi\left(-\frac{\rho^{*}}{\sigma}\right)\right)+\frac{1}{2}\\
 & =\frac{1}{2}-\frac{k}{2\rho^{*}}\left(1-2\Phi\left(-\frac{\rho^{*}}{\sigma}\right)\right)\\
 & =0,
\end{align*}
where the second equality applies change-of-variable and the last
equality follows from the definition of $\rho^{*}$, which exists and is unique when $k>C\sigma$ by Lemma \ref{lem:large.3}. Thus, $g_{\text{linear}}(\cdotp)$
has a unique maximization point in $[-k,\infty)$ at $\gamma=0$. 
\end{proof}

\begin{lem}\label{lem:erm}
Under assumptions made in Theorem \ref{thm:main.regret.1},    \[d_{0}:=d_{0}((w^*)^\top Y ):=\mathbf{1}\{(w^*)^\top Y \geq 0\}\] is not MMR optimal.
\end{lem}

\begin{proof}
By Step 1 in the proof of part (i) of Theorem \ref{thm:main.regret.1}, we know the MMR value of problem \eqref{eq:MMR.value} equals $\overline{k}_{w^{*}}(0)/2$. In contrast, we will show that
\[
R_{w^*,0}^{*}:=\sup_{\theta\in\Theta}U(\theta)\left(\mathbf{1}\{U(\theta)\geq0\}-\mathbb{E}_{m(\theta)}[d_{0}((w^{*})^{\top}Y)]\right)> \overline{k}_{w^{*}}(0)/2.
\]
Write
\begin{align*}
R_{w^*,0}^{*}&=\sup_{\gamma\in\Gamma_{w^{*}}}\left(\sup_{U^{*}\in(-\overline{k}_{w^{*}}(-\gamma),\:\overline{k}_{w^{*}}(\gamma))}U^{*}\left(\mathbf{1}\{U^{*}\geq0\}-\mathbb{E}_{\gamma}[d_{0}((w^{*})^{\top}Y)]\right)\right)\\
&=  \sup_{\gamma\in\Gamma_{w^{*}},\overline{k}_{w^{*}}(\gamma)>0}\overline{k}_{w^{*}}(\gamma)\left(1-\mathbb{E}_{\gamma}[d_{0}((w^{*})^{\top}Y)]\right)\\
&= \sup_{\gamma\in\Gamma_{w^{*}},\overline{k}_{w^{*}}(\gamma)>0}\overline{k}_{w^{*}}(\gamma)\Phi \left(-\frac{\gamma}{\sigma} \right),
\end{align*}
using first centrosymmetry of $\Theta$ and then $(w^{*})^{\top}Y\sim N(\gamma,\sigma^2)$, where $\sigma^2 :=(w^*)^\top\Sigma w^*$. Write $g(\gamma) : =\overline{k}_{w^{*}}(\gamma)\Phi(-\frac{\gamma}{\sigma})$ for $\gamma \in \Gamma_{w^*}$ such that $\overline{k}_{w^{*}}(\gamma)>0$. By definition
\[ g(0) = \overline{k}_{w^*}(0) / 2.  \]
Let $\partial g(\cdot)$ be the generalized gradient of $g(\cdot)$. In the following, we show that $0\notin\partial g(0)$. By \citet[][Proposition 2.3.2]{clarke1990optimization}, $0$ is then not a local maximum or minimum; hence, $R^*_{w^*,0}>\overline{k}_{w^*}(0) / 2$.

Note $\Phi(\cdot)$ is strictly differentiable and thus Lipschitz near $0$ \cite[][Proposition 2.2.4]{clarke1990optimization}. Also, as $\overline{k}_{w^*}(0)$ is concave and bounded from below near 0, $\overline{k}_{w^*}(\cdot)$ must be  Lipschitz near $0$ \cite[][Proposition 2.2.6]{clarke1990optimization}. Moreover, both $\overline{k}_{w^*}$ and $\Phi$ are regular at $0$ \cite[][Proposition 2.3.6]{clarke1990optimization} as well as positive. By \citet[][Proposition 2.3.13]{clarke1990optimization}, the (appropriately generalized) chain rule can be applied to $g$ to characterize $\partial g(0)$.

As $\Phi$ is strictly differentiable, its generalized gradient coincides with the unique derivative \cite[][Proposition 2.2.4]{clarke1990optimization}). As $\overline{k}_{w^*}$ is concave and Lipschitz near 0, its generalized gradient coincides with its superdifferential \cite[][Proposition 2.2.7]{clarke1990optimization}). Hence, let $\tilde{s}_{w^*}(0)$ be a supergradient of $\overline{k}_{w^*}(0)$. We may calculate the generalized gradient of $g(\gamma)$ at $\gamma=0$ as
\begin{align*}
&\frac{\tilde{s}_{w^*}(0)}{2}-\frac{{}\overline{k}_{w^{*}}(0)\phi(0)}{\sigma}\\
=&\frac{\phi(0)}{\sigma}\left(\sqrt{\frac{\pi}{2}} \sigma\cdot \tilde{s}_{w^*}(0)-\overline{k}_{w^*}(0)\right)\\
\leq& \frac{\phi(0)}{\sigma}\left(\sqrt{\frac{\pi}{2}} \sigma\cdot s_{w^*}(0)-\overline{k}_{w^*}(0)\right)<0,
\end{align*}

where the first inequality follows as $s_{w^*}(0)$ is the largest supergradient so  $s_{w^*}(0)\geq \tilde{s}_{w^*}(0)$, and the second inequality follows from noting $\overline{I}(\mathbf{0})=\overline{k}_w^*(0)$ under stated assumptions (by Step 1 in the proof of Theorem \ref{thm:main.regret.1}(i)) and by picking $\overline{I}(\mathbf{0})$ large enough so that $\overline{I}(\mathbf{0}) > \sqrt{\pi/2} \cdot \sigma\cdot s_{w^*}(0)$. Thus, we have shown that $\partial g(0)<0$, and therefore $0$ is not a local maximum or minimum.
\end{proof}

\subsection{Proof of Proposition \ref{thm:running.example}}\label{sec:proof.thm.running}

\textbf{Statements (i)-(ii)}. In the running example, the expected regret of a rule $d(\cdot)$ can be written as 
\[
R(d,\mu,\mu_{0})=\mu_{0}\left(\mathbf{1}\{\mu_{0}\geq0\}-\mathbb{E}_{\mu}[d(Y)]\right),\quad\mu\in M,\mu_{0}\in I(\mu).
\]
where $Y\sim N(\mu,\Sigma)$; for future use, we state its likelihood
\[
f(y\mid\mu)=\frac{1}{\sqrt{(2\pi)^{n}\Pi_{j=1}^{n}\sigma^2_{j}}}\exp\left(-\frac{1}{2}\sum_{i=1}^{n}\frac{\left(y_{i}-\mu_{i}\right)^{2}}{\sigma_{i}^{2}}\right).
\]

Consider the class of decision rules (parameterized by scalar $m_0\geq C\Vert x_{1}-x_{0}\Vert$)
\begin{eqnarray*}
    d_{m_0}& :=& \mathbf{1}\{w_{m_0}^\top Y \geq 0\} \\
    w_{m_0}^\top &:= & \left(1, \frac{\max \{m_0 - C\left\Vert x_2-x_0 \right\Vert,0\}/\sigma_2^2}{(m_0 - C\left\Vert x_1-x_0 \right\Vert)/\sigma_1^2},\ldots,\frac{\max \{m_0 - C\left\Vert x_n-x_0 \right\Vert,0\}/\sigma_n^2}{(m_0 - C\left\Vert x_1-x_0 \right\Vert)/\sigma_1^2}\right),
\end{eqnarray*}
with the understanding that, for $m_0= C\Vert x_{1}-x_{0}\Vert$, we have $w_{m_0}^\top=(1,0,\ldots,0)$. Consider also the class of priors $\pi_{m_0}$ that randomize evenly over $\left\{\left(\mu_{0},\mu^{\top}\right)^{\top},\left(-\mu_0,-\mu^{\top}\right)^{\top}\right\}$, where 
\begin{eqnarray}
\mu_{0} & = & m_0\nonumber \\
\mu_{j} & = & \max\left\{ m_0-C\left\Vert x_{j}-x_{0}\right\Vert ,0\right\} ,~~~j=1,\ldots,n\label{eq:define mu}
\end{eqnarray}
for some $m_0\geq  C\Vert x_{1}-x_{0}\Vert$. We will show that i) for any prior $\pi_{m_0}$, rule $d_{m_0}$ is a corresponding Bayes rule, uniquely so if $m_0>C\Vert x_1-x_0 \Vert$, ii) for any decision rule $d_{m_0}$, a prior $\pi_{\tilde{m}_0}$ (note $\tilde{m}_0 \neq m_0$ in general) is least favorable, and finally that iii) the resulting best-response mapping has a fixed point $m_0^*$. This fixed point defines the MMR rule from the proposition, which is furthermore unique whenever it is uniquely Bayes against $\pi_{m_0^*}$.\footnote{Thus, we use the game theoretic characterization of maximin-type decision rules (e.g., \citet[][Section 5]{berger1985statistical}).}

Regarding step i), if $m_0>C\Vert x_1-x_0 \Vert$, the unique Bayes response to $\pi_{m_0}$ equals
\begin{eqnarray*}
 &&   \mathbf{1}\left\{ \mathbb{E}[\mu_{0}\mid Y]\geq0\right\} \\
 &=& \mathbf{1}\left\{ f(Y|\mu)-f(Y|-\mu)\geq0\right\} \\
 &=& d_{m_0}(Y), 
\end{eqnarray*}
where the last step uses familiar normal likelihood algebra.\footnote{This may be easier to see upon multiplying $w_{m_0}$ through by $m_0-C\Vert x_1-x_0 \Vert)/\sigma_1^2$. Our notation is meant to clarify continuity and convergence to $(1,0,\ldots,0)$.} Any decision rule is Bayes against $\pi_{m_0}$ if $m_0=C\Vert x_1-x_0 \Vert$.

Regarding ii), observe that expected regret of $d_{m_0}$ depends on $\theta$ only through $(\mu,\mu_0)$ and that maximizing it amounts to solving
\begin{align*}
&\sup_{\mu\in M,\mu_{0}\in I(\mu)} \mu_{0}\left(\mathbf{1}\{\mu_{0}\geq0\}-\Phi\left(\frac{w_{m_{0}}^{\top}\mu}{\sqrt{w_{m_{0}}^{\top}\Sigma w_{m_{0}}}}\right)\right)\\
&=  \sup_{\mu\in M,\mu_{0}\in I(\mu),\mu_{0}\geq0}\mu_{0}\Phi\left(-\frac{w_{m_{0}}^{\top}\mu}{\sqrt{w_{m_{0}}^{\top}\Sigma w_{m_{0}}}}\right),
\end{align*}
where we used centrosymmetry of $\Theta$ and where 
\begin{align*}
M & =\left\{ \mu\in\mathbb{R}^{n}:\left|\mu_{i}-\mu_{j}\right|\leq C\left\Vert x_{i}-x_{j}\right\Vert ,i,j=1\ldots n,\text{\ensuremath{}}i\neq j\right\} .\\
I(\mu) & =\left\{ u\in\mathbb{R}:\left\vert \mu_{i}-u\right\vert \leq C\left\Vert x_{i}-x_{0}\right\Vert ,~~i=1,\ldots.n\right\} .
\end{align*}
Let $j_{m_{0}}^{*}$ be the highest index $j$ for which $w_{m_{0},j}$
is not $0$. Then, $\mu_{0}\Phi\left(-\frac{w_{m_{0}}^{\top}\mu}{\sqrt{w_{m_{0}}^{\top}\Sigma w_{m_{0}}}}\right)$
does not depend on $(\mu_{j_{m_{0}}^{*}+1},\ldots,\mu_{n})$ and decreases
in $(\mu_{1},\ldots,\mu_{j_{m_{0}}^{*}})$. It follows that some prior $\pi_{\tilde{m}_0}$ is least favorable, where furthermore $\tilde{m}_0$ is the optimal argument $\mu_0$ in
\begin{gather*}
\sup_{\mu\in M,\mu_{0}\in I(\mu),\mu_{0}\geq0}\mu_{0}\Phi\left(-\frac{w_{m_{0}}^{\top}\mu}{\sqrt{w_{m_{0}}^{\top}\Sigma w_{m_{0}}}}\right)  =\sup_{\mu_{0}\geq0}g(\mu_{0},m_{0}), \\
g(\mu_{0},m_{0})=\mu_{0}\Phi\left(-\frac{\sum_{j=1}^{n}\frac{\max\left\{ m_{0}-C\left\Vert x_{j}-x_{0}\right\Vert ,0\right\} \left(\mu_{0}-C\left\Vert x_{j}-x_{0}\right\Vert \right)}{\sigma_{j}^{2}}}{\sqrt{\sum_{j=1}^{n}\frac{\max^{2}\left\{ m_{0}-C\left\Vert x_{j}-x_{0}\right\Vert ,0\right\} }{\sigma_{j}^{2}}}}\right).
\end{gather*}
Regarding iii), the best-response mapping $\psi: \left[C\left\Vert x_{1}-x_{0}\right\Vert,\infty\right)\rightrightarrows\left[0,\infty\right)$ defined by\footnote{The mapping is in fact a function, but establishing that would be unnecessary work.}
\begin{equation}
\psi(m_0):=\arg\sup_{\mu_0\geq0}g(\mu_{0},m_{0})   \label{eq:define_best_response}
\end{equation}
has a fixed point $m_0^* \in \psi(m_0^*)$. To see this, first compactify the domain of $\mu_0$ in the above definition by noting that $\tilde{m}_0$ can be universally bounded from above. This is because, for any $m_0$ under consideration, one has $g(C\Vert x_1-x_0\Vert,m_0)\geq C\Vert x_1-x_0\Vert/2$ but also
\begin{eqnarray*}
    g(\mu_0,m_0) =  \mu_0 (1-\Pr(w_{m_0}^\top Y \geq 0)) \leq \mu_0 (1-\Pr(Y \geq \mathbf{0})) = \mu_0 \left(1-\prod_{j=1}^n \Phi(\mu_j/\sigma_j)\right).
\end{eqnarray*}
Using \eqref{eq:define mu}, this upper bound is seen to vanish as $m_0 \to \infty$. Hence, it is w.l.o.g. to change the constraint set in \eqref{eq:define_best_response} to $0\leq \mu_0 \leq \overline{\mu}_0$ for $\overline{\mu}_0$ large enough (but independent of $m_0$). Given this compactification, continuity of $g(\cdot)$ implies nonemptiness and upper hemicontinuity of $\psi(\cdot)$. Next, by algebra resembling Proposition 7 in \citet[][see also Lemma \ref{lem:linear.embedding.upper.bound} above]{stoye2012minimax}, for any fixed $m_0$ the function $g(\mu_{0},m_{0})$ is first concave then convex in $\mu_{0}$ and converges to $0$ [$-\infty$] as $\mu_{0}\to\infty$ [$\mu_{0}\to-\infty$]. Hence, $\psi(\cdot)$ is interval-valued. These observations jointly imply that the graph of $\psi(\cdot)$ is path-connected. We show below that (with slight abuse of notation for set-valued mappings) $\psi(C\left\Vert x_{1}-x_{0}\right\Vert) \geq C\left\Vert x_{1}-x_{0}\right\Vert$, and we already know that $\psi(m_0)<m_0$ for $m_0>\overline{\mu}_0$. This establishes existence of $m_0^*$.

We next show that $\psi(C\Vert x_{1}-x_{0}\Vert) \geq C\Vert x_{1}-x_{0}\Vert$, strictly so if $C \Vert x_1 - x_0 \Vert < \sqrt{\pi / 2} \cdot \sigma_1$. This is because for $m_0=C\Vert x_{1}-x_{0}\Vert$, we have
\begin{eqnarray*}
g(\mu_0,m_0) &= & \mu_0\Phi\left(\frac{C\left\Vert x_{1}-x_{0}\right\Vert-\mu_0}{\sigma_1}\right) \\
\implies \left. \frac{\partial g(\mu_{0},m_{0})}{\partial\mu_{0}} \right\Vert_{\mu_0=C\left\Vert x_{1}-x_{0}\right\Vert} &=& - \frac{C\left\Vert x_{1}-x_{0}\right\Vert \phi(0)}{\sigma_1} + \frac{1}{2}.
\end{eqnarray*}
After substituting in for $\phi(0)$ and simplifying, the above partial derivative is seen to have the same sign as $\sqrt{\pi/2}\cdot \sigma_1-C\Vert x_{1}-x_{0}\Vert$. This establishes the claim and also proves statement (ii) because the fixed point $m_0^*= C\Vert x_{1}-x_{0}\Vert$ has been discovered for that statement's case.

This concludes the proof. For ease of computation, we note that, after tedious algebra, $m_0^*$ can be uniquely characterized by
\begin{equation}\label{eq:small.discrepancy.m.star}
\frac{\Phi\left(-\sqrt{\sum_{j=1}^{n}\frac{\max^{2}\left\{ m_{0}^*-C\left\Vert x_{j}-x_{0}\right\Vert ,0\right\} }{\sigma_{j}^{2}}}\right)}{\phi\left(-\sqrt{\sum_{j=1}^{n}\frac{\max^{2}\left\{ m_{0}^*-C\left\Vert x_{j}-x_{0}\right\Vert ,0\right\} }{\sigma_{j}^{2}}}\right)}=m_0^*\frac{\sum_{j=1}^{n}\frac{\max\left\{ m_{0}^*-C\left\Vert x_{j}-x_{0}\right\Vert ,0\right\} }{\sigma_{j}^{2}}}{\sqrt{\sum_{j=1}^{n}\frac{\max^{2}\left\{ m_{0}^*-C\left\Vert x_{j}-x_{0}\right\Vert ,0\right\} }{\sigma_{j}^{2}}}}.  
\end{equation}

\textbf{Statement (iii)}. 
Below, we show $d^{*}_{\text{RT}}((w^*)^{\top}Y)=\Phi\bigl((w^*)^{\top}Y/\tilde{\sigma})$, where $\tilde{\sigma} = \sqrt{ 2C^2 \|x_1-x_0\|^2/\pi - \sigma_1^2 }$ and $w^*=(1,0,\ldots,0)^{\top}$ is MMR optimal and satisfies \eqref{eq:yata.property.1} and \eqref{eq:yata.property.2}. Therefore, results in Theorem \ref{thm:main.regret.1} apply. In particular, Lemma \ref{lem:linear.embedding} implies that $d^{*}_\text{linear}$ defined in \eqref{eq:linear.rule} with $\rho$ such that  $\frac{C\left\Vert x_1-x_0 \right\Vert-\rho^*}{2 C\left\Vert x_1-x_0 \right\Vert} = \Phi\left(-\rho^*/\sigma_1\right)$, as well as any convex combination of $d^*_{\text{RT}}$ and $d^{*}_{\text{linear}}$, are also MMR optimal. Let
\[\mathbf{R}:=\min_{d\in\mathcal{D}_n}\sup_{\theta\in\Theta} U(\theta) \left(  \mathbf{1}\{U(\theta) \geq 0\} - \mathbb{E}_{m(\theta)}[d(Y)] \right)\]
be the minimax value of the problem in the running example.

{\scshape Step 1:} We show $\mathbf{R}\leq C\|x_1-x_0\|/2$ because this bound is attained by $d_{\text{RT}}^*(\cdot)$. To see this, use Step 2 of the proof of Theorem \ref{thm:main.regret.1} to write
\begin{eqnarray}
\mathbf{R} &\leq &\sup_{\theta\in\Theta} \theta_0 \left(  \mathbf{1}\{\theta_0 \geq 0\} - \mathbb{E}_{m(\theta)}[d^*_{\text{RT}}((w^*)^{\top}Y)] \right) \notag \\
&= &\sup_{\gamma\in \mathbb{R}}\left(\sup_{U^{*}\in \left [ - \overline{k}_{w^*}(-\gamma), \:   \overline{k}_{w^*}(\gamma)  \right]}U^{*}\left(\mathbf{1}\{U^{*}\geq0\}-\mathbb{E}_{\gamma}[d^*_{\text{RT}}((w^{*})^{\top}Y)]\right)\right),\label{eq:running.example.mmr.2}
\end{eqnarray}
where $w^*=(1,0,\ldots,0)^\top$ and $\overline{k}_{w^*}(\gamma)$ is defined in \eqref{eq:isu} and calculated as
\begin{align}
\sup\quad &\theta_0\notag\\
\text{s.t.}\quad&  |\theta_i-\theta_j| \leq C \| x_i - x_j \|,\quad i,j=0,...,n,\label{eq:running.example.linear.program}\\
&\theta_1=\gamma.\notag
\end{align}
As $\|x_1-x_0\| \leq \|x_j-x_0\|$ for all $j=1,...,n$, the linear program \eqref{eq:running.example.linear.program} admits a simple solution $\overline{k}_{w^*}(\gamma)=\gamma+C\|x_{0}-x_{1}\|$. Therefore, \eqref{eq:running.example.mmr.2} can be further written as  \begin{equation}\label{eq:running.example.mmr.3}
\sup_{\gamma\in \mathbb{R}}\left(\sup_{U^{*}\in \left [ \gamma-C \left\Vert x_0-x_1 \right\Vert, \:   \gamma+C\left\Vert x_0-x_1 \right\Vert \right]}U^{*}\left(\mathbf{1}\{U^{*}\geq0\}-\mathbb{E}_{\gamma}[d^*_{\text{RT}}((w^*)^{\top}Y)]\right)\right),
\end{equation}
where $(w^*)^{\top}Y\sim N(\gamma, \sigma_{1}^2)$. Applying \cite{stoye2012minimax} with $k=C\|x_0-x_1\|$ and $\sigma=\sigma_1$, we conclude that when $C \| x_1 - x_0 \| > \sqrt{\pi / 2} \cdot \sigma_1$, $d^*_{\text{RT}}((w^*)^{\top}Y)$ solves the minimax problem \[
\min_{d\in\mathcal{D}}\sup_{\gamma\in \mathbb{R}}\left(\sup_{U^{*}\in \left [ \gamma-C\|x_{0}-x_{1}\|, \:   \gamma+C\|x_{0}-x_{1}\| \right]}U^{*}\left(\mathbf{1}\{U^{*}\geq0\}-\mathbb{E}_{\gamma}[d((w^*)^{\top}Y)]\right)\right)
\]
and the value of \eqref{eq:running.example.mmr.3} equals $C\|x_0-x_1\|/2$.

{\scshape Step 2:} $\mathbf{R}\geq C\|x_1-x_0\|/2$ because this value is attained by setting $(\theta_1,\ldots,\theta_n)=\mathbf{0}$:
\begin{align}
\mathbf{R}&\geq \min_{d\in\mathcal{D}_n}\sup_{\theta\in\Theta:\theta=(\theta_0,\mathbf{0})} \theta_0 \left(  \mathbf{1}\{\theta_0 \geq 0\} - \mathbb{E}_\mathbf{0}[d(Y)] \right)\notag \\
&=\min_{d\in\mathcal{D}_n}\max\left\{\overline{I}(\mathbf{0})   \left(  1- \mathbb{E}_{0}[d(Y)] \right), \overline{I}(\mathbf{0})  \mathbb{E}_\mathbf{0}[d(Y)]\right\}\notag,
\end{align}
where the last line used $\underline{I}(\mathbf{0})=-\overline{I}(\mathbf{0})$. This minimum is attained by any rule with $\mathbb{E}_\mathbf{0}[d(Y)]=1/2$, and its value equals $\overline{I}(\textbf{0})/2=C\Vert x_1-x_0\Vert/2$.

As $\mathbb{E}_\mathbf{0}[d_{\text{RT}}^*((w^*)^\top Y)]=1/2$, the last step above also verifies $\eqref{eq:yata.property.1}$. Optimality of $d^*_{\text{RT}}$ in Step 1 verified \eqref{eq:yata.property.2}.

\textbf{Statement (iv)}.
Using that $x_1$ is a unique nearest neighbor, we have that, for $\mu$ close to $\mathbf{0}$, $\overline{I}(\mu) = \mu_1 + C \left\Vert x_1-x_0 \right\Vert$. This is clearly differentiable at $\mu=\bm{0}$ and Theorem \ref{thm:main.regret.1}(iii) therefore applies.

\subsection{Proof of Proposition \ref{prop:linear.vs.RT}}\label{sec:proof.prop.linear.vs.RT}
Let $k=C\left\Vert x_{1}-x_{0}\right\Vert $, $\sigma=\sigma_{1}$, and recall $w^{*}=(1,0,\ldots0)^{\top}$. Denote by $\pi_{\text{RT}}(\gamma):=\mathbb{E}_{\gamma}\left[\delta_{\text{RT}}^{*}\left(\left(w^{*}\right)^{\top}Y\right)\right]$ and $\pi_{\text{linear}}(\gamma):=\mathbb{E}_{\gamma}\left[\delta_{\text{linear}}^{*}\left(\left(w^{*}\right)^{\top}Y\right)\right]$ the expected adoption probabilities of $\delta_{\text{RT}}^{*}$ and $\delta_{\text{linear}}^{*}$, respectively, where $\left(w^{*}\right)^{\top}Y\sim N(\gamma,\sigma)$ for $\gamma\in\mathbb{R}$. Algebra shows:
\begin{align*}
\pi_{\text{RT}}(\gamma) & =\Phi\left(\frac{\gamma}{\frac{k}{\sqrt{\frac{\pi}{2}}}}\right), & \pi_{\text{linear}}(\gamma)=1-\int_{0}^{1}\Phi\left(\frac{2\rho^{*}x-\rho^{*}-\gamma}{\sigma}\right)dx,\\
\pi^{(1)}_{\text{RT}}(\gamma) & =\frac{1}{\frac{k}{\sqrt{\frac{\pi}{2}}}}\phi\left(\frac{\gamma}{\frac{k}{\sqrt{\frac{\pi}{2}}}}\right), & \pi_{\text{linear}}^{(1)}(\gamma)=\frac{1}{2\rho^{*}}\left[\Phi\left(\frac{\rho^{*}-\gamma}{\sigma}\right)-\Phi\left(\frac{-\rho^{*}-\gamma}{\sigma}\right)\right],\\
\pi^{(2)}_{\text{RT}}(\gamma) & =-\frac{\gamma}{\left(\frac{k}{\sqrt{\frac{\pi}{2}}}\right)^{3}}\phi\left(\frac{\gamma}{\frac{k}{\sqrt{\frac{\pi}{2}}}}\right), & \pi_{\text{linear}}^{(2)}(\gamma)=\frac{1}{2\rho^{*}}\left[\frac{1}{\sigma}\phi\left(\frac{\rho^{*}+\gamma}{\sigma}\right)-\frac{1}{\sigma}\phi\left(\frac{\rho^{*}-\gamma}{\sigma}\right)\right].
\end{align*}
Due to symmetry, we only need to focus on $\gamma\geq0$. 

{\scshape Step 1:} We show that, for $\gamma\geq0$,
\begin{align*}
\overline{R}_{w^{*}}(\delta_{\text{RT}}^{*},\gamma) & =(\gamma+k)\left(1-\pi_{\text{RT}}(\gamma)\right)\\
\overline{R}_{w^{*}}(\delta_{\text{linear}}^{*},\gamma) & =(\gamma+k)\left(1-\pi_{\text{linear}}(\gamma)\right).
\end{align*}
By inspection of
\[
\overline{R}_{w^{*}}(\delta_{\text{RT}}^{*},\gamma)=\max\left\{ (\gamma+k)\left(1-\pi_{\text{RT}}(\gamma)\right),(k-\gamma)\pi_{\text{RT}}(\gamma)\right\},
\]
this is obviously true when $\gamma\geq k$. As $\pi_{\text{RT}}(\gamma)$ is strictly increasing and concave in $\gamma$ with $\pi_{\text{RT}}(k)<1$, $\pi_{\text{RT}}(0)=\frac{1}{2}$ and $\pi^{(1)}(0)=\frac{1}{2k}$, one has $\pi_{\text{RT}}(\gamma)\leq\frac{\gamma+k}{2k}$, implying
\[
\max\left\{ (\gamma+k)\left(1-\pi_{\text{RT}}(\gamma)\right),(k-\gamma)\pi_{\text{RT}}(\gamma)\right\} =(\gamma+k)\left(1-\pi_{\text{RT}}(\gamma)\right)
\]
for $0\leq\gamma<k$ as well. Analogous arguments reveal the same
for $\delta_{\text{linear}}^{*}$. 

{\scshape Step 2:} By step 1, it suffices to compare $\pi_{\text{RT}}(\gamma)$ and $\pi_{\text{linear}}(\gamma)$. Write 
\[
\pi_{\text{linear}}(\gamma)-\pi_{\text{RT}}(\gamma)=\int\left[\delta_{\text{linear}}^{*}(x)-\delta_{\text{RT}}^{*}(x)\right]\phi\left(\frac{x-\gamma}{\sigma}\right)dx.
\]
Consider first sign changes of $\delta_{\text{linear}}^{*}(x)-\delta_{\text{RT}}^{*}(x)$. Recall that (i) $\delta_{\text{linear}}^{*}(0)=\delta_{\text{RT}}^{*}(0)=\frac{1}{2}$, (ii) $\delta_{\text{linear}}^{*}$ is piecewise linear and attains values $\delta_{\text{linear}}^{*}(x)=0$ for $x<-\rho^{*}$ and $\delta_{\text{linear}}^{*}(x)=1$ for $x>\rho^{*}$, (iii) $\delta_{\text{RT}}^{*}(x)$ is strictly convex on $(-\infty,0)$, strictly concave on $(0,\infty)$, and its range is $(0,1)$.  
This implies that $\delta_{\text{linear}}^{*}(x)-\delta_{\text{RT}}^{*}(x)$ has at most three sign changes. By \citet[Theorem 3]{karlin1957polya}, $\pi_{\text{linear}}(\gamma)-\pi_{\text{RT}}(\gamma)$ then has at most three sign changes as well. Moreover, $0$ must be a point of sign change because $\pi_{\text{linear}}(0)=\pi_{\text{RT}}(0)$ and $\pi_{\text{linear}}(-\gamma)-\pi_{\text{RT}}(-\gamma)=-\left(\pi_{\text{linear}}(\gamma)-\pi_{\text{RT}}(\gamma)\right)$.

{\scshape Step 3:} We show that $\pi_{\text{linear}}(\gamma)>\pi_{\text{RT}}(\gamma)$
for all $\gamma>0$ sufficiently large. To see this, note $\pi_{\text{linear}}(\gamma)\geq\mathbb{E}_{\gamma}\left[\mathbf{1}\left(\left(w^{*}\right)^{\top}Y\geq\rho^{*}\right)\right]=\Phi\left(\frac{\gamma-\rho^{*}}{\sigma}\right)$
and $k>\sqrt{\frac{\pi}{2}}\sigma$. It follows that $\Phi\left(\frac{\gamma-\rho^{*}}{\sigma}\right)>\Phi\left(\frac{\gamma}{\frac{k}{\sqrt{\frac{\pi}{2}}}}\right)=\pi_{\text{linear}}(\gamma)$
for all $\gamma>\frac{k\rho^{*}}{k-\sqrt{\frac{\pi}{2}}\sigma}$. 

{\scshape Step 4:} Conclusions of Steps 2 and 3 imply that one of the following
two must be true:
\begin{itemize}
\item[(i)] $\pi_{\text{linear}}(\gamma)-\pi_{\text{RT}}(\gamma)$ has exactly one sign change from ``$-$'' to ``$+$'' at $0$, i.e., $\pi_{\text{linear}}(\gamma)\geq\pi_{\text{RT}}(\gamma)$ for all $\gamma\geq0$ with the inequality strict when $\gamma>0$.
\item[(ii)] $\pi_{\text{linear}}(\gamma)-\pi_{\text{RT}}(\gamma)$ has exactly three sign changes. In this case, the order of the sign changes must be ``$-+-+$'', and $0$ must be a sign change from $+$ to $-$. We then conclude that $\pi_{\text{linear}}(\gamma)>\pi_{\text{RT}}(\gamma)$ for all $\gamma>\underline{\gamma}$, where $\underline{\gamma}$ is the unique number in $(0,\infty)$ such that $\pi_{\text{linear}}(\underline{\gamma})=\pi_{\text{RT}}(\underline{\gamma})$, i.e., $\Phi\left(-\frac{\underline{\gamma}}{\frac{k}{\sqrt{\frac{\pi}{2}}}}\right)=\int_{0}^{1}\Phi\left(\frac{2\rho^{*}x-\rho^{*}-\underline{\gamma}}{\sigma}\right)dx$. 
\end{itemize}
{\scshape Step 5:} To determine which case applies, we study the
sign of $\pi_{\text{linear}}(\varepsilon)-\pi_{\text{RT}}(\varepsilon)$
for $\varepsilon$ sufficiently close to zero. Note
\begin{align*}
\pi_{\text{linear}}(0) & =\pi_{\text{RT}}(0)=\frac{1}{2},\\
\pi_{\text{linear}}^{(1)}(0) & =\pi_{\text{RT}}^{(1)}(0)=\frac{1}{2k},\\
\pi_{\text{linear}}^{(2)}(0) & =\pi_{\text{RT}}^{(2)}(0)=0,\\
\pi_{\text{linear}}^{(3)}(0) & =-\frac{\phi(\frac{\rho^{*}}{\sigma})}{\sigma^{3}},\\
\pi_{\text{RT}}^{(3)}(0) & =-\left(\frac{\sqrt{\frac{\pi}{2}}}{k}\right)^3\phi\left(0\right).
\end{align*}
Thus, if $\phi(\frac{\rho^{*}}{\sigma})\left(\frac{k}{\sqrt{\frac{\pi}{2}}}\right)^{3}\leq\sigma^{3}\phi\left(0\right)$, we have $\pi_{\text{linear}}^{(3)}(0)\geq\pi_{\text{RT}}^{(3)}(0)$, implying $\pi_{\text{linear}}(\varepsilon)\geq\pi_{\text{RT}}(\varepsilon)$ for $\varepsilon>0$ sufficiently small. As a result, case (i) in Step 4 must hold. Otherwise, we have $\pi_{\text{linear}}^{(3)}(0)<\pi_{\text{RT}}^{(3)}(0)$, implying $\pi_{\text{linear}}(\varepsilon)<\pi_{\text{RT}}(\varepsilon)$ for $\varepsilon>0$ sufficiently small, and case (ii) applies. This completes the proof.
 
\subsection{Proof of Proposition \ref{prop:profiled.regret.dominance}}\label{sec:proof.profiled.regret.dominance}

First, we may calculate 
\begin{align}
\overline{R}_{r}(d_{\text{coin-flip}},\mu) & =\sup_{\theta\in\Theta,m(\theta)=\mu}U(\theta)\left[\mathbf{1}\left\{ U(\theta)\geq0\right\} -\frac{1}{2}\right]\nonumber \\
 & =\frac{1}{2}\max\left\{ \sup_{\theta\in\Theta,m(\theta)=\mu,U(\theta)\geq0}U(\theta),\sup_{\theta\in\Theta,m(\theta)=\mu,U(\theta)\leq0}-U(\theta)\right\} \nonumber \\
 & =\frac{1}{2}\max\left\{ \max\left\{ 0,\overline{I}(\mu)\right\} ,-\min\left\{ 0,\underline{I}(\mu)\right\} \right\} \nonumber \\
 & =\frac{1}{2}\max\left\{ \overline{I}(\mu),-\underline{I}(\mu)\right\} .\label{eq:profiled.regret.dominance.pf.2}
\end{align}

Next, consider the MMR optimal rule $d^{*}$ that depends on the data
only through $(w^{*})^{\top}Y$ and that satisfies \eqref{eq:yata.property.1} and \eqref{eq:yata.property.2}.
By \eqref{eq:yata.property.1} and centrosymmetry of $\Theta$:
\begin{align}
\overline{R}_{r}(d^{*},\mathbf{0}) & =\sup_{\theta\in\Theta,m(\theta)=\mathbf{0}}U(\theta)\left[\mathbf{1}\left\{ U(\theta)\geq0\right\} -\mathbb{E}_{\mathbf{0}}\left[d^{*}\left((w^{*})^{\top}Y\right)\right]\right]\nonumber \\
 & =\sup_{\theta\in\Theta,m(\theta)=\mathbf{0}}U(\theta)\left[\mathbf{1}\left\{ U(\theta)\geq0\right\} -\frac{1}{2}\right]\nonumber \\
 & =\frac{\overline{I}(\mathbf{0})}{2}.\label{eq:profiled.regret.dominance.pf.1}
\end{align}

Thus, for all $\mu\in M,\mu\neq\mathbf{0}$, note we have
\begin{align*}
\overline{R}_{r}(d^*,\mu) & \leq\overline{R}_{r}(d^*,\mathbf{0})=\frac{\overline{I}(\mathbf{0})}{2}\leq\frac{1}{2}\max\left\{ \overline{I}(\mu),-\underline{I}(\mu)\right\} =\overline{R}_{r}(d_{\text{coin-flip}},\mu)\label{eq:profiled.regret.dominance.pf.3}
\end{align*}
where the first relation follows from the fact that $d^*$ is MMR
optimal with \eqref{eq:yata.property.2}  holding true, the second relation follows from
(\ref{eq:profiled.regret.dominance.pf.1}), the third relation follows
from assumption (with the inequality strict for some $\mu\in M$),
and the last relation follows from (\ref{eq:profiled.regret.dominance.pf.2}).
We then conclude that $\overline{R}_{r}(d^*,\mu)\leq\overline{R}_{r}(d_{\text{coin-flip}},\mu)$
for all $\mu\in M$, with the inequality strict for those $\mu\in M$
such that $\overline{I}(\mathbf{0})<\max\left\{ \overline{I}(\mu),-\underline{I}(\mu)\right\} $.
That is, $d^*$ dominates $d_{\text{coin-flip}}$ in terms of $\overline{R}_{r}(d,\mu)$.

\subsection{Proof of Proposition \ref{prop:aversion.fraction}}\label{sec:proof.aversion.fraction}
First, we show that, for any $d\in\mathcal{D}_{n,1}$, we have
\[
\sup_{\theta}R^{c}(d_{\text{linear}}^{*},\theta)\leq\sup_{\theta}R^{c}(d,\theta).
\]
We establish the above claim in two steps. 

{\scshape Step 1:} We show that for any $d\in\mathcal{D}_{n,1}$, 
\[
\sup_{\theta}R^{c}(d,\theta)\geq\overline{R}_{w^{*}}^{c}(d_{\text{linear}}^{*},0).
\]
Pick any $d\in\mathcal{D}_{n,1}$. Note
\begin{align*}
 & \sup_{\theta}R^{c}(d,\theta)=\sup_{\gamma\in\Gamma_{w^{*}}}\overline{R}_{w^{*}}^{c}(d,\gamma)\geq\overline{R}_{w^{*}}^{c}(d,0),
\end{align*}
where 
\[
\overline{R}_{w^{*}}^{c}(d,0)=\overline{R}_{w^{*}}(d,0)+\mathbb{E}_{0}\left[c(d(Y))\right],
\]
and $\mathbb{E}_{\gamma}\left[\cdotp\right]$ denotes expectation
with respect to $\left(w^{*}\right)^{\top}Y\sim N(\gamma,\left(w^{*}\right)^{\top}\Sigma w^{*})$.
As $d\in\mathcal{D}_{n,1}\subset\tilde{\mathcal{D}}_{n}$, we have $\mathbb{E}_{0}\left[d\left(\left(w^{*}\right)^{\top}Y\right)\right]=\frac{1}{2}$,
and 
\[
\overline{R}_{w^{*}}(d,0)=\overline{R}_{w^{*}}(d_{\text{linear}}^{*},0)=\frac{I(0)}{2}.
\]
Moreover, the structure of the constant cost function is such that
\[
c(d(Y))=c\mathbf{1}\left\{ \left(w^{*}\right)^{\top}Y\in(0,1)\right\} .
\]
By the ``least-randomizing'' property of $d_{\text{linear}}^{*}$
established in Theorem \ref{thm:main.regret.2}, 
\[
d_{\text{linear}}^{*}\left(\left(w^{*}\right)^{\top}Y\right)\in(0,1)\Rightarrow d\left(\left(w^{*}\right)^{\top}Y\right)\in(0,1).
\]
Conclude that 
\[
\mathbb{E}_{0}\left[c(d(Y))\right]\geq\mathbb{E}_{0}\left[c(d_{\text{linear}}^{*}(Y))\right].
\]
Therefore,
\[
\overline{R}_{w^{*}}^{c}(d,0)\geq\overline{R}_{w^{*}}^{c}(d_{\text{linear}}^{*},0),
\]
and as a result, 
\[
\sup_{\theta}R^{c}(d,\theta)\geq\overline{R}_{w^{*}}^{c}(d_{\text{linear}}^{*},0).
\]

{\scshape Step 2:} We show 
\[
\overline{R}_{w^{*}}^{c}(d_{\text{linear}}^{*},0)=\sup_{\theta}R^{c}(d_{\text{linear}}^{*},\theta).
\]
To see this, note
\begin{align*}
\sup_{\theta}R^{c}(d_{\text{linear}}^{*},\theta) & \leq\sup_{\theta}R(d_{\text{linear}}^{*},\theta)+\sup_{\theta}\mathbb{E}_{m(\theta)}\left[c(d_{\text{linear}}^{*}(Y))\right].
\end{align*}
By the fact that $d_{\text{linear}}^{*}$ is MMR optimal for the risk
without the cost function, we have $\sup_{\theta}R(d_{\text{linear}}^{*},\theta)=\frac{\overline{I}(0)}{2}$
. Also note that
\begin{align*}
\sup_{\theta\in\Theta}\mathbb{E}_{m(\theta)}\left[c(d_{\text{linear}}^{*}(Y))\right] & =\sup_{\gamma\in\Gamma_{w^{*}}}\mathbb{E}_{\gamma}\left[c(d_{\text{linear}}^{*}(Y))\right]=\mathbb{E}_{0}\left[c(d_{\text{linear}}^{*}(Y))\right],
\end{align*}
where the second equality follows by observing that $\left(w^{*}\right)^{\top}Y$
is normally distributed and that $d_{\text{linear}}^{*}(Y)\in(0,1)$
when $\left(w^{*}\right)^{\top}Y\in(-\rho^{*},\rho^{*})$. Thus, 
\begin{align*}
\sup_{\theta}R^{c}(d_{\text{linear}}^{*},\theta) & \leq\frac{\overline{I}(0)}{2}+\mathbb{E}_{0}\left[c(d_{\text{linear}}^{*}(Y))\right]=\overline{R}_{w^{*}}^{c}(d_{\text{linear}}^{*},0)
\end{align*}
But note by definition, 

\[
\sup_{\theta}R^{c}(d_{\text{linear}}^{*},\theta)\geq\overline{R}_{w^{*}}^{c}(d_{\text{linear}}^{*},0).
\]
Then, conclude that
\[
\sup_{\theta}R^{c}(d_{\text{linear}}^{*},\theta)=\overline{R}_{w^{*}}^{c}(d_{\text{linear}}^{*},0).
\]
Conclusions of steps 1 and 2 above imply 
\[
\sup_{\theta}R^{c}(d_{\text{linear}}^{*},\theta)\leq\sup_{\theta}R^{c}(d,\theta),\forall d\in\mathcal{D}_{n,1}.
\]
Next, we show that, for any $d\in\mathcal{D}_{n,2}$,
\[
\sup_{\theta}R^{c}(d_{\text{linear}}^{*},\theta)\leq\sup_{\theta}R^{c}(d,\theta).
\]
for all 
\[
c\leq\frac{\sup_{\mu\in M,\overline{I}(\mu)>0}\mathrm{g}_{w^{*}}(\mu)-\overline{I}(0)/2}{P_{0}\left\{ (-\rho^{*},\rho^{*})\right\}}.
\]
To see this, note we already calculated 
\[
\sup_{\theta}R^{c}(d_{\text{linear}}^{*},\theta)=\frac{\overline{I}(0)}{2}+\mathbb{E}_{0}\left[c(d_{\text{linear}}^{*}(Y))\right],
\]
where algebra shows 
\begin{align*}
\mathbb{E}_{0}\left[c(d_{\text{linear}}^{*}(Y))\right] & =c\mathbb{E}_{0}\left[\mathbf{1}\left\{ d_{\text{linear}}^{*}(Y)\in(0,1)\right\} \right]\\
 & =c\mathbb{E}_{0}\left[\mathbf{1}\left\{ \left(w^{*}\right)^{\top}Y\in(-\rho^{*},\rho^{*})\right\} \right]\\
 & =cP_{0}\left\{ (-\rho^{*},\rho^{*})\right\}.
\end{align*}
In contrast, note for any $d\in\mathcal{D}_{n,2}$, 
\begin{align*}
\sup_{\theta}R^{c}(d,\theta) & \geq\inf_{d\in\mathcal{D}_{n,2}}\sup_{\theta}R^{c}(d,\theta)=\sup_{\theta}R^{c}(d_{0}^{*},\theta),
\end{align*}
where $d_{0}^{*}=\mathbf{1}\left\{ \left(w^{*}\right)^{\top}Y\geq0\right\} $.
That is, $d_{0}^{*}$ is the best threshold rule among all linear
threshold rules that depends on data only via $\left(w^{*}\right)^{\top}Y$.
As $d^*_{0}$ is non-fractional, algebra shows
\begin{align*}
\sup_{\theta}R^{c}(d_{0}^{*},\theta) & =\sup_{\theta}R(d_{0}^{*},\theta)=\sup_{\mu\in M,\overline{I}(\mu)>0}\mathrm{g}_{w^{*}}(\mu),
\end{align*}
where $\mathrm{g}_{w}(\mu)=\overline{I}(\mu)\Phi\left(-\frac{w^{\top}\mu}{\sqrt{w^{\top}\Sigma w}}\right)$
is also defined in the proof for Theorem \ref{thm:main.regret.1}(iii). Therefore, as long
as 
\[
c\leq\frac{\sup_{\mu\in M,\overline{I}(\mu)>0}\mathrm{g}_{w^{*}}(\mu)-\overline{I}(0)/2}{P_{0}\left\{ (-\rho^{*},\rho^{*})\right\}},
\]
we have, for all $d\in\mathcal{D}_{n,2}$,
\[
\sup_{\theta}R^{c}(d_{\text{linear}}^{*},\theta)\leq\sup_{\theta}R^{c}(d,\theta).
\]
This completes the proof of Proposition \ref{prop:aversion.fraction}.

\subsection{Additional Results}
 
\subsubsection{Profiled Regret in the Running Example}\label{sec:profile.regret.computation}

In the running example, consider $w^*$-profiled regret, recalling that $w^*=(1,0)^{\top}$. By the definition in \eqref{eq:profile.regret}, the $w^*$-profiled regret function of a rule $d\in\mathcal{D}_2$ equals 
\begin{eqnarray*}
\overline{R}_{w^*}(d,\gamma) = \sup_{\theta \in \Theta \textrm{ s.t. } m_1(\theta) = \gamma } U(\theta)\left(\mathbf{1}\{U (\theta)\geq0\}-\mathbb{E}_{m(\theta)}[d(Y)]\right).
\end{eqnarray*}
Specifically, we look at three rules:
\begin{align*}
 d_{0}((w^*)^\top Y)&=d_{0}(Y_1)=\mathbf{1}\{Y_1\geq0\}\\
 d^*_\text{RT}((w^*)^\top Y)&=d^*_\text{RT}(Y_1)\\
 d^*_\text{linear}((w^*)^\top Y)&=d^*_\text{linear}(Y_1)
\end{align*}
defined in \eqref{eq:yata.rule}-\eqref{eq:d.erm}. 

As these rules depend on data only via  $(w^*)^\top Y$, for $d\in\{d_{0},d^*_\text{RT},d^*_{\text{linear}}\}$ we can write
\begin{align*}
\overline{R}_{w^*}(d,\gamma)&=\sup_{\theta \in \Theta \textrm{ s.t. } m_1(\theta) = \gamma } U(\theta)\left(\mathbf{1}\{U (\theta)\geq0\}-\mathbb{E}_{m_{1}(\theta)}[d(Y_1)]\right)\\
&=\sup_{U^{*}\in \left [ - \overline{k}_{w^*}(-\gamma), \:   \overline{k}_{w^*}(\gamma)  \right]}U^{*}\left(\mathbf{1}\{U^{*}\geq0\}-\mathbb{E}_{\gamma}[d(Y_1)]\right),
\end{align*}
where $\overline{k}_{w^*}(\gamma)$ is defined in \eqref{eq:isu}, and \eqref{eq:running.example.linear.program} solves for $\overline{k}_{w^*}(\gamma)=\gamma+k$, where $k=C\|x_0-x_1\|$. Therefore, we may further calculate 
\begin{align*}
&\overline{R}_{w^*}(d,\gamma)\\
=&\sup_{U^{*}\in \left[\gamma-k, \:   \gamma+k\right]}U^{*}\left(\mathbf{1}\{U^{*}\geq0\}-\mathbb{E}_{\gamma}[d(Y_1)]\right)\\
=&\max\left\{\sup_{U^{*}\in \left[\gamma-k, \:   \gamma+k\right],U^*\geq0}U^{*}\left(1-\mathbb{E}_{\gamma}[d(Y_1)]\right),\sup_{U^{*}\in \left[\gamma-k, \:   \gamma+k\right],U^*\leq0}-U^{*}\mathbb{E}_{\gamma}[d(Y_1)]\right\}\\
=&\begin{cases}
(-\gamma+k)\mathbb{E}_{\gamma}[d(Y_1)], & \text{if }\gamma<-k,\\
\max\left\{(\gamma+k)\left(1-\mathbb{E}_{\gamma}[d(Y_1)]\right),(-\gamma+k)\mathbb{E}_{\gamma}[d(Y_1)]\right\},  & \text{if }-k\leq \gamma\leq k,\\
(\gamma+k)(1-\mathbb{E}_{\gamma}[d(Y_1)]), & \text{if }\gamma>k.
\end{cases}
\end{align*}

As $Y_1\sim N(\gamma,\sigma_1)$, algebra shows $\mathbb{E}_{\gamma}[d_{0}(Y_1)]=\Phi\left(\frac{\gamma}{\sigma_1}\right)$, $\mathbb{E}_{\gamma}[d_{\text{RT}}^*(Y_1)]=\Phi\left(\sqrt{\frac{\pi}{2}}\frac{\gamma}{k}\right)$, and
\begin{align*}
\mathbb{E}_{\gamma}[d^*_{\text{linear}}(Y_1)]&=\Phi\left(\frac{\gamma-\rho^*}{\sigma_1}\right)+\frac{\sigma_1}{2\rho^*}\frac{1}{\sqrt{2\pi}}\left[e^{-\frac{1}{2}\left(\frac{\rho^*+\gamma}{\sigma_1}\right)^2}-e^{-\frac{1}{2}\left(\frac{\rho^*-\gamma}{\sigma_1}\right)^2}\right]\\
&+\frac{\gamma+\rho^*}{2\rho^*}\left[\Phi\left(\frac{\rho^*-\gamma}{\sigma_1}\right)-\Phi\left(\frac{-\rho^*-\gamma}{\sigma_1}\right)\right].    
\end{align*}
The $w^*$-profiled regret of the three rules can then be calculated easily following our characterizations above.

\subsubsection{Plot of Profiled Regret for the Plug-in Rule}\label{sec:plot.profiled.regret}

It would be interesting to compare $d^*_{\text{linear}}(Y_1)$ to a plug-in rule based on an estimated version of $I(\mu)$. In the example, a natural such estimator would be 
$[\widehat{\underline{I}},\widehat{\overline{I}}]$,  
where 
\[
\widehat{\overline{I}}=\min_{i=1,\ldots,n} \left\{ \hat{\mu}_i + C\left\Vert x_i-x_0\right\Vert  \right\}, ~~~
\widehat{\underline{I}}=\max_{i=1,\ldots,n} \left\{ \hat{\mu}_i - C\left\Vert x_i-x_0\right\Vert \right \},
\]
and $(\hat{\mu}_1,\ldots,\hat{\mu}_n)$ is a constrained (to $M$) maximum likelihood estimator. Based on the estimated identified set $[\widehat{\underline{I}},\widehat{\overline{I}}]$\footnote{However, whenever $d^*_{\text{linear}}$ randomizes, it is not necessarily the case that $[\widehat{\underline{I}},\widehat{\overline{I}}]$ contains zero. }, a natural plug-in rule is then
\begin{equation}\label{eq:d.plug.in}
d_{\text{plug-in}}(Y):=\begin{cases}
0, & \widehat{\overline{I}}<0,\\
\frac{\widehat{\overline{I}}}{\widehat{\overline{I}}-\widehat{\underline{I}}}, & \text{otherwise},\\
1, & \widehat{\underline{I}}>0.
\end{cases}
\end{equation}
We plot and compare the $w^*$-profiled regrets of  $d_{\text{plug-in}}$ and other rules, including $d^*_{\text{linear}}$ and $d^*_{\text{RT}}$. See Figure \ref{figure:profiled.regret.plug.in} for details. Note $d_{\text{plug-in}}$ is not MMR optimal: its  $w^*$-profiled regret at $\gamma=0$ is slightly larger than the MMR value of the problem. The $w^*$-profiled regret curve of $d_{\text{plug-in}}$ is bell-shaped and similar to that of $d^*_{\text{linear}}$.

\begin{figure}[http]
    \centering
    \includegraphics[scale=0.22]
{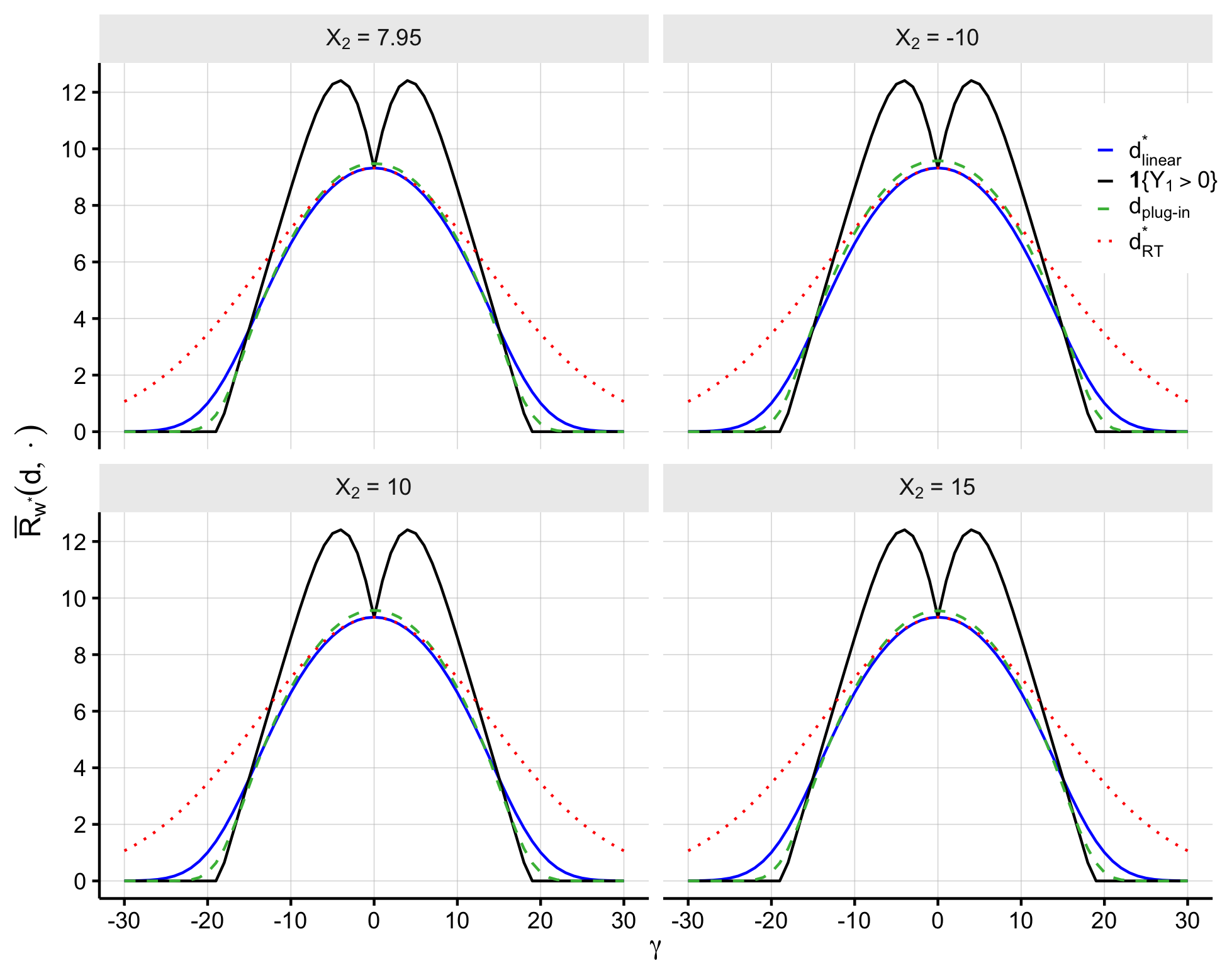}
    \caption{$w^*$-profiled regrets of $d_\text{plug-in}$ in \eqref{eq:d.plug.in} and other rules. The top left plot uses parameters of the running examples from Figure \ref{figure:NewRule.1}; the rest of the plots change the value of $x_2$ while keeping other parameters the same.}\label{figure:profiled.regret.plug.in}
\end{figure}

\subsubsection{Verification of Theorem \ref{thm:admin} for the Example in Section \ref{sec:iv.example}} \label{sec:iv.example.verification}
By Theorem \ref{thm:admin}, it suffices to show that the statistical model in \eqref{eq:IV.data} and the welfare contrast in \eqref{eq:IV.welfare} display nontrivial partial identification in the sense of Definition \ref{asm:1}. Note the image of $(m_1(\theta),m_2(\theta))^\top$ in this example is defined as
\[
M	=\left\{ \left(\mu_{1},\mu_{2}\right)^\top\in\mathbb{R}^{2}|m_1(\theta)=\mu_{1},m_2(\theta)=\mu_{2},\theta\in\Theta\right\}.
\]

The restrictions on $p(0)$ and $p(1)$ are: $p(1)\in[0,1]$, $p(0)\in[0,1]$, $p(1)\geq p(0)$ and $p(1)+\alpha\leq1$. Without any additional shape restrictions on $\text{MTE}(\cdot)$, any functions from $[0,1]$ to $[-1,1]$ are compatible with the model. Hence, we can deduce that 
\[
M	=\left\{ \left(\mu_{1},\mu_{2}\right)^\top\in\mathbb{R}^{2}|\mu_{1}\in[-\mu_{2},\mu_{2}],0\leq\mu_{2}\leq1-\alpha\right\} .
\]
The identified set of $U(\theta)$ is then
\[I(\mu_{1},\mu_{2})=\left\{ u\in\mathbb{R}|U(\theta)=u,m_1(\theta)=\mu_{1},m_2(\theta)=\mu_{2},\theta\in\Theta\right\},\text{for all } \left(\mu_{1},\mu_{2}\right)^\top\in M. \]
with extrema
\begin{align*}
\overline{I}(\mu_{1},\mu_{2}) & =\sup_{m_1(\theta)=\mu_{1},m_2(\theta)=\mu_{2},\theta\in\Theta}\left\{ \frac{m_1(\theta)}{\alpha+m_2(\theta)}-\frac{m_1(\theta)}{m_2(\theta)}+\frac{1}{\alpha+m_2(\theta)}\int_{p(1)}^{p(1)+\alpha}\text{MTE}(v)dv\right\} \\
\text{ } & =\frac{\mu_{1}}{\alpha+\mu_{2}}-\frac{\mu_{1}}{\mu_{2}}+\frac{1}{\alpha+\mu_{2}}\sup\left\{ \int_{p(1)}^{p(1)+\alpha}\text{MTE}(v)dv\right\} \\
 & =\frac{\mu_{1}}{\alpha+\mu_{2}}-\frac{\mu_{1}}{\mu_{2}}+\frac{\alpha}{\alpha+\mu_{2}}, 
\end{align*}
and similarly
$$\underline{I}(\mu_1,\mu_2) =\frac{\mu_{1}}{\alpha+\mu_{2}}-\frac{\mu_{1}}{\mu_{2}}-\frac{\alpha}{\alpha+\mu_{2}}.$$
Let $\mathcal{S}=\left\{ \left(\mu_{1},\mu_{2}\right)^\top\in M:\mu_1\in(-\mu_2,\mu_2),\mu_2\neq0\right\}$.  We can verify that for any $(\mu_1,\mu_2)^\top \in \mathcal{S}$, 
\[
\underline{I}(\mu_{1},\mu_{2})=\frac{\mu_{1}}{\alpha+\mu_{2}}-\frac{\mu_{1}}{\mu_{2}}-\frac{\alpha}{\alpha+\mu_{2}}<0<\frac{\mu_{1}}{\alpha+\mu_{2}}-\frac{\mu_{1}}{\mu_{2}}+\frac{\alpha}{\alpha+\mu_{2}}=\overline{I}(\mu_{1},\mu_{2}).
\]
Therefore, the statistical model in \eqref{eq:IV.data} and the welfare contrast in \eqref{eq:IV.welfare} exhibit nontrivial partial identification as defined in \ref{asm:1}, and Theorem \ref{thm:admin} applies to the example in Section \ref{sec:iv.example}.

\bibliographystyle{ecta}
\bibliography{bibliography}

\begin{thebibliography}{77}
\newcommand{\enquote}[1]{``#1''}
\expandafter\ifx\csname natexlab\endcsname\relax\def\natexlab#1{#1}\fi

\bibitem[\protect\citeauthoryear{Adjaho and Christensen}{Adjaho and Christensen}{2022}]{adjaho2022externally}
\textsc{Adjaho, C. and T.~Christensen} (2022): \enquote{Externally Valid Treatment Choice,} \emph{arXiv preprint arXiv:2205.05561}.

\bibitem[\protect\citeauthoryear{Anderson and Rubin}{Anderson and Rubin}{1949}]{anderson49}
\textsc{Anderson, T. and H.~Rubin} (1949): \enquote{Estimation of the Parameters of a Single Equation in a Complete System of Stochastic Equations,} \emph{The Annals of Mathematical Statistics}, 20, 46--63.

\bibitem[\protect\citeauthoryear{Andrews and Mikusheva}{Andrews and Mikusheva}{2022}]{andrews2022gmm}
\textsc{Andrews, I. and A.~Mikusheva} (2022): \enquote{GMM is Inadmissible Under Weak Identification,} \emph{arXiv preprint arXiv:2204.12462}.

\bibitem[\protect\citeauthoryear{Aradillas~Fern{\'a}ndez, Blanchet, Montiel~Olea, Qiu, Stoye, and Tan}{Aradillas~Fern{\'a}ndez et~al.}{2024}]{fernandez2024epsilon}
\textsc{Aradillas~Fern{\'a}ndez, A., J.~Blanchet, J.~L. Montiel~Olea, C.~Qiu, J.~Stoye, and L.~Tan} (2024): \enquote{$\epsilon$-Minimax Solutions of Statistical Decision Problems via the Hedge Algorithm,} .

\bibitem[\protect\citeauthoryear{Aradillas~Fern\'andez, Montiel~Olea, Qiu, Stoye, and Tinda}{Aradillas~Fern\'andez et~al.}{2024}]{twopointprior}
\textsc{Aradillas~Fern\'andez, A., J.~L. Montiel~Olea, C.~Qiu, J.~Stoye, and S.~Tinda} (2024): \enquote{Robust Bayes Treatment Choice with Partial Identification,} Discussion paper, Cornell University.

\bibitem[\protect\citeauthoryear{Athey and Wager}{Athey and Wager}{2021}]{AW20}
\textsc{Athey, S. and S.~Wager} (2021): \enquote{Efficient policy learning with observational data,} \emph{Econometrica}, 89, 133--161.

\bibitem[\protect\citeauthoryear{Ben-Michael, Greiner, Imai, and Jiang}{Ben-Michael et~al.}{2025}]{ben2021safe}
\textsc{Ben-Michael, E., D.~J. Greiner, K.~Imai, and Z.~Jiang} (2025): \enquote{Safe policy learning through extrapolation: Application to pre-trial risk assessment,} \emph{Journal of the American Statistical Association}, in press.

\bibitem[\protect\citeauthoryear{Ben-Michael, Imai, and Jiang}{Ben-Michael et~al.}{2024}]{ben2022policy}
\textsc{Ben-Michael, E., K.~Imai, and Z.~Jiang} (2024): \enquote{Policy learning with asymmetric utilities,} \emph{Journal of the American Statistical Association}, 119, 3045--3058.

\bibitem[\protect\citeauthoryear{Berger}{Berger}{1985}]{berger1985statistical}
\textsc{Berger, J.} (1985): \emph{Statistical decision theory and Bayesian analysis}, Springer.

\bibitem[\protect\citeauthoryear{Bhattacharya and Dupas}{Bhattacharya and Dupas}{2012}]{BhattacharyaDupas2012}
\textsc{Bhattacharya, D. and P.~Dupas} (2012): \enquote{Inferring welfare maximizing treatment assignment under budget constraints,} \emph{Journal of Econometrics}, 167, 168--196.

\bibitem[\protect\citeauthoryear{Canner}{Canner}{1970}]{canner}
\textsc{Canner, P.~L.} (1970): \enquote{Selecting One of Two Treatments when the Responses are Dichotomous,} \emph{Journal of the American Statistical Association}, 65, 293--306.

\bibitem[\protect\citeauthoryear{Casella and Berger}{Casella and Berger}{2002}]{casella2002statistical}
\textsc{Casella, G. and R.~L. Berger} (2002): \emph{Statistical inference}, Thomson Learning.

\bibitem[\protect\citeauthoryear{Chamberlain}{Chamberlain}{2000}]{chamberlain2000econometric}
\textsc{Chamberlain, G.} (2000): \enquote{Econometric applications of maxmin expected utility,} \emph{Journal of Applied Econometrics}, 15, 625--644.

\bibitem[\protect\citeauthoryear{Chamberlain}{Chamberlain}{2011}]{chamberlain2012}
---\hspace{-.1pt}---\hspace{-.1pt}--- (2011): \enquote{{Bayesian Aspects of Treatment Choice},} in \emph{{The Oxford Handbook of Bayesian Econometrics}}, Oxford University Press.

\bibitem[\protect\citeauthoryear{Chen and Guggenberger}{Chen and Guggenberger}{2024}]{ChenGug}
\textsc{Chen, H. and P.~Guggenberger} (2024): \enquote{A note on minimax regret rules with multiple treatments in finite samples,} Discussion paper, The Pennsylvania State University.

\bibitem[\protect\citeauthoryear{Christensen, Moon, and Schorfheide}{Christensen et~al.}{2022}]{christensen2022optimal}
\textsc{Christensen, T., H.~R. Moon, and F.~Schorfheide} (2022): \enquote{Optimal Discrete Decisions when Payoffs are Partially Identified,} \emph{arXiv preprint arXiv:2204.11748}.

\bibitem[\protect\citeauthoryear{Clarke}{Clarke}{1990}]{clarke1990optimization}
\textsc{Clarke, F.~H.} (1990): \emph{Optimization and nonsmooth analysis}, SIAM.

\bibitem[\protect\citeauthoryear{D'Adamo}{D'Adamo}{2021}]{d2021policy}
\textsc{D'Adamo, R.} (2021): \enquote{Orthogonal Policy Learning Under Ambiguity,} \emph{arXiv preprint arXiv:2111.10904}.

\bibitem[\protect\citeauthoryear{Daskalakis, Skoulakis, and Zampetakis}{Daskalakis et~al.}{2021}]{daskalakis2021complexity}
\textsc{Daskalakis, C., S.~Skoulakis, and M.~Zampetakis} (2021): \enquote{The complexity of constrained min-max optimization,} in \emph{Proceedings of the 53rd Annual ACM SIGACT Symposium on Theory of Computing}, 1466--1478.

\bibitem[\protect\citeauthoryear{Dehejia}{Dehejia}{2005}]{Dehejia2005}
\textsc{Dehejia} (2005): \enquote{Program evaluation as a decision problem,} \emph{Journal of Econometrics}, 125, 141--173.

\bibitem[\protect\citeauthoryear{Diegert, Masten, and Poirier}{Diegert et~al.}{2022}]{diegert2022assessing}
\textsc{Diegert, P., M.~A. Masten, and A.~Poirier} (2022): \enquote{Assessing omitted variable bias when the controls are endogenous,} \emph{arXiv preprint arXiv:2206.02303}.

\bibitem[\protect\citeauthoryear{Dvoretzky, Wald, and Wolfowitz}{Dvoretzky et~al.}{1951}]{dww}
\textsc{Dvoretzky, A., A.~Wald, and J.~Wolfowitz} (1951): \enquote{{Elimination of Randomization in Certain Statistical Decision Procedures and Zero-Sum Two-Person Games},} \emph{The Annals of Mathematical Statistics}, 22, 1 -- 21.

\bibitem[\protect\citeauthoryear{Ferguson}{Ferguson}{1967}]{Ferguson67}
\textsc{Ferguson, T.} (1967): \emph{Mathematical Statistics: A Decision Theoretic Approach}, vol.~7, Academic Press New York.

\bibitem[\protect\citeauthoryear{Giacomini and Kitagawa}{Giacomini and Kitagawa}{2021}]{GiacominiKitagawa}
\textsc{Giacomini, R. and T.~Kitagawa} (2021): \enquote{Robust Bayesian Inference for Set-Identified Models,} \emph{Econometrica}, 89, 1519--1556.

\bibitem[\protect\citeauthoryear{Guggenberger and Huang}{Guggenberger and Huang}{2025}]{guggenberger2025numerical}
\textsc{Guggenberger, P. and J.~Huang} (2025): \enquote{On the numerical approximation of minimax regret rules via fictitious play,} \emph{arXiv preprint arXiv:2503.10932}.

\bibitem[\protect\citeauthoryear{Hansen}{Hansen}{2022}]{Hansen}
\textsc{Hansen, B.~E.} (2022): \emph{Econometrics}, Princeton University Press.

\bibitem[\protect\citeauthoryear{Heckman and Vytlacil}{Heckman and Vytlacil}{1999}]{heckman1999local}
\textsc{Heckman, J.~J. and E.~J. Vytlacil} (1999): \enquote{Local instrumental variables and latent variable models for identifying and bounding treatment effects,} \emph{Proceedings of the National Academy of Sciences}, 96, 4730--4734.

\bibitem[\protect\citeauthoryear{Heckman and Vytlacil}{Heckman and Vytlacil}{2005}]{heckman2005structural}
---\hspace{-.1pt}---\hspace{-.1pt}--- (2005): \enquote{Structural equations, treatment effects, and econometric policy evaluation 1,} \emph{Econometrica}, 73, 669--738.

\bibitem[\protect\citeauthoryear{Hirano and Porter}{Hirano and Porter}{2009}]{HiranoPorter2009}
\textsc{Hirano, K. and J.~R. Porter} (2009): \enquote{Asymptotics for statistical treatment rules,} \emph{Econometrica}, 77, 1683--1701.

\bibitem[\protect\citeauthoryear{Hirano and Porter}{Hirano and Porter}{2020}]{HiranoPorter2020}
---\hspace{-.1pt}---\hspace{-.1pt}--- (2020): \enquote{Asymptotic analysis of statistical decision rules in econometrics,} in \emph{Handbook of Econometrics, Volume 7A}, ed. by S.~N. Durlauf, L.~P. Hansen, J.~J. Heckman, and R.~L. Matzkin, Elsevier, vol.~7 of \emph{Handbook of Econometrics}, 283--354.

\bibitem[\protect\citeauthoryear{Ida, Ishihara, Ito, Kido, Kitagawa, Sakaguchi, and Sasaki}{Ida et~al.}{2022}]{ida2022choosing}
\textsc{Ida, T., T.~Ishihara, K.~Ito, D.~Kido, T.~Kitagawa, S.~Sakaguchi, and S.~Sasaki} (2022): \enquote{Choosing Who Chooses: Selection-Driven Targeting in Energy Rebate Programs,} Tech. rep., National Bureau of Economic Research.

\bibitem[\protect\citeauthoryear{Imbens and Angrist}{Imbens and Angrist}{1994}]{imbens1994identification}
\textsc{Imbens, G.~W. and J.~D. Angrist} (1994): \enquote{Identification and Estimation of Local Average Treatment Effects,} \emph{Econometrica}, 467--475.

\bibitem[\protect\citeauthoryear{Ishihara and Kitagawa}{Ishihara and Kitagawa}{2021}]{ishihara2021}
\textsc{Ishihara, T. and T.~Kitagawa} (2021): \enquote{Evidence Aggregation for Treatment Choice,} ArXiv:2108.06473 [econ.EM], \url{https://doi.org/10.48550/arXiv.2108.06473}.

\bibitem[\protect\citeauthoryear{Kallus and Zhou}{Kallus and Zhou}{2018}]{kallus2018confounding}
\textsc{Kallus, N. and A.~Zhou} (2018): \enquote{Confounding-robust policy improvement,} \emph{Advances in neural information processing systems}, 31.

\bibitem[\protect\citeauthoryear{Karlin}{Karlin}{1957}]{karlin1957polya}
\textsc{Karlin, S.} (1957): \enquote{P{\'o}lya type distributions, II,} \emph{The Annals of Mathematical Statistics}, 28, 281--308.

\bibitem[\protect\citeauthoryear{Karlin and Rubin}{Karlin and Rubin}{1956}]{karlin1956theory}
\textsc{Karlin, S. and H.~Rubin} (1956): \enquote{The theory of decision procedures for distributions with monotone likelihood ratio,} \emph{The Annals of Mathematical Statistics}, 272--299.

\bibitem[\protect\citeauthoryear{Khan, Rath, and Sun}{Khan et~al.}{2006}]{purify}
\textsc{Khan, M.~A., K.~P. Rath, and Y.~Sun} (2006): \enquote{The Dvoretzky-Wald-Wolfowitz Theorem and Purification in Atomless Finite-action Games,} \emph{International Journal of Game Theory}, 34, 91--104.

\bibitem[\protect\citeauthoryear{Kido}{Kido}{2022}]{kido2022distributionally}
\textsc{Kido, D.} (2022): \enquote{Distributionally Robust Policy Learning with Wasserstein Distance,} \emph{arXiv preprint arXiv:2205.04637}.

\bibitem[\protect\citeauthoryear{Kido}{Kido}{2023}]{kido2023locally}
---\hspace{-.1pt}---\hspace{-.1pt}--- (2023): \enquote{Locally Asymptotically Minimax Statistical Treatment Rules Under Partial Identification,} \emph{arXiv preprint arXiv:2311.08958}.

\bibitem[\protect\citeauthoryear{Kitagawa, Lee, and Qiu}{Kitagawa et~al.}{2022}]{kitagawa2022treatment}
\textsc{Kitagawa, T., S.~Lee, and C.~Qiu} (2022): \enquote{Treatment Choice with Nonlinear Regret,} \emph{arXiv preprint arXiv:2205.08586}.

\bibitem[\protect\citeauthoryear{Kitagawa, Sakaguchi, and Tetenov}{Kitagawa et~al.}{2021}]{KST21}
\textsc{Kitagawa, T., S.~Sakaguchi, and A.~Tetenov} (2021): \enquote{Constrained classification and policy learning,} \emph{arXiv preprint}.

\bibitem[\protect\citeauthoryear{Kitagawa and Tetenov}{Kitagawa and Tetenov}{2018}]{kitagawa2018should}
\textsc{Kitagawa, T. and A.~Tetenov} (2018): \enquote{Who should be treated? {E}mpirical welfare maximization methods for treatment choice,} \emph{Econometrica}, 86, 591--616.

\bibitem[\protect\citeauthoryear{Kitagawa and Tetenov}{Kitagawa and Tetenov}{2021}]{KT19}
---\hspace{-.1pt}---\hspace{-.1pt}--- (2021): \enquote{Equality-Minded Treatment Choice,} \emph{Journal of Business Economics and Statistics}, 39, 561--574.

\bibitem[\protect\citeauthoryear{Kitagawa and Wang}{Kitagawa and Wang}{2020}]{KW20}
\textsc{Kitagawa, T. and G.~Wang} (2020): \enquote{Who should get vaccinated? Individualized allocation of vaccines over SIR network,} \emph{arXiv preprint}.

\bibitem[\protect\citeauthoryear{Lehmann and Casella}{Lehmann and Casella}{1998}]{lehmann2006theory}
\textsc{Lehmann, E.~L. and G.~Casella} (1998): \emph{Theory of point estimation}, Springer Science \& Business Media, second ed.

\bibitem[\protect\citeauthoryear{Lehmann and Romano}{Lehmann and Romano}{2005}]{lehmann05testing}
\textsc{Lehmann, E.~L. and J.~P. Romano} (2005): \emph{Testing statistical hypotheses}, vol.~3, Springer.

\bibitem[\protect\citeauthoryear{Lei, Sahoo, and Wager}{Lei et~al.}{2023}]{lei2023policy}
\textsc{Lei, L., R.~Sahoo, and S.~Wager} (2023): \enquote{Policy Learning under Biased Sample Selection,} \emph{arXiv preprint arXiv:2304.11735}.

\bibitem[\protect\citeauthoryear{Manski}{Manski}{2000}]{Manski2000}
\textsc{Manski, C.~F.} (2000): \enquote{Identification problems and decisions under ambiguity: empirical analysis of treatment response and normative analysis of treatment choice,} \emph{Journal of Econometrics}, 95, 415--442.

\bibitem[\protect\citeauthoryear{Manski}{Manski}{2004}]{manski2004statistical}
---\hspace{-.1pt}---\hspace{-.1pt}--- (2004): \enquote{Statistical treatment rules for heterogeneous populations,} \emph{Econometrica}, 72, 1221--1246.

\bibitem[\protect\citeauthoryear{Manski}{Manski}{2005}]{manski2005social}
---\hspace{-.1pt}---\hspace{-.1pt}--- (2005): \emph{Social choice with partial knowledge of treatment response}, Princeton University Press.

\bibitem[\protect\citeauthoryear{Manski}{Manski}{2007{\natexlab{a}}}]{manski2007identification}
---\hspace{-.1pt}---\hspace{-.1pt}--- (2007{\natexlab{a}}): \emph{Identification for prediction and decision}, Harvard University Press.

\bibitem[\protect\citeauthoryear{Manski}{Manski}{2007{\natexlab{b}}}]{Manski2007}
---\hspace{-.1pt}---\hspace{-.1pt}--- (2007{\natexlab{b}}): \enquote{Minimax-regret treatment choice with missing outcome data,} \emph{Journal of Econometrics}, 139, 105--115.

\bibitem[\protect\citeauthoryear{Manski}{Manski}{2020}]{manski2020towards}
---\hspace{-.1pt}---\hspace{-.1pt}--- (2020): \enquote{Towards credible patient-centered meta-analysis,} \emph{Epidemiology}.

\bibitem[\protect\citeauthoryear{Manski}{Manski}{2021}]{manski2021econometrics}
---\hspace{-.1pt}---\hspace{-.1pt}--- (2021): \enquote{Econometrics for decision making: Building foundations sketched by {Haavelmo} and {Wald},} \emph{Econometrica}, 89, 2827--2853.

\bibitem[\protect\citeauthoryear{Manski}{Manski}{2024}]{manski2022identification}
---\hspace{-.1pt}---\hspace{-.1pt}--- (2024): \enquote{Identification and Statistical Decision Theory,} \emph{Econometric Theory}, 1--17.

\bibitem[\protect\citeauthoryear{Manski and Tetenov}{Manski and Tetenov}{2007}]{manski2007admissible}
\textsc{Manski, C.~F. and A.~Tetenov} (2007): \enquote{Admissible treatment rules for a risk-averse planner with experimental data on an innovation,} \emph{Journal of Statistical Planning and Inference}, 137, 1998--2010.

\bibitem[\protect\citeauthoryear{Mattner}{Mattner}{1993}]{mattner1993some}
\textsc{Mattner, L.} (1993): \enquote{Some incomplete but boundedly complete location families,} \emph{The Annals of Statistics}, 2158--2162.

\bibitem[\protect\citeauthoryear{Mbakop and Tabord-Meehan}{Mbakop and Tabord-Meehan}{2021}]{MT17}
\textsc{Mbakop, E. and M.~Tabord-Meehan} (2021): \enquote{Model selection for treatment choice: Penalized welfare maximization,} \emph{Econometrica}, 89, 825--848.

\bibitem[\protect\citeauthoryear{Menzel}{Menzel}{2023}]{menzel2023transfer}
\textsc{Menzel, K.} (2023): \enquote{Transfer Estimates for Causal Effects across Heterogeneous Sites,} \emph{arXiv preprint arXiv:2305.01435}.

\bibitem[\protect\citeauthoryear{Mogstad, Santos, and Torgovitsky}{Mogstad et~al.}{2018}]{mogstad2018using}
\textsc{Mogstad, M., A.~Santos, and A.~Torgovitsky} (2018): \enquote{Using instrumental variables for inference about policy relevant treatment parameters,} \emph{Econometrica}, 86, 1589--1619.

\bibitem[\protect\citeauthoryear{Mogstad and Torgovitsky}{Mogstad and Torgovitsky}{2018}]{mogstad2018identification}
\textsc{Mogstad, M. and A.~Torgovitsky} (2018): \enquote{Identification and extrapolation of causal effects with instrumental variables,} \emph{Annual Review of Economics}, 10, 577--613.

\bibitem[\protect\citeauthoryear{M{\"u}ller}{M{\"u}ller}{2011}]{muller2011efficient}
\textsc{M{\"u}ller, U.~K.} (2011): \enquote{Efficient tests under a weak convergence assumption,} \emph{Econometrica}, 79, 395--435.

\bibitem[\protect\citeauthoryear{Parmigiani}{Parmigiani}{1992}]{ParmigianiTD}
\textsc{Parmigiani, G.} (1992): \enquote{Minimax, information and ultrapessimism,} \emph{Theory and Decision}, 33, 241--252.

\bibitem[\protect\citeauthoryear{Rockafellar}{Rockafellar}{1997}]{rockafellar1997convex}
\textsc{Rockafellar, R.~T.} (1997): \emph{Convex analysis}, vol.~11, Princeton university press.

\bibitem[\protect\citeauthoryear{Sadler}{Sadler}{2015}]{SadlerTD}
\textsc{Sadler, E.} (2015): \enquote{Minimax and the value of information,} \emph{Theory and Decision}, 78, 575--586.

\bibitem[\protect\citeauthoryear{Savage}{Savage}{1951}]{Savage51}
\textsc{Savage, L.} (1951): \enquote{The theory of statistical decision,} \emph{Journal of the American Statistical Association}, 46, 55--67.

\bibitem[\protect\citeauthoryear{Schlag}{Schlag}{2006}]{schlag2006eleven}
\textsc{Schlag, K.~H.} (2006): \enquote{{ELEVEN} - Tests needed for a Recommendation,} Tech. rep., European University Institute Working Paper, ECO No. 2006/2, \url{https://cadmus.eui.eu/bitstream/handle/1814/3937/ECO2006-2.pdf}.

\bibitem[\protect\citeauthoryear{Skorohod}{Skorohod}{2012}]{skorohod2012integration}
\textsc{Skorohod, A.~V.} (2012): \emph{Integration in Hilbert space}, vol.~79, Springer Science \& Business Media.

\bibitem[\protect\citeauthoryear{Stoye}{Stoye}{2007}]{Stoye07}
\textsc{Stoye, J.} (2007): \enquote{Minimax Regret Treatment Choice With Incomplete Data and many Treatments,} \emph{Econometric Theory}, 23, 190--199.

\bibitem[\protect\citeauthoryear{Stoye}{Stoye}{2009{\natexlab{a}}}]{stoye2009minimax}
---\hspace{-.1pt}---\hspace{-.1pt}--- (2009{\natexlab{a}}): \enquote{Minimax regret treatment choice with finite samples,} \emph{Journal of Econometrics}, 151, 70--81.

\bibitem[\protect\citeauthoryear{Stoye}{Stoye}{2009{\natexlab{b}}}]{stoye2009partial}
---\hspace{-.1pt}---\hspace{-.1pt}--- (2009{\natexlab{b}}): \enquote{Partial identification and robust treatment choice: an application to young offenders,} \emph{Journal of Statistical Theory and Practice}, 3, 239--254.

\bibitem[\protect\citeauthoryear{Stoye}{Stoye}{2012{\natexlab{a}}}]{stoye2012minimax}
---\hspace{-.1pt}---\hspace{-.1pt}--- (2012{\natexlab{a}}): \enquote{Minimax regret treatment choice with covariates or with limited validity of experiments,} \emph{Journal of Econometrics}, 166, 138--156.

\bibitem[\protect\citeauthoryear{Stoye}{Stoye}{2012{\natexlab{b}}}]{stoye2012new}
---\hspace{-.1pt}---\hspace{-.1pt}--- (2012{\natexlab{b}}): \enquote{New perspectives on statistical decisions under ambiguity,} \emph{Annu. Rev. Econ.}, 4, 257--282.

\bibitem[\protect\citeauthoryear{Tetenov}{Tetenov}{2012}]{tetenov2012statistical}
\textsc{Tetenov, A.} (2012): \enquote{Statistical treatment choice based on asymmetric minimax regret criteria,} \emph{Journal of Econometrics}, 166, 157--165.

\bibitem[\protect\citeauthoryear{Wald}{Wald}{1950}]{Wald50}
\textsc{Wald, A.} (1950): \emph{Statistical Decision Functions}, New York: Wiley.

\bibitem[\protect\citeauthoryear{Yata}{Yata}{2023}]{yata2021}
\textsc{Yata, K.} (2023): \enquote{Optimal Decision Rules Under Partial Identification,} ArXiv:2111.04926 [econ.EM], \url{https://doi.org/10.48550/arXiv.2111.04926}.

\bibitem[\protect\citeauthoryear{Yu and Kouvelis}{Yu and Kouvelis}{1995}]{yu1995min}
\textsc{Yu, G. and P.~Kouvelis} (1995): \enquote{On min-max optimization of a collection of classical discrete optimization problems,} in \emph{Minimax and applications}, Springer, 157--171.

\end{thebibliography}

\end{document}